\documentclass[11pt]{article}
\usepackage[margin=1in]{geometry}
\usepackage[utf8]{inputenc}
\usepackage[T1]{fontenc}
\usepackage{lmodern}
\usepackage{amssymb,amsmath,amsthm,amsfonts}
\usepackage{textgreek}
\usepackage{mathtools}
\usepackage{enumitem}
\usepackage[numbers,comma,sort&compress]{natbib}
\usepackage{authblk}
\usepackage{graphicx, caption, subcaption}
\usepackage{svg}
\usepackage{float}
\usepackage[ruled,vlined,linesnumbered]{algorithm2e}
\usepackage{algpseudocode}
\usepackage{qcircuit}
\usepackage{physics}
\usepackage{footnote}
\usepackage{xcolor}
\usepackage{mathrsfs}
\usepackage{bbm}
\usepackage{bm}
\usepackage{hhline}
\usepackage{cases}
\usepackage{tikz}
\usepackage{booktabs}
\usepackage{longtable}
\usepackage{tabularx}
\usepackage{aliascnt}
\usepackage[colorlinks=true,allcolors=blue]{hyperref}
\usepackage[capitalize]{cleveref}
\usepackage{etoc}
\usetikzlibrary{arrows.meta,positioning}

\renewcommand{\Re}{\operatorname{Re}}

\newcommand{\poly}{\operatorname{poly}}
\renewcommand{\Tr}{\operatorname{Tr}}

\newcommand{\I}{\mathrm{i}}

\newcommand{\mc}[1]{\mathcal{#1}}

\newcommand{\wt}[1]{\widetilde{#1}}

\renewcommand{\abs}[1]{\left\lvert#1\right\rvert}
\renewcommand{\norm}[1]{\left\lVert#1\right\rVert}

\newcommand{\ud}{\,\mathrm{d}}

\newcommand{\RR}{\mathbb{R}}
\newcommand{\CC}{\mathbb{C}}

\newcommand{\PP}{\mathbb{P}}

\newcommand{\cA}{\mathcal{A}}

\newcommand{\cS}{\mathcal{S}}

\newcommand{\cM}{\mathcal{M}}

\theoremstyle{plain}
\newtheorem{theorem}{Theorem}
\newaliascnt{proposition}{theorem}
\newtheorem{proposition}[proposition]{Proposition}
\aliascntresetthe{proposition}
\newaliascnt{lemma}{theorem}
\newtheorem{lemma}[lemma]{Lemma}
\aliascntresetthe{lemma}
\newaliascnt{corollary}{theorem}
\newtheorem{corollary}[corollary]{Corollary}
\aliascntresetthe{corollary}
\theoremstyle{definition}
\newtheorem{definition}{Definition}

\newtheorem{remark}{Remark}

\theoremstyle{plain}

\theoremstyle{plain}
\newaliascnt{lem}{thm}

\aliascntresetthe{lem}
\theoremstyle{plain}

\theoremstyle{plain}
\newtheorem*{lem*}{\protect\lemmaname}
\theoremstyle{plain}
\newaliascnt{prop}{thm}

\aliascntresetthe{prop}
\theoremstyle{plain}

\theoremstyle{definition}

\providecommand{\definitionname}{Definition}
\providecommand{\assumptionname}{Assumption}
\providecommand{\corollaryname}{Corollary}
\providecommand{\lemmaname}{Lemma}
\providecommand{\propositionname}{Proposition}
\providecommand{\remarkname}{Remark}
\providecommand{\examplename}{Example}
\providecommand{\theoremname}{Theorem}
\providecommand{\conjecturename}{Conjecture}
\providecommand{\assumptionname}{Assumption}

\Crefname{lem}{Lemma}{Lemmas}
\Crefname{prop}{Proposition}{Propositions}
\crefname{lemma}{lemma}{lemmas}
\Crefname{lemma}{Lemma}{Lemmas}
\crefname{proposition}{proposition}{propositions}
\Crefname{proposition}{Proposition}{Propositions}
\crefname{corollary}{corollary}{corollaries}
\Crefname{corollary}{Corollary}{Corollaries}

\def\>{\rangle}
\def\<{\langle}

\newcommand{\R}{\mathbb{R}}

\renewcommand{\d}{\mathrm{d}}

\SetKwInput{KwData}{Input}
\SetKwInput{KwResult}{Output}

\usepackage{braket}

\begin{document}

\title{Quantum score matching with applications to learning thermal states}
\author[1,*]{Yulong Dong}
\author[2,3,\textdagger]{Jiaqi Leng}
\affil[1]{Department of Electrical Engineering and Computer Science, University of Michigan, Ann Arbor, Michigan 48109, USA}
\affil[2]{Institute for Advanced Computing, Virginia Tech, Alexandria, Virginia 22305, USA}
\affil[3]{Department of Computer Science, Virginia Tech, Alexandria, Virginia 22305, USA}
\renewcommand{\Affilfont}{\small}
\date{\today}

\maketitle

\begingroup
\renewcommand{\thefootnote}{\fnsymbol{footnote}}
\footnotetext[1]{The authors are listed in alphabetical order. Email: \href{mailto:dongyl@umich.edu}{dongyl@umich.edu}.}
\footnotetext[2]{The authors are listed in alphabetical order. Email: \href{mailto:jiaqil@vt.edu}{jiaqil@vt.edu}.}
\endgroup

\begin{abstract}
Score matching has driven major advances in classical generative learning by enabling models to learn from data without evaluating intractable normalization constants, or partition functions. Yet, extending this principle to quantum learning requires rethinking its foundations, as quantum states are described by noncommuting density operators rather than scalar probabilities. The noncommutativity creates fundamental challenges not only in defining quantum scores, but also in developing a training framework with efficient circuit implementations and rigorous theoretical guarantees. In this work, we bridge this gap by establishing a general quantum score-matching framework with end-to-end theoretical guarantees. Applied to Gibbs-state learning, our approach avoids additional thermal-state preparation and achieves information-theoretically optimal sample complexity in the high-temperature regime for Hamiltonians with bounded locality and interaction degree. This positions score matching as a new route to state-of-the-art performance in learning quantum Gibbs states. 

Beyond these theoretical results, numerical simulations show that our method remains effective even when gradients are estimated inaccurately under limited measurement budgets. Experiments on IBM quantum hardware further demonstrate that quantum score matching is NISQ-friendly: without any error mitigation or correction, it reduces the relative Hamiltonian-parameter error from 64\% to approximately 10\%. Together, these results extend score matching into an experimentally realizable paradigm for quantum-state learning.
\end{abstract}

\etocdepthtag.toc{maintext}
\begingroup
\etocsettagdepth{maintext}{all}
\etocsettagdepth{appendices}{none}
\tableofcontents
\endgroup

\section{Introduction}
\label{sec:introduction}

Score matching has driven a wave of advances in classical generative machine learning, from learning unnormalized statistical models~\cite{hyvarinen2005estimation,hyvarinen2007extensions} to score-based diffusion models for image generation~\cite{sohldickstein2015deep,song2019generative,ho2020denoising,song2021score}. At the heart of its success is a change of perspective: a distribution can be learned by matching how its log-density varies, without evaluating the normalized density itself~\cite{lyu2009interpretation,yu2019generalized}. This variation is described by the score, the gradient of the log-density, which points locally towards more likely data configurations~\cite{hyvarinen2005estimation,song2019generative}. Taking this gradient eliminates the normalization constant, also known as the partition function, whose computation can be prohibitively expensive~\cite{hyvarinen2005estimation,lyu2009interpretation}. Crucially, the score-matching objective can be evaluated from samples without knowing the data distribution's score~\cite{hyvarinen2005estimation,vincent2011connection,song2020sliced}. Together, these properties allow models to be learned even when their normalized densities are intractable to evaluate. This prospect is particularly compelling for quantum learning: experiments provide copies of unknown states whose density matrices grow exponentially with system size~\cite{haah2017tomography}, while even specifying a thermal-state model can leave its partition function computationally inaccessible~\cite{bravyi2022complexity,alhambra2023thermal}. Can the same score-matching principle turn these quantum data into a learning procedure without requiring explicit access to the unknown state or evaluation of the model's partition function?

However, such a quantum extension of score matching is not straightforward. A classical probability density is a scalar function, and its logarithmic derivative can be handled with commutative calculus. A quantum state, by contrast, is a matrix-valued operator that need not commute with its variation~\cite{helstrom1976quantum,holevo1982probabilistic,jiang2014quantum}. Consequently, the choice of an operator-valued score matters~\cite{petz1996monotone,lesniewski1999monotone,petz2011introduction}, and the classical identity that turns score matching into an expectation over samples does not automatically extend to matrices. Defining a quantum score is therefore only the first step. We must also determine how to match it using copies of an unknown state and how to evaluate the parameter gradients needed for training on a quantum computer. These challenges require us to revisit the foundations of classical learning theory and quantum information theory. We need a brand-new framework that defines quantum scores within noncommutative algebra and makes score matching trainable and implementable on quantum computers.

Existing tools in quantum information theory address only parts of this problem, but do not yet provide a score-matching framework. On the metric side, quantum Fisher information uses logarithmic derivatives to quantify the distinguishability of nearby states, giving these derivatives a statistical interpretation~\cite{braunstein1994statistical,paris2009quantum,jiang2014quantum,liu2020quantum,zhou2019correspondence}. These tools also underpin precision bounds and optimal protocols in quantum sensing~\cite{eldredge2018optimal,zhou2018heisenberg,zhou2021asymptotic}. Noncommutative differential calculus supplies ways to differentiate functions of density operators~\cite{carlen2017gradient,carlenmaas2019}. Yet these results do not show how matching operator-valued quantum scores can both identify the target state and be evaluated using only copies of it. On the algorithmic side, quantum Boltzmann machines have been developed for learning based on quantum relative entropy~\cite{amin2018quantum,kieferova2017tomography,wiebe2019generative,wilde2025boltzmann}. However, training requires copies of the thermal state of the current model, so its efficiency also depends on the cost of thermal-state preparation. Other lines of work borrow classical learning ideas to learn quantum states or ensembles, either by converting quantum information into measurement probabilities for classical processing or by constructing diffusion-inspired variational architectures~\cite{coyle2020born,liu2025measurement,kwun2025mixed}. Nonetheless, these routes do not establish score matching directly at the level of density operators and the noncommutative operator algebra intrinsic to quantum theory.

Quantum state learning is the analogue of learning an underlying classical data distribution: an unknown quantum state is characterized using copies of that state~\cite{haah2017tomography,rouze2024efficient}. Specifically, the Gibbs state is $\sigma=e^{-\beta H(\theta^*)}/\Tr(e^{-\beta H(\theta^*)})$, where $H(\theta^*)$ is the Hamiltonian whose parameters $\theta^*$ are the learning target. Recent developments have established sample-efficient learning of local Hamiltonians from Gibbs states~\cite{anshu2021sample} and polynomial-time algorithms at any fixed temperature for Hamiltonians with bounded locality and interaction degree~\cite{bakshi2023learning,narayanan2025improved}. Among them, high-temperature algorithms achieve information-theoretically optimal sample complexity~\cite{haah2024learning}, while local measurement protocols for lattice Hamiltonians attain near-optimal dependence on system size and accuracy beyond the high-temperature regime~\cite{chen2025learning}. At a known temperature, recovering these Hamiltonian parameters from copies of the Gibbs state is a well-studied Hamiltonian-learning problem~\cite{bairey2019learning,zhao2023maximum,gu2024practical}. Beyond learning from thermal states, access to real-time dynamics and quantum control enables Heisenberg-limited Hamiltonian learning~\cite{huang2023heisenberg,dutkiewicz2024control,bakshi2024structure,hu2025ansatz,ma2024compressed}. Yet, these approaches do not provide a score-matching-based method for learning quantum states.

In this paper, we introduce a quantum score-matching framework and demonstrate its application to Gibbs-state learning. The framework is built on two pillars: (i) a noncommutative extension of classical score matching that connects score representation to global identifiability of the target state within the model class, and (ii) a circuit-efficient representation of the gradients needed to train the model using copies of an unknown quantum state. These pillars extend the mathematical and statistical theory of score matching to noncommutative algebra and quantum-computable procedures. Applying our framework to Gibbs-state learning, we show that the resulting procedure is information-theoretically optimal in sample complexity for learning high-temperature Gibbs states. We therefore provide a score-matching approach that matches information-theoretic lower bounds and the state-of-the-art performance of existing Gibbs-state learning algorithms in the high-temperature regime~\cite{haah2024learning}. Remarkably, our method is not only theoretically clean but also NISQ-ready. We deploy our method to learn a four-qubit transverse-field Ising model (TFIM) on IBM quantum hardware and demonstrate that the learned Hamiltonian parameters converge to about 10\% relative error with only 32 measurement shots per circuit, without error mitigation or error correction. We also perform numerical simulations to show the feasibility of learning larger systems under limited resource budgets.

\subsection{Summary of main theoretical results}
\label{subsec:summary-main-results}

We summarize our main theoretical results here. Let $\sigma \succ 0$ be the unknown data state and $\rho(\theta)\succ0$ a model family containing $\sigma$. On the matrix manifold with a Harmitian frame $\{A_a\}_{a \in \mc{A}}$, the derivative alongside the $a$-th direction is defined as $\partial_aO = - \I [A_a, O]$. Furthermore, we define the \emph{symmetric} score $S_a(\theta)$ by $\partial_a \rho(\theta) = (\rho(\theta) S_a(\theta) + S_a(\theta) \rho(\theta))/2$. This is the symmetric logarithmic derivative along the derivation, rather than along a model parameter. More generally, we can define a $\lambda$-score whereas the symmetric score is a special case with $\lambda = 1/2$. More details are provided in \cref{append:lambda_score}. All theory below can be parallelized to the $\lambda$-score setting. For simplicity, we focus on the symmetric score in the main text.

Our first result states that the squared score discrepancy can be minimized using just the copy of the unknown state, without requiring access to the target density operator or its functions. This parallelizes the classical score-matching identity in Ref.~\cite{hyvarinen2005estimation}, which turns the score-matching objective into a sampleable form.

\begin{theorem}[Quantum score matching, informal]
\label{thm:intro-framework}
Let $S_a(\theta), S_a^*$ be the scores of the model state and the target state, respectively. The squared score discrepancy
\begin{equation}\label{eq:intro-score-discrepancy}
J_Q(\theta)=\frac12\sum_{a\in\cA}\operatorname{Tr}\bigl(\sigma(S_a(\theta)-S_a^*)^2\bigr)
\end{equation}
admits the following representation
\begin{equation}\label{eq:intro-hyvarinen}
J_Q(\theta)=\operatorname{Tr}\bigl(\sigma\mathcal L(\theta)\bigr)+\text{constant},\qquad \mathcal L(\theta)=\sum_{a\in\cA}\left(\frac12S_a^2(\theta)+\partial_a S_a(\theta)\right).
\end{equation}
Therefore, $J_Q(\theta)$ and $\Tr(\sigma\mathcal L(\theta))$ have the same global minima, and the latter can be evaluated using only copies of $\sigma$ through Monte Carlo.
\end{theorem}

The discrepancy metric is equal to zero when the model density operator matches the target. The converse asks whether every global minimizer represents the target state, which is referred to as \emph{global identifiability}. This property requires a stronger condition on the derivation frame. Intuitively, when the derivation frame is sufficiently rich so that there is no nontrivial ambiguity in the matched scores, the target is globally identifiable. The formal result is given in \Cref{lem:global-identifiability-sld}. Specifically, for Hamiltonians with bounded locality and interaction degree, at high temperature, the single-qubit frame $\{X_x,Z_x\}_{x=1}^n$ suffices to guarantee global identifiability (\cref{thm:uniform-local-pauli-high-temperature}).

The formalism above lays the foundation for the quantum score-matching framework. Next, we need to train a model to match the score. Suppose the model is a Gibbs state $e^{-\beta H(\theta)}/\Tr(e^{-\beta H(\theta)})$ with a known inverse temperature $\beta$ and a parametrized Hamiltonian $H(\theta)=\sum_{j=1}^m\theta_jP_j$, where $P_j$ are distinct nonidentity Pauli operators. The gradient of the score-matching objective can be expressed as $\partial_{\theta_j}J_Q(\theta)=\operatorname{Tr}\bigl(\sigma\mathcal G_j(\theta)\bigr)$, where $\mathcal G_j(\theta)$ is an observable that can be estimated using copies of $\sigma$. The formal result is given in \cref{thm:parameter-shift-gradient-cost}.

\begin{theorem}[Circuit-efficient gradient measurement, informal]
\label{thm:intro-gradients}
Suppose $m=\operatorname{poly}(n)$, $|\theta_j|\leq 1$, and the derivation frame contains $\operatorname{poly}(n)$ Pauli operators. Assume query access to $e^{-\I tH(\theta)}$, with evolution time counted as a resource. For any $\beta>0$ and $0<\varepsilon,\delta<1$, a randomized circuit protocol returns an estimate $\widehat g$ satisfying $\mathbb P(\|\widehat g-\nabla_\theta J_Q(\theta)\|_\infty\leq\varepsilon)\geq1-\delta$. Its total copy cost and maximal Hamiltonian-evolution time per measurement circuit are
\begin{equation}\label{eq:intro-gradient-resources}
N_{\mathrm{copy}}=\mathcal{O}\!\left(\frac{\operatorname{poly}(n,\beta)}{\varepsilon^2}\log\frac{2m}{\delta}\right),\qquad T_{\max}=\mathcal{O}\!\left(\beta\log\frac{\operatorname{poly}(n,\beta)}{\varepsilon}\right).
\end{equation}
\end{theorem}

The result in the preceding theorem can be further improved in the high-temperature regime for Hamiltonians of bounded locality and interaction degree. In particular, the quasi-locality of the gradient observables allows multiple gradient coordinates to be estimated simultaneously from a single copy of $\sigma$. The analysis of the objective landscape further guarantees the convergence of projected gradient descent. We apply our score-matching framework to obtain the following end-to-end guarantee for learning high-temperature Gibbs states, whose copy complexity matches the minimax lower bound up to constants.

\begin{theorem}[Sample-optimal high-temperature learning, informal]
\label{thm:intro-high-temperature}
Suppose $H^*=\sum_{j=1}^m\theta_j^*P_j$ is an $n$-qubit Hamiltonian where each Pauli term has non-identity weight at most $k$ and intersects at most $D$ other terms. Suppose $|\theta_j^*|\le1$. For fixed $k,D$, there is a threshold $\beta_0>0$ independent of $n$ such that, for $0<\beta<\beta_0$, quantum score matching with the single-qubit frame $\{X_x,Z_x\}_{x=1}^n$ returns $\widehat\theta$ satisfying
\begin{equation}\label{eq:intro-learning-guarantee}
\mathbb P\left(\|\widehat\theta-\theta^*\|_\infty\leq\varepsilon\right)\geq1-\delta
\end{equation}
using
\begin{equation}\label{eq:intro-copy-complexity}
N_{\rm copy}=\mathcal{O}\left(\frac{1}{\beta^2\varepsilon^2}\log\frac{2m}{\delta}\right)
\end{equation}
copies of $\sigma$, for any $0<\varepsilon, \delta<1$.
\end{theorem}

The guarantee combines identifiability from a local frame, controlled optimization dynamics, and shared-copy estimation of quasi-local gradient observables. The formal result is \cref{thm:end-to-end-qsm-upper-bound}. In the parameter regime of \cref{thm:minimax-sample-lower-bound}, the copy complexity matches the minimax lower bound up to constants. In particular, the logarithmic dependence on $m$ is achieved by the measurement protocol, rather than by estimating each coordinate on a separate batch of copies.

\paragraph{Concurrent work.}
While finalizing this manuscript, we became aware of an independent concurrent work~\cite{shukla2026operator}, submitted to arXiv on September 21, 2026 (two days before ours). They propose an operator score matching formulation corresponding to the $\lambda=1$ case of our framework (see~\cref{append:lambda_score}) and study its application to Hamiltonian learning primarily through numerical experiments across a range of physical models. Our work develops a general quantum score-matching framework with systematic and rigorous analyses, alongside demonstrations through numerical simulations and quantum experiments on IBM hardware. Both frameworks may also be applied beyond the high-temperature regime, and both papers report numerical success on particular low-temperature instances. It is worth noting that such instance-specific success should be distinguished from a general efficiency guarantee. In particular, information-theoretic lower bounds established in prior work~\cite{haah2024learning} and in ours~\cref{thm:minimax-sample-lower-bound} show that the required copy complexity can grow exponentially with inverse temperature.

\section{Quantum score matching}

\subsection{Symmetric quantum score}

As in the classical setting, we will define the score function $S(\theta)$ of a density operator $\rho(\theta)$ as its logarithmic derivative.
To generalize this notion to density operators, we first define the derivation on matrix algebras. 

Let $\cM_n$ denote the algebra consisting of all $2^n$-by-$2^n$ complex-valued matrices; without ambiguity, we may simply write $\cM$.
We denote $\|X\|$ as the spectral norm of $X$.
For a vector $v$, $\|v\|_p$ denotes its $\ell^p$-norm.

\begin{definition}[Derivation]
    Given a discrete index set $\cA$ and a set of Hermitian matrices $\{A_a\}_{a\in \cA}$, 
    we define the derivation $\partial_a\colon \cM \to \cM$ for each $a \in \cA$:
    \begin{align}
        \partial_a(\cdot) := -i [A_a, \cdot].
    \end{align}
    The set $\{A_a\}_{a\in \cA}$ is called the \textit{frame of derivations}, or simply a \textit{frame}.
\end{definition}
\begin{lemma}
    The derivation $\partial_a$ satisfies the Leibniz rule and is star-preserving. Namely, for any $X, Y \in \cM$, 
    \begin{align}
        \partial_a(XY) = \partial_a(X) Y + X \partial_a(Y),\qquad \partial_a(X^\dagger) = (\partial_a(X))^\dagger.
    \end{align}
\end{lemma}
\begin{proof}
    Straightforward calculation. The second identity immediately implies that $(\partial_a(X))^\dagger = \partial_a(X)$ as long as $X^\dagger = X$.
\end{proof}

The classical score a probability density $q(x;\theta)$ is $s(x;\theta) := \nabla_x \log(q(x;\theta))$. By the chain rule, we have 
\begin{align}
    \nabla_x q(x;\theta) = q(x;\theta) s(q(x;\theta)).
\end{align}
Therefore, we may define a \textit{quantum score} function $S$ of a density operator $\rho$ such that $\partial_a \rho = \rho * S$,
where $*$ represents certain multiplicative operations, for example, left or right multiplication. 
Due to the non-commutativity of matrix multiplication, different definitions of multiplication lead to different quantum scores. 
In this work, we focus on the \textit{symmetric quantum score} (also known as \textit{symmetric logarithmic derivative}, or \textit{SLD}), where the ``multiplicative operation'' is given by an equal combination of left and right multiplication.

\begin{definition}[Symmetric quantum score]
    Let $\rho\succ0$ and let $\{A_a\}_{a\in\cA}$ be a derivation frame.
    The symmetric quantum score of a density operator $\rho$ with respect to $a \in \cA$, denoted by $S_a$ is the unique Hermitian matrix that solves the following equation,
    \begin{align}\label{eqn:sld-definition}
        \partial_a \rho = \frac{1}{2}\left(\rho S_a + S_a \rho\right).
    \end{align}
    We write $\mathscr{I}_{\rho}(X) := (\rho X + X\rho)/2$, so $S_a := \mathscr{I}^{-1}_{\rho}(\partial_a \rho)$.
\end{definition}
We use the $\rho$-weighted inner product
$\langle X,Y\rangle_\rho:=\Tr(X^\dagger\mathscr I_\rho(Y))$ and its norm
$\|X\|_\rho^2:=\langle X,X\rangle_\rho$. For Hermitian $X,Y$, this gives
$\langle X,Y\rangle_\rho=\frac12\Tr(\rho\{X,Y\})$ and
$\|X\|_\rho^2=\Tr(\rho X^2)$.
The uniqueness and Hermiticity of the symmetric quantum score is standard in quantum information theory. 
This definition can be generalized to allow an unequal combination of left and right multiplications, giving rise to the definition of $\lambda$-quantum score. We establish several properties of the $\lambda$-score, see~\cref{append:lambda_score}.

The symmetric quantum score admits an elegant closed-form expression.
We denote the left and right multiplication by $\rho \succ 0$ as $\mathbf{L}_\rho$, $\mathbf{R}_\rho$, and the 
modular operator is defined as
\begin{align*}
    \mathbf{\Delta}_\rho(X) := \mathbf{L}_\rho \mathbf{R}^{-1}_\rho (X) = \rho X \rho^{-1}.
\end{align*}

\begin{lemma}
    There is a unique Hermitian matrix $S_a$ that solves~\cref{eqn:sld-definition}.
    Moreover, it admits a closed-form expression, 
    \begin{align}
        S_a = 2\I \tanh \circ \log\left(\mathbf{\Delta}^{1/2}_\rho\right)(A_a).
    \end{align}
\end{lemma}
\begin{proof}
    Note that the equation $\{X, \rho\}=0$ has a unique solution $X=0$ when $\rho \succ 0$. This implies that~\cref{eqn:sld-definition} has a unique solution. 
    \cref{eqn:sld-definition} is a continuous-time Lyapunov equation, and its solution is given by
    \begin{align}
        S_a = \I \int^\infty_0 e^{-\rho t/2} \left(\rho A_a - A_a \rho\right) e^{-\rho t/2}\,\ud t.
    \end{align}
    Since $\mathbf{L}_\rho$ and $\mathbf{R}_\rho$ commute, we have 
    \begin{align*}
        S_a &= \I \int^\infty_0 e^{-(\mathbf{L}_\rho+\mathbf{R}_\rho)t/2}(\mathbf{L}_\rho - \mathbf{R}_\rho) = 2i \frac{\mathbf{L}_\rho - \mathbf{R}_\rho}{\mathbf{L}_\rho + \mathbf{R}_\rho}(A_a),\\
        &= 2\I \frac{\mathbf{\Delta}_\rho - I}{\mathbf{\Delta}_\rho + I}(A_a) = 2\I \tanh \circ \log\left(\mathbf{\Delta}^{1/2}_\rho\right)(A_a),
    \end{align*}
    where in the second step we use $\int^\infty_0 e^{-(\lambda+\mu)t/2} = \frac{2}{\lambda+\mu}$ for any $\lambda, \mu > 0$. The last step follows from $\frac{x-1}{x+1} = \tanh(\frac{1}{2}\log x)$.
\end{proof}

For a non-degenerate probability density $q(x) \propto e^{-h(x)}$, we have the identity $s(x) = - \nabla_x h(x)$. However, the symmetric quantum score does not satisfy this identity. In other words, if $\rho \propto e^{-H}$ for some Hamiltonian $H$, generally $S_a \neq - \partial_a H = i[A_a,H]$.
In~\cref{append:matrix-logarithm-recover}, we show that this identity can be recovered using the Bogoliubo--Kubo--Mori score. However, the resulting objective does not directly admit sample-based evaluation and training.
Fortunately, the symmetric quantum score turns out to give a quantum operator learning framework that is computationally tractable.

\subsection{Score-based Hamiltonian learning and the Hyv\"arinen formula}
\label{sec:score-based-learning}

Now, we introduce the standard quantum Gibbs state learning problem. We adopt the standard Pauli Hamiltonian model. 
Suppose that we have access to a data ensemble $\sigma$ is the Gibbs state of a $k$-local Hamiltonian $H^*$ at an inverse temperature $\beta > 0$,
\begin{align}
    \sigma = \frac{1}{\Tr(e^{-\beta H^*})} \exp(-\beta H^*), \qquad \text{with}\quad H^* = \sum^m_{j=1} \theta^*_j P_j\,.
\end{align}
where each of the $m$ distinct Pauli product terms act non-trivially on at most $k$ qubits and the coefficients satisfy $\theta^*_j \in [-1,1]$ for all $j\in [m]$.
Assuming that the inverse temperature $\beta$ and the Pauli terms $\{P_j\}^m_{j=1}$ are known, our goal is to learn the parameter $\theta^*$.

Inspired by the classical score matching framework, we turn the quantum Gibbs state learning problem into the following optimization problem, where the objective is to minimize the Fisher divergence between the observed $\sigma$ and a parametrized family $\rho(\theta)$.
We set $\Theta=[-1,1]^m$ and consider the thermal state
\begin{equation}
    \rho(\theta) := \frac{1}{\Tr(e^{-\beta H(\theta)})} \exp(-\beta H(\theta)), \qquad \text{with}\quad H(\theta) = \sum^m_{j=1} \theta_j P_j\,.
\end{equation}

The Fisher divergence between $\rho(\theta)$ and $\sigma$ is given by 
\begin{align}\label{eqn:score-matching-objective}
    J_Q(\theta) := \mathcal{F}(\rho(\theta)\|\sigma) = \frac{1}{2}\sum_{a\in \cA}\left\|S_{a}(\theta) - S^*_a\right\|^2_{\sigma}\, ,
\end{align}
where $S_a(\theta) := \mathscr{I}^{-1}_{\rho(\theta)}(\partial_a \rho(\theta))$ and $S^*_a := \mathscr{I}^{-1}_{\sigma}(\partial_a \sigma)$ are the symmetric score of $\rho(\theta)$ and $\sigma$ (with respect to $a \in \cA$), respectively.

The function $J_Q(\theta)$ will be referred as the score-matching objective, as $\theta = \theta^*$ is a global minimizer of $J_Q(\theta)$. However, it may not be the case that \textit{any} global minimizer of $J_Q(\theta)$ can recover the target Gibbs state $\sigma$.
We say the target parameter $\theta^*$ is \textit{globally identifiable} if every global minimum of the objective $J_Q(\theta)$ correspond to $\sigma=\rho(\theta^*)$. In the following lemma, we characterize the global identifiability in quantum score matching by a commutant condition of the derivation frame. 
This lemma is a direct corollary of~\cref{thm:global-identifiability}; more details are available in~\cref{append:lambda_score}.

Given a set of matrices $\cS = \{M_j\}_{j\in J}$ where $J$ is a finite index set, we denote $\cS'$ as the joint commutant of all $M_j \in \cS$, i.e., 
\begin{equation}
    \cS' = \bigcap_{j\in J}\left(\ker \mathbf{ad}_{M_j}\right),\qquad \text{where}\quad \mathbf{ad}_{M_j}(\cdot) := [M_j, \cdot]
\end{equation} 

\begin{lemma}[Global identifiability]
    \label{lem:global-identifiability-sld}
    Let $\{A_a\}_{a\in \cA}$ be a derivation frame and $\sigma \propto e^{-\beta H^*}$ be a quantum Gibbs state.
    Using this derivation frame and target state, we define $J_Q(\theta)$ as in~\cref{eqn:score-matching-objective}, let $\hat{\theta}$ be a global minimizer of $J_Q(\theta)$. 
    Define 
    \begin{align}\label{eqn:K-sigma-sld}
        K^\sigma_{a} := \operatorname{sech}\circ \log\!\left(\mathbf{\Delta}^{1/2}_\sigma\right)(A_a).
    \end{align}
    If $\{K^\sigma_{a}:a\in\cA\}'=\CC I$, we have $\rho(\hat{\theta}) = \sigma$.
\end{lemma}

The definition of $J_Q(\theta)$ contains the unknown target scores $S_a^*$, which can not be efficiently computed without preparing the Gibbs state $\rho(\theta)$. 
A major advantage of classical score matching is that the objective $J_Q(\theta)$ can be rewritten as an expectation over the data ensemble $\sigma$ itself via integration by parts, a result first proposed in Ref.~\cite{hyvarinen2005estimation}.
We show that this structure generalizes to our quantum score matching framework, and it enables efficient training of the score matching objective without preparing the Gibbs state $\rho(\theta)$ nor the partition function $\Tr(\rho(\theta))$.

\begin{lemma}[Quantum Hyv\"arinen formula]\label{lem:quantum-hyvarinen-sld}
    The score-matching objective can be expressed as 
    \begin{align}\label{eqn:qsm-loss-observable}
    J_Q(\theta) = \Tr\!\left(\sigma\mathcal L(\theta)\right)+C,\qquad \mathcal L(\theta) := \frac{1}{2}\sum_{a\in \cA} \left(S^2_a(\theta) + 2\partial_a S_a(\theta)\right)
    \end{align}
    and the constant $C := \frac12\sum_a \|S^*_a\|^2_\sigma$ is independent of $\theta$.
\end{lemma}

The proof of~\Cref{lem:quantum-hyvarinen-sld} is given in~\Cref{append:lambda_score}; in fact, we prove a stronger result (see~\Cref{lem:quantum-hyvarinen}), where a similar Hyv\"arinen formula holds for all $\lambda$-score (and the symmetric quantum score is a special case with $\lambda=1/2$).

\subsection{Gradient estimation and training}

Classical score matching often relies on gradient-based methods to find the desired parameters.
In quantum score matching, similarly, we prove that $\nabla J_Q(\theta^*)=0$, meaning that the target parameter $\theta^*$ is always a first-order stationary point of the objective $J_Q(\theta)$, see~\Cref{lem:first-order-stationary}.
This motivates us to consider gradient-based methods in the quantum score matching framework.

It remains to estimate the gradient of the objective $J_Q(\theta)$.
We denote $T_{a,j}(\theta) := \partial_{\theta_j} S_a(\theta)$, and the derivatives of $J_Q(\theta)$ can be written as
\begin{equation}\label{eqn:exact-gradient-observable-final}
    \partial_{\theta_j}J_Q(\theta)=\Tr\left(\sigma\mathcal G_j(\theta)\right),\qquad \mathcal G_j(\theta):=\sum_a\left(\frac12\{S_a(\theta),T_{a,j}(\theta)\}-\I[A_a,T_{a,j}(\theta)]\right).
\end{equation}
The observable $S_a(\theta)$ admits the following random-time representation (we write $H_\theta := H(\theta)$):
\begin{align}\label{eqn:S_a_Fourier}
    S_a(\theta) = -\beta \int^\infty_0 q_\beta(t) \frac{1}{2t}\int^t_{-t} \tau^{H_\theta}_u(\partial_a(H_\theta)) \ud u \ud t,\qquad \tau^{H}_u(\cdot) := e^{iu H} \cdot  e^{-iu H},
\end{align}
where 
\begin{align}
    q_\beta(t) = \frac{4t}{\beta^2\sinh(\pi t/\beta)}, \qquad t>0,
    \label{eqn:appendix-random-time-density}
\end{align}
is a probability density on $(0,\infty)$. 
\Cref{eqn:S_a_Fourier} allows us to compute the observable $S_a(\theta)$ as a probabilistic expectation over a random variable $\xi \sim \mathrm{Uniform}[-t, t]$ conditioned on $t \sim q_\beta(t)$, in other words,
\begin{align}\label{eqn:probablistic_S_a}
    S_a(\theta) = -\beta \mathbb{E}_{\xi}\left(\tau^{H_\theta}_{\xi}(\partial_a H_\theta)\right).
\end{align}
Therefore, by differentiating~\cref{eqn:S_a_Fourier} using the chain rule, we have $T_{a,j}(\theta) = T^{(1)}_{a,j}(\theta) + T^{(2)}_{a,j}(\theta)$, where
\begin{align}
    T^{(1)}_{a,j}(\theta) &:= -\beta \mathbb{E}_{\xi}\left(\tau^{H_\theta}_\xi(\partial_a(\partial_{\theta_j} H_{\theta}))\right) = -\beta \mathbb{E}_{\xi}\left(\tau^{H_\theta}_\xi(\partial_a(P_j))\right),\label{eqn:probablistic_T_a_1}\\
    T^{(2)}_{a,j}(\theta) &:= \I \beta \mathbb{E}_\xi\left(\xi \int^1_0 \tau^{H_{\theta}}_{(1-s)\xi}\left([\tau^{H_{\theta}}_{s\xi}(\partial_a H_\theta), P_j]\right)\,\ud s\right).\label{eqn:probablistic_T_a_2}
\end{align}
The second identity involves commutators with Pauli matrices $P_j$. Because $P_j^2=I$, we further exploit the exact parameter-shift representation
\begin{equation}
    \I[P_j,O]=\mathcal R_{j,+1}(O)-\mathcal R_{j,-1}(O)
    =2\mathbb E_\eta(\eta\mathcal R_{j,\eta}(O)),
    \qquad
    \mathcal R_{j,\eta}(O):=e^{\eta\I\pi P_j/4}Oe^{-\eta\I\pi P_j/4}.
    \label{eqn:gradient-sampling-parameter-shift}
\end{equation}
This further simplifies the second term to
\begin{align}
    T^{(2)}_{a,j}(\theta) = -2\beta \mathbb{E}_{\xi,\eta}\left(\xi\eta \int^1_0 \tau^{H_{\theta}}_{(1-s)\xi}\left(\mathcal{R}_{j,\eta}(\tau^{H_{\theta}}_{s\xi}(\partial_a H_\theta))\right)\,\ud s\right),\label{eqn:probablistic_T_a_2_final}
\end{align}
where $\eta=\pm 1$ is a fair coin independent of $\xi$. The integral over $s$ is equivalently an expectation over an independent $s\sim\mathrm{Uniform}[0,1]$.

While these identities are complicated, we note that $\partial_a H_\theta$, $\partial_a P_j$ can be represented as weighted sums of unitary matrices (e.g., Paulis). Therefore, by plugging~\cref{eqn:probablistic_S_a}, \cref{eqn:probablistic_T_a_1}, and \cref{eqn:probablistic_T_a_2_final} into~\cref{eqn:exact-gradient-observable-final}, we can represent the gradient observable $\mathcal{G}_j(\theta)$ as an expectation over multiple branches of unitary operators.
This gives an algorithm to estimate the gradient of $J_Q(\theta)$ using time evolution of $H_\theta$, Pauli insertions and rotations, and copies of $\sigma$.

\begin{theorem}[Informal version of~\Cref{thm:parameter-shift-gradient-cost}]
    \label{thm:informal-gradient-estimation-cost}
    Fix $\theta$, $\beta>0$, a coordinate $j$, $0<\delta<1$, and $\varepsilon>0$.
    Assume access to the time-evolution oracle $e^{\I tH_\theta}$, and there are at most $\poly(n)$ number of generators in the derivation frame $\cA$.
    With success probability at least $1-\delta$, we can estimate $\partial_{\theta_j}J_Q(\theta)$ to additive error at most $\varepsilon$ using $\poly(n,\beta,1/\varepsilon)$ copies of the thermal state $\sigma$, and the maximal evolution time of $H_\theta$ is $\widetilde{\mathcal{O}}(\beta)$.
\end{theorem}

The formal statement of~\Cref{thm:informal-gradient-estimation-cost}, including the algorithm and detailed proof, is available in~\cref{append:gradient-estimation-proof}.
It is worth noting that the algorithm discussed here is very generic and is likely far from optimal, especially for structured thermal states or in the high-temperature regime. 
In~\cref{subsec:shared-copy-gradient-estimation}, we give a shared-copy gradient estimation algorithm for learning high-temperature Gibbs states. 
In this regime, it is sufficient to use the single-qubit frame $\cA=\{X_x,Z_x:1\le x\le n\}$, and the sample complexity of gradient estimation reduces to $\widetilde{\mathcal{O}}\left(\frac{\beta^2}{\varepsilon^2}\log\frac{m}{\delta}\right)$. This leads to a learning algorithm that matches the Hamiltonian learning sample lower bound up to a constant.

\section{Efficient Hamiltonian learning at high temperature}
\label{sec:high-temperature-learning}

Based on the quantum score-matching framework, we now investigate the task of learning quantum Gibbs states in the high-temperature regime, where $\beta$ is below a constant threshold that is independent of the system size.
We retain the model $H(\theta)=\sum_{j=1}^m\theta_jP_j$, $\Theta=[-1,1]^m$, and $\sigma=\rho(\theta^*)$ from~\cref{sec:score-based-learning}. Throughout this section, the known $P_j$ are distinct non-identity Pauli strings of weight at most $k$, and each term overlaps at most $D$ other terms. Excluding the identity fixes the unobservable additive constant in the Hamiltonian.
We use the single-qubit frame $\{A_a\}_{a\in\cA}=\{X_x,Z_x\}_{x=1}^n$, with index set $\cA=[n]\times\{X,Z\}$ and $A_{(x,p)}=p_x$. The symbol $m$ always counts Hamiltonian terms; $N_{\rm copy}$ counts copies of $\sigma$.
For learning a Hamiltonian with $m$ local terms to coefficient error $\|\widehat\theta-\theta^*\|_\infty\leq\varepsilon$, Ref.~\cite{haah2024learning} proved a worst-case lower bound of $\Omega\left(\frac{e^{2\beta}}{\beta^2\varepsilon^2}\log(m)\right)$ copies of the Gibbs state at constant failure probability.
In \cref{sec:information-lower-bounds}, we refine the lower-bound analysis through a general bound on the statistical distinguishability of perturbations confined to thermally suppressed subspaces, rather than relying on a specific hard-instance construction, and consequently show that exponential low-temperature hardness already arises in noninteracting one-local Pauli models.

In this section, we show that in the high-temperature regime, quantum score matching gives a near-optimal Hamiltonian learning algorithm that saturates the known sample complexity lower bound.

\subsection{Local Pauli frames suffice in high temperature}
\label{sec:high-temp-identifiability}

In~\Cref{lem:global-identifiability-sld}, we give a global identifiability test by computing the joint commutant of the operators $\{K^\sigma_a:a\in \cA\}$.
However, this commutant test is highly non-trivial as the operators $K^\sigma_a$ are very complicated.
In general, we may need an informationally complete frame to guarantee the global identifiability of the score-matching objective, which involves exponentially many ($4^n-1$) generators and thus is computationally intractable for learning many-body Hamiltonians.

In the high-temperature limit $\beta\to0^+$, the Gibbs state approaches $2^{-n}I$, so $K^\sigma_a\to A_a$.
The limiting frame satisfies $\{A_a:a\in\cA\}'=\CC I$. The following theorem shows that this commutant condition persists at sufficiently small positive $\beta$, uniformly in the system size. The restriction $\beta>0$ is essential for learning Hamiltonian coefficients, since at $\beta=0$ all parameters give the same Gibbs state.

\begin{theorem}
    \label{thm:uniform-local-pauli-high-temperature}
    Suppose $H^*=\sum_j \theta^*_j P_j$ is a $k$-local Hamiltonian with each $|\theta^*_j|\leq1$ and every local Pauli term intersects at most $D$ other local terms.
    There exists a threshold $\beta_0 = \Omega(1/\poly(k,D))$ such that for any $0 < \beta < \beta_0$, any global minimizer $\widehat\theta$ of $J_Q(\theta)$ using the single-qubit derivation frame $\{A_a\}_{a\in\cA}=\{X_x,Z_x\}_{x=1}^n$ corresponds to the target Gibbs state, i.e., $\rho(\widehat\theta) = \sigma$.
\end{theorem}

The proof can be found in~\cref{append:local-pauli-frame-high-temperature}. In fact, we prove a stronger version that applies to any $\lambda$-score (note that the symmetric score corresponds to $\lambda = 1/2$).
The proof bounds the perturbation of the commutant test from its $\beta=0$ limit. A Lieb--Robinson bound controls this perturbation using only $\beta$, $k$, and $D$, yielding a threshold $\beta_0 = \Omega(1/\poly(k,D))$ independent of the system size.

\subsection{Sample-efficient gradient estimation}
\label{subsec:shared-copy-gradient-estimation}

Recall that the gradient of the score-matching objective can be written as $\partial_{\theta_j}J_Q(\theta) = \Tr(\sigma \mathcal{G}_j(\theta))$, where
\begin{align}
    \mathcal G_j(\theta) := \sum_{a\in\cA}\left(\frac12\{S_a(\theta),T_{a,j}(\theta)\} -\I[A_a,T_{a,j}(\theta)]\right),
    \label{eqn:qsm-gradient-observable}
\end{align}
where $T_{a,j}(\theta)=\partial_{\theta_j}S_a(\theta)$, as in~\cref{eqn:exact-gradient-observable-final}.

Gradient-based training on $J_Q(\theta)$ requires all $m$ coordinates of the gradient. The gradient estimation subroutine in~\Cref{thm:parameter-shift-gradient-cost} only evaluates one coordinate of $\nabla_\theta J_Q(\theta)$ at a time, so each gradient estimation incurs a linear dependence on the number $m$ of parameters.
This strategy is likely not optimal, as most Hamiltonian terms are spatially separated.
In the high-temperature regime, we find that the gradient observables $\mathcal{G}_j$ are \textit{quasi-local}, meaning that they are concentrated near the support $X_j:=\operatorname{supp}(P_j)$ of the $j$th Hamiltonian term, with a tail that remains summable after accounting for the number of overlapping neighborhoods.

The key idea is to measure spatially separated contributions to different gradient coordinates on the same copy of $\sigma$. We decompose each gradient observable into local terms supported on progressively larger neighborhoods, and group terms with disjoint supports for simultaneous measurement. At sufficiently high temperature, the decay of these contributions offsets the measurement cost of larger neighborhoods. This yields a shared-copy estimator whose copy cost depends only logarithmically on $m$, rather than linearly.

\begin{theorem}[Shared-copy gradient estimation]
\label{thm:shared-copy-gradient-estimation}
    There is a $\beta_0 = \Omega(1/\poly(D,k))$ such that for every $0<\beta\leq\beta_0$, $0<\tau<1$, $0<\delta<1$, and a fixed $\theta\in\Theta$, there is an algorithm that returns an estimate $\widehat g$ satisfying
    \begin{align}
        \PP\left(\|\widehat g-\nabla_\theta J_Q(\theta)\|_\infty \leq\tau\right) \geq 1 - \delta
    \end{align}
    using
    \begin{align}
        N_{\rm copy}= \mathcal{O}\left(
            \left(1+\frac{\beta^2}{\tau^2}\right)
            \log\frac{2m}{\delta}
        \right)
        \label{eqn:shared-copy-gradient-cost}
    \end{align}
    copies of $\sigma$, using joint measurements on disjoint local neighborhoods as specified in~\cref{alg:shared-copy-gradient}. Here, the pre-constant scales polynomially in $k$ and $D$. This is a copy-complexity bound; it does not bound the cost of synthesizing the neighborhood measurements.
\end{theorem}

The algorithm and proof of~\Cref{thm:shared-copy-gradient-estimation} are available in~\Cref{append:shared_copy_grad_estimate}.

\subsection{Learning dynamics and sample complexity}
\label{subsec:general-lambda-score-derivatives}

The learning dynamics of quantum score matching are given by the gradient flow on $J_Q(\theta)$,
\begin{align}\label{eq:learning-dynamics}
    \frac{\mathrm d\theta(t)}{\mathrm d t} = - \nabla_\theta J_Q(\theta),\qquad \theta(0) = \theta_0.
\end{align}
We denote the Hessian of the score-matching objective $J_Q(\theta)$ as $\mathsf{H}(\theta)$. 
In the high-temperature regime ($0 < \beta \ll 1$), we find that $\mathsf{H}(\theta) = \beta^2 \Gamma + \mathcal{O}(\beta^3)$ (see~\Cref{lem:finite-high-temperature-hessian}), where the leading-order term $\Gamma = (\Gamma_{i,j})^m_{i,j=1}$ admits an analytical form,
\begin{align}
    \Gamma_{i,j} = 2^{-n} \sum_{a\in \cA}\Tr\left([P_i, A_a]^\dagger[P_j, A_a]\right).
\end{align}
Here, $\{A_a\}_a$ are the derivation frame generators and $P_i = \partial_{\theta_i}H(\theta)$ forms the linear basis of the parametrized Hamiltonian.
Since $\Gamma$ is the Gram matrix of the operator tuples $\{v_i := \frac{1}{\sqrt{2^{-n}}}([P_i,A_a])_a\}$, it is positive semidefinite. Consequently, the objective $J_Q(\theta)$ is always a convex function for sufficiently high temperature.
In what follows, we use the single-qubit derivation frame $A_a = \{X_x, Z_x: 1\le x \le n\}$, thus $\Gamma$ is diagonal.
At a site carrying an $X$ or $Z$ factor, exactly one of the two local frame elements anticommutes with $P_j$; at a site carrying a $Y$ factor, both do. Therefore
\begin{align}
    \Gamma_{j,j} = 4\bigl(n_X(P_j)+n_Z(P_j)+2n_Y(P_j)\bigr).
    \label{eqn:score-gram-diagonal}
\end{align}
Here $n_X(P_j)$, $n_Y(P_j)$, and $n_Z(P_j)$ count the corresponding single-qubit factors. As every $P_j$ is non-identity and has weight at most $k$, then $4 \le \Gamma_{j,j} \le 8k$, and the condition number $\kappa(\Gamma) \le 2k$.

The above analysis shows that, in the high-temperature regime, the score-matching objective $J_Q(\theta)$ has a benign optimization landscape. 
Moreover, the asymptotic Hessian $\mathsf{H}(\theta) \approx \beta^2 \Gamma$ can also serve as a good pre-conditioner that accelerates the convergence of gradient-based methods.
However, the sample-based gradient estimation subroutine in~\Cref{thm:shared-copy-gradient-estimation} may still introduce stochastic errors in each step, which could lead to slow or no convergence in gradient-based training.

In the following theorem, we show that gradient methods are robust to the noise introduced by gradient estimation. 
Specifically, given an estimate $\widehat g_t$ of the gradient at the $t$-th update $\theta^{(t)}$, we use the update rule,
\begin{align}
    \theta^{(t+1)} = \Pi_\Theta\left(\theta^{(t)}-\beta^{-2}\Gamma^{-1}\widehat g_t\right),
    \label{eqn:stochastic-high-temperature-update}
\end{align}
where $\Pi_\Theta$ is coordinatewise projection onto $[-1,1]^m$. Here, $\beta^{-2}\Gamma^{-1}$ is a pre-conditioner applied to the gradient estimate.
We prove that, with only $\log(1/\varepsilon)$ iterations, the update can achieve an $\varepsilon$-approximate of the global minimizer $\theta^*$.

\begin{theorem}
\label{thm:end-to-end-qsm-upper-bound}
    There is a $\beta_0 = \Omega(1/\poly(D,k))$ such that for every $0<\beta\leq\beta_0$, $0<\varepsilon<1$, and $0<\delta<1$, the iteration
    \cref{eqn:stochastic-high-temperature-update}, initialized at
    $\theta^{(0)}=0$ and using the tolerance and failure-probability schedules below, returns $\widehat\theta=\theta^{(T)}$ after $T=\lceil\log(1/\varepsilon)/\log(4/3)\rceil$ rounds,
    satisfying
    \begin{align}
        \PP\left(
            \|\widehat\theta-\theta^*\|_\infty
            \leq\varepsilon
        \right)
        \geq1-\delta
    \end{align}
    using
    \begin{align}
        N_{\rm copy} = \mathcal{O}\left(\frac{1}{\beta^2\varepsilon^2}\log\frac{2m}{\delta}\right)
        \label{eqn:end-to-end-qsm-copy-bound}
    \end{align}
    copies of $\sigma$. 
\end{theorem}

The proof of the theorem can be found in~\Cref{append:high-temp-sample-complexity-proof}. For bounded-weight Pauli Hamiltonians with bounded term-overlap degree, $m\geq2$, $0<\varepsilon<1/4$, and $0<\delta\leq1/4$, this upper bound matches the lower bound in \cref{thm:minimax-sample-lower-bound}, up to constants.

\section{Quantum score matching with finite measurement budgets}

A moderate measurement budget per iteration is sufficient for finite-shot QSM to converge to a small neighborhood of the target at small and moderate inverse temperatures. The intuition is similar to stochastic gradient descent: individual gradient estimates need not be accurate if their errors average out and the accumulated updates retain a systematic descent direction. We demonstrate this behavior in an eight-qubit inhomogeneous transverse-field Ising model (TFIM) and then examine why the same budget becomes less effective as the temperature is lowered. The simulations separate the accuracy of an individual gradient estimate from the accuracy ultimately achieved by the learning dynamics.

We consider the eight-qubit TFIM Hamiltonian
\begin{equation}\label{eq:numerical-inhomogeneous-tfim}
H(\theta)=-\sum_{i=1}^{7}J_i Z_iZ_{i+1}-\sum_{i=1}^{8}h_i^xX_i,
\end{equation}
with target coefficients $J_i^*=1$ and $(h_i^x)^*=1.5$. All $15$ coefficients are learned independently using the local frame $\{X_x,Z_x\}_{x=1}^8$. At each iteration, we estimate the complete gradient using the randomized Hadamard-test construction in \cref{alg:randomized-gradient-coordinate}, without coherent amplitude amplification, and allocate a total budget of $N$ measurement shots across the gradient coordinates according to their estimator ranges. We use preconditioned gradient descent and compare a far initialization ($J_i^{(0)}=0.5$ and $(h_i^x)^{(0)}=1$) with random local initializations at a relative distance of $0.05$ from $\theta^*$. The complete simulation protocol is given in \cref{app:finite-shot-numerical-methods}. We measure learning accuracy by
\begin{equation}\label{eq:numerical-relative-learning-error}
e_{\mathrm{rel}}(\theta):=\frac{\|\theta-\theta^*\|_2}{\|\theta^*\|_2}.
\end{equation}

\Cref{fig:finite-shot-learning-dynamics} shows that the resulting stochastic gradients support end-to-end learning with a moderate value of $N$. For $\beta\leq0.6$, $N=10^5$ is sufficient to reduce the mean parameter error below $10\%$ from either initialization. Increasing the budget to $N=10^6$ extends this behavior to approximately $\beta=1$. Within this regime, the final error decreases approximately as $N^{-1/2}$. At larger $\beta$, the same budget reaches a progressively higher error floor. The similar temperature dependence from the two initializations shows that this deterioration is not solely caused by the difficulty of reaching the target region from a distant point. The initial rise from local initialization is a finite-shot effect rather than repulsion by the population landscape. Close to $\theta^*$, finite-shot errors can initially displace the parameters from the target. This transient becomes more pronounced as $\beta$ increases because both the absolute and relative gradient estimation errors grow, as shown in panels~(a) and~(b) of \cref{fig:finite-shot-gradient-diagnostics}.

\begin{figure}[ht!]
    \centering
    \includegraphics[width=0.96\linewidth]{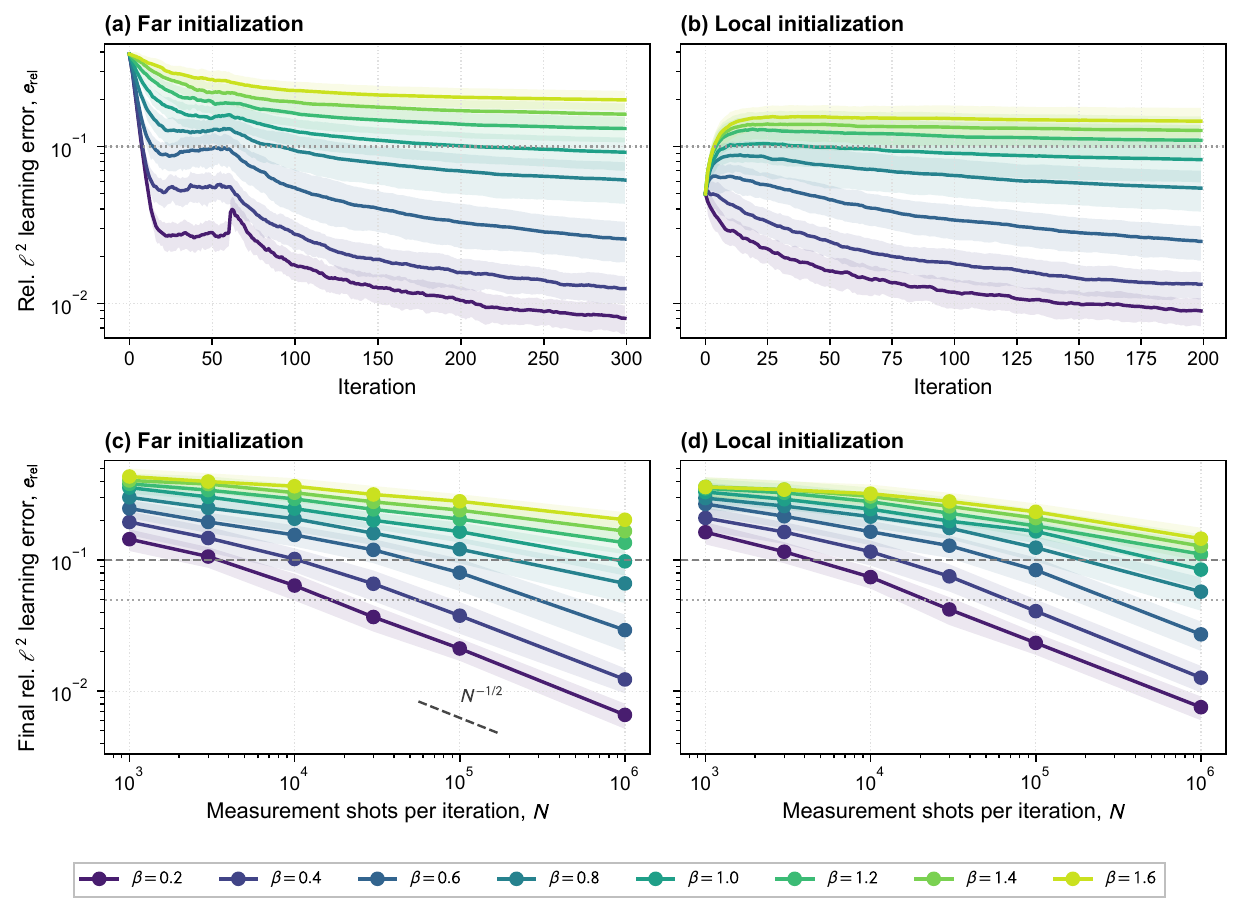}
    \caption{Finite-shot QSM learning for the inhomogeneous TFIM in \cref{eq:numerical-inhomogeneous-tfim}. The top row shows the mean relative parameter error at $N=10^6$ from far and local initializations. The bottom row shows the error of the Polyak--Ruppert averaged output as a function of the total measurement budget per iteration. Solid curves and shaded regions show the mean and one standard deviation over $100$ independent trajectories. The short dashed segment indicates the $N^{-1/2}$ sampling rate. Optimization details are given in \cref{app:finite-shot-numerical-methods}.}
    \label{fig:finite-shot-learning-dynamics}
\end{figure}

The population landscapes in \cref{fig:finite-shot-learning-landscape} locate the mechanism behind this loss of noise tolerance. The target remains a clear minimum in each displayed slice. As $\beta$ increases, however, the low-loss valley becomes increasingly elongated and tilted in the displayed coordinate plane. This increasing anisotropy reflects a more ill-conditioned Hessian. The endpoint distribution consequently changes from overlapping the target neighborhood at high temperature to remaining outside it at low temperature. The noisy dynamics can therefore continue to reduce the loss while making little progress toward the target parameters.

\begin{figure}[ht!]
    \centering
    \includegraphics[width=0.96\linewidth]{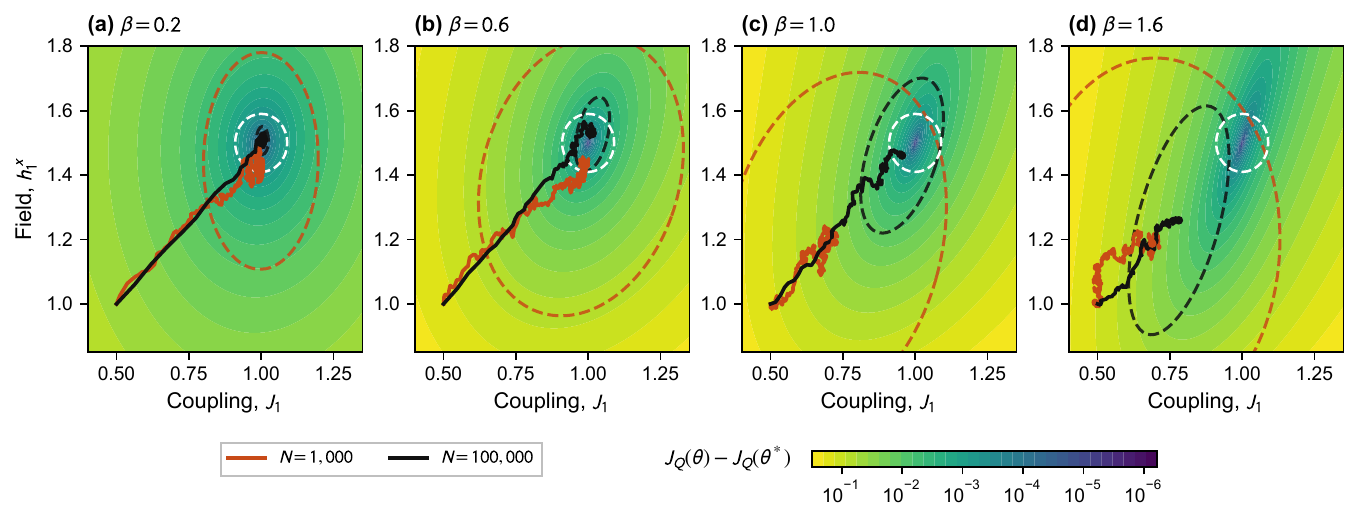}
    \caption{Temperature-dependent population landscape and finite-shot learning dynamics from far initialization. The color map shows $J_Q(\theta)-J_Q(\theta^*)$ in the $(J_1,h_1^x)$ plane with the remaining coefficients fixed at their target values. Orange and black curves are the mean trajectories for $N=10^3$ and $10^5$, respectively. Dashed ellipses show one standard deviation of the trajectory endpoints over $100$ independent trajectories, and the white circle has radius $0.05\|(J_1^*,(h_1^x)^*)\|_2$. The trajectories are projected onto the displayed coordinate plane.}
    \label{fig:finite-shot-learning-landscape}
\end{figure}

The comparison in \cref{fig:finite-shot-gradient-diagnostics} further shows that successful learning does not require an accurate gradient estimate at every iteration. Even at $\beta=0.2$, the relative gradient error becomes larger than one and the cosine similarity remains far below one over much of the trajectory, while the parameter error continues to decrease. As $\beta$ increases, the relative error grows and the alignment with the population gradient weakens. A fixed shot budget therefore resolves progressively less of the descent signal, and the learning dynamics approaches the target less closely. Accurate gradient estimation is therefore a sufficient but potentially overly pessimistic requirement for end-to-end learning. The stochastic updates need only retain enough systematic descent when accumulated over the optimization trajectory. This distinction is substantial at moderate temperature but narrows as the useful gradient signal becomes difficult to resolve at low temperature.

\begin{figure}[ht!]
    \centering
    \includegraphics[width=0.8\linewidth]{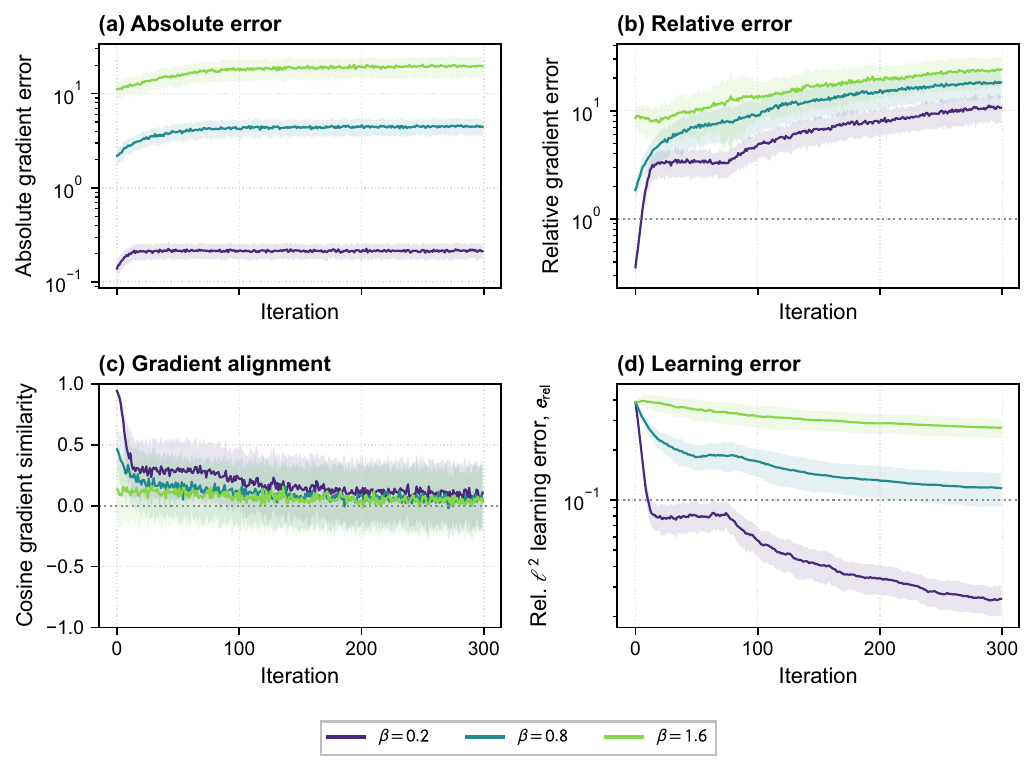}
    \caption{Gradient accuracy and learning from far initialization with $N=10^5$. The first three panels show the absolute gradient error $\|\widehat{\nabla J_Q}-\nabla J_Q\|_2$, its value relative to $\|\nabla J_Q\|_2$, and the cosine similarity between the estimated and exact gradients. The final panel shows the relative parameter error. Curves and shaded regions show the pointwise mean and one standard deviation over $100$ independent trajectories.}
    \label{fig:finite-shot-gradient-diagnostics}
\end{figure}

\section{Quantum score matching on IBM quantum hardware}\label{sec:ibm-hardware-experiment}

The simulation in the preceding section accounts for finite measurement shots and a finite number of randomized circuit instances, but assumes perfect quantum circuits without device noise. To see whether our method can run on real quantum devices, we deploy it on the IBM \texttt{ibm\_pittsburgh} quantum processor. We focus on learning the four-qubit TFIM in \cref{eq:numerical-inhomogeneous-tfim} with a homogeneous coupling parameter $J$ and field parameter $h$. These two parameters are shared across qubit sites, with $J_i=J$ and $h_i^x=h$, and we learn them within the parameter space $(J,h)\in\Theta_{\mathrm{num}}:=[0,2]^2$. The derivation frame is $\{X_x,Z_x\}_{x=1}^4$. The target state is specified by $J^*=1$ and $h^*=1.5$, while the optimization starts from $J^{(0)}=h^{(0)}=0.5$.

The target state is implemented as a probabilistic mixture of its energy eigenstates, with each randomly drawn eigenstate prepared using Qiskit's StatePreparation method. The time evolution of the model Hamiltonian is implemented by second-order Trotterization. An additional ancilla qubit is used to perform the Hadamard test, so each estimation circuit uses five qubits. At each gradient estimation step along a trajectory, we draw $256$ randomized circuit instances and measure each instance using $32$ shots. Hence, each complete gradient descent step requires only $8192$ measurement shots across these circuit instances, each on a freshly prepared target-state copy. The circuit instances are packed into parallel blocks on the processor, as described in \cref{app:ibm-hardware-implementation}.

As shown in \cref{fig:ibm-hardware-learning-errors}, the hardware measurements drive both parameters towards their target values, despite the presence of complex hardware noise and errors. After $45$ updates, the mean relative parameter error falls from $62\%$ to about $10\%$ at $\beta=0.2$ and to about $15\%$ at $\beta=0.4$. The separate coupling and field errors show that this improvement is not confined to one parameter.

\begin{figure}[ht!]
    \centering
    \includegraphics[width=0.98\linewidth]{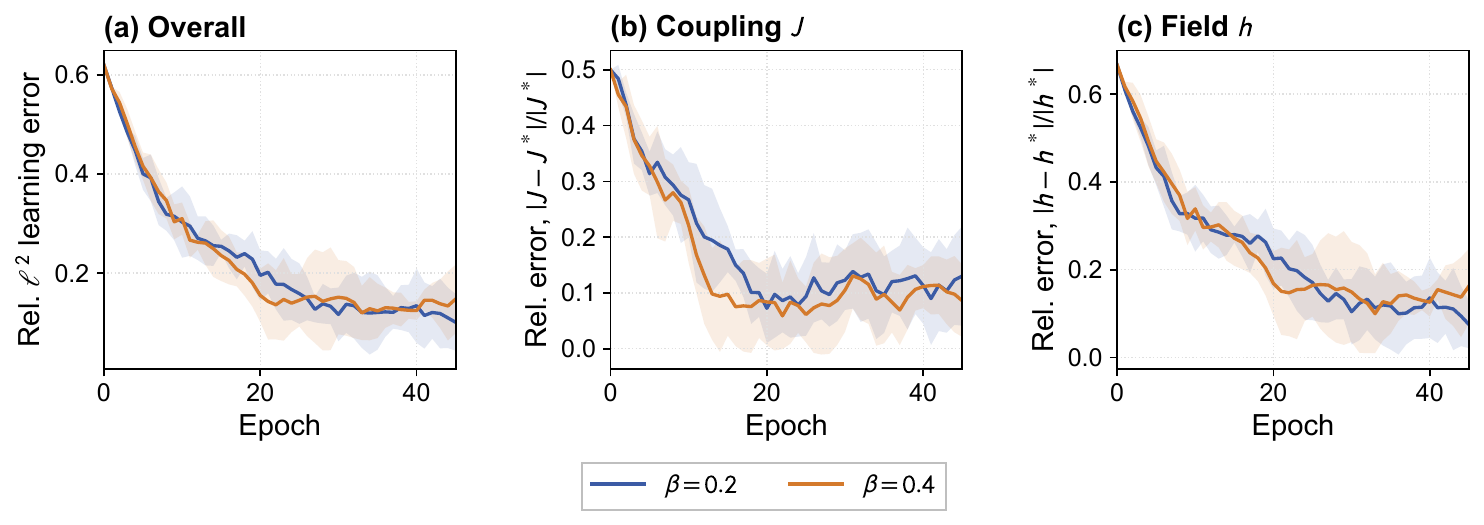}
    \caption{Learning a four-qubit TFIM on IBM quantum hardware. Curves and shaded regions show the mean and one sample standard deviation of the errors over five independent runs from a common initialization.}
    \label{fig:ibm-hardware-learning-errors}
\end{figure}

Furthermore, the mean parameter trajectories in \cref{fig:ibm-hardware-mean-dynamics} show how learning proceeds from the common initialization. Despite real-device noise, the optimizer moves closer to the target parameters as training progresses. The background population objective visualizes the optimization landscape, whose contours become more elongated and tilted as temperature decreases. The complete per-epoch trajectory summary is given in \cref{tab:ibm-epoch-trajectories}.

\begin{figure}[ht!]
    \centering
    \includegraphics[width=0.7\textwidth]{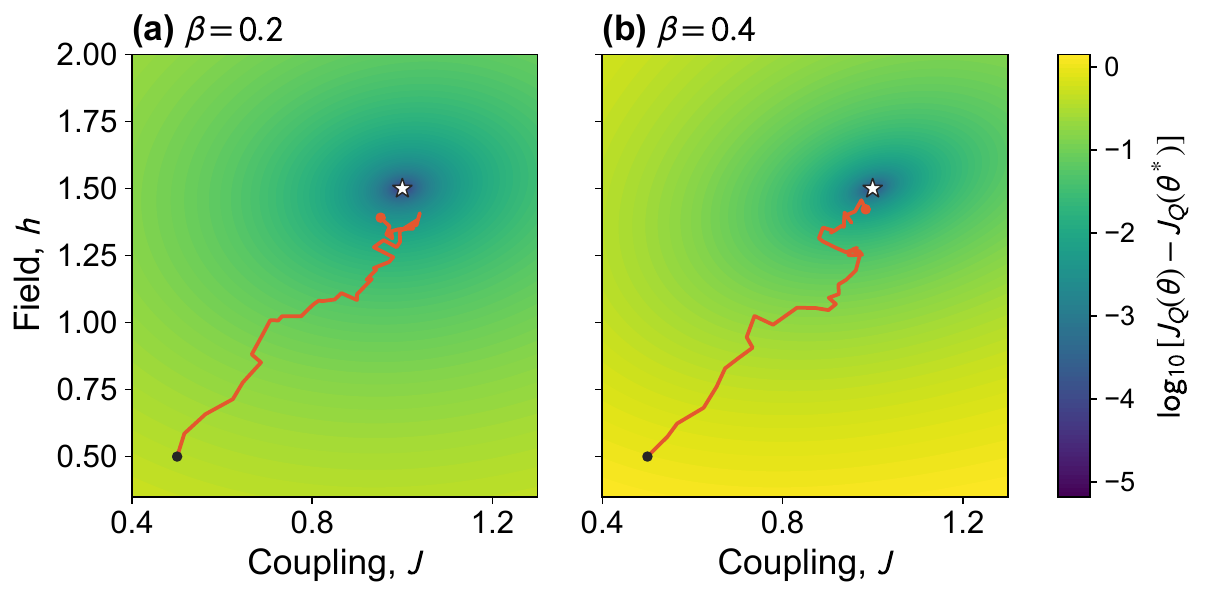}
    \caption{Mean hardware learning trajectories over $45$ updates on the population loss landscape. Each curve averages the parameters of five independent runs. The black point and white star mark the initialization and target.}
    \label{fig:ibm-hardware-mean-dynamics}
\end{figure}

\section{Discussion}

In this work, we develop quantum score matching, which extends a classical statistical learning principle from probability distributions to density operators. 
The central step is to connect the matching of noncommutative scores to an objective that can be evaluated using copies of an unknown state. This connection makes quantum scores not only a way to describe variations of a state, but also a means of learning it. 
We demonstrate the strength of this approach in the task of Gibbs-state learning, where our framework admits efficient quantum-circuit implementation and leads to sample-optimal learning in the high-temperature regime.
It thus advances both the theoretical and algorithmic foundations of quantum learning.

Our numerical and hardware experiments also point to a distinction between the accuracy of gradient estimation and the accuracy of learning. Controlling every gradient estimate accurately provides a route to rigorous convergence guarantees, but successful learning need not require such precise estimates at every iteration. In our simulations, parameter errors decrease even when individual gradient estimates have substantial relative errors, while the IBM experiments show that learning can proceed under both limited measurement budgets and device noise. Yet, using fewer measurement shots per iteration may require more iterations for convergence. This suggests that measurement budgets and optimization should be designed together.

The framework also leaves room to design the quantum score itself. The parameter $\lambda$ determines the operator ordering used to define and match scores, while the derivation frame specifies the directions along which states are compared. Our focus in the main text, with $\lambda=1/2$, already reveals useful analytical structure and supports rigorous learning guarantees. Exploring other score parameters and derivation frames may improve optimization and reduce measurement costs while preserving identifiability. 

It is worth noting that our quantum score-matching identities are universal and apply to state families beyond Gibbs-state models. However, obtaining explicit score representations generally requires a separate derivation associated with the specific parametrized state family. We leave the development of such generalizations to future work.

\paragraph{Acknowledgements.}
This research used resources of the National Energy Research Scientific Computing Center (NERSC), a Department of Energy Office of Science User Facility under Contract No.~DE-AC02-05CH11231, using NERSC awards DDR-ERCAP0038972 and DDR-ERCAP0038957 for high-performance computing and DDR-ERCAP0038973 for IBM Quantum Hub resources (Y.D.).
J.L. is partially sponsored by the Commonwealth Cyber Initiative (CCI) Faculty Fellow program.

\paragraph{Code and data availability.}
The code and data supporting the results of this study are available at \url{https://github.com/dongsnaq/Quantum-Score-Matching}.

\bibliographystyle{alphaurl}
\bibliography{ref}

\newpage
\appendix
\noindent
\begin{center}
    {\bf\Large Appendices}
\end{center}

\etocdepthtag.toc{appendices}
\begingroup
\etocsettagdepth{maintext}{none}
\etocsettagdepth{appendices}{all}
\etocsettocstyle{}{}
\tableofcontents
\endgroup

\section{Formulations of quantum score}

In this appendix, we discuss several ways to define the quantum score, i.e., the logarithmic derivative of a quantum density operator.

\paragraph{Weighted operator Fourier transform.}
Let $H=\sum_{j} \lambda_j P_j$ be an $n$-qubit Hamiltonian, where $\lambda_j$ and $P_j$ denote the eigenvalues and eigenprojectors of $H$. We denote $\mathrm{Spec}(H)$ as the spectrum (i.e., the set of eigenvalues) of $H$. The Bohr frequencies of $H$ are
\begin{align}
    B_H = \{\nu = \lambda_i - \lambda_j \colon\, \lambda_i, \lambda_j \in \mathrm{Spec}(H)\}.
\end{align}
A (non-zero) matrix $A \in \CC^{2^n\times 2^n}$ can be decomposed into Bohr components, 
\begin{align}
    A = \sum_{\nu \in B_H} A_\nu,\quad A_\nu \coloneqq \sum_{\lambda_i - \lambda_j = \nu} P_i A P_j.
\end{align}
Given an integrable function $\varphi(t)$, its Fourier transform is
\begin{align}
    \hat{\varphi}(\xi) := \int^\infty_{-\infty} e^{i\xi t} \varphi(t)~\d t,
\end{align}
and the inverse Fourier transform gives $\varphi(t) = \frac{1}{2\pi}\int_{-\infty}^{+\infty} e^{-i\xi t}\hat{\varphi}(\xi)~\d \xi$. 
The weighted operator Fourier transform (WOFT) implements a linear combination of the Bohr-frequency components weighted by $\hat{\varphi}$,
\begin{align}\label{eqn:woft}
    \int^\infty_{-\infty} A(t) \varphi(t)~\d t = \sum_{\nu \in B_H} \hat{\varphi}(\nu)A_\nu,\qquad \text{where}\quad A(t) := e^{iHt}A e^{-iHt} =\sum_{\nu \in B_H} e^{i \nu t} A_{\nu}.
\end{align}

\subsection{The $\lambda$-quantum score}\label{append:lambda_score}

Given a derivation generator $A_a$ and an $\lambda \in [0,1]$, we define the \textit{$\lambda$-score} of $\rho$ as the operator $S_{a, \lambda}$ that solves the matrix equation:
\begin{align}\label{eqn:general-lambda-definition}
    \partial_a \rho = (1-\lambda)\rho S_{\lambda, a} + \lambda S_{\lambda, a} \rho.
\end{align}
The special cases for $\lambda = 0, 1/2, 1$ are called the left, symmetric, and right logarithmic derivatives of $\rho$, respectively.

In the following lemma, we show that the $\lambda$-score of a full-rank $\rho$ is always well-defined.

\begin{lemma}
    Suppose that $\rho \succ 0$. For any $0 \le \lambda \le 1$, the super-operator
    \begin{align}
        \mathscr{I}_{\rho, \lambda}(X) := (1-\lambda)\rho X + \lambda X \rho
    \end{align}
    is strictly positive. Therefore, $\mathscr{I}_{\rho, \lambda}(\cdot)$ is invertible and the $\lambda$-score of $\rho$ is
    \begin{align}
        S_{a,\lambda} = \mathscr{I}^{-1}_{\rho, \lambda}(\partial_a \rho).
    \end{align}
\end{lemma}
\begin{proof}
    Any strictly positive linear operator in a finite-dimensional vector space is invertible.
    It suffices to check $\mathscr{I}_{\rho,\lambda}$ is positive. For any $X \neq 0$, 
    \begin{align}
        \Tr(X^\dagger \mathscr{I}_{\rho,\lambda}(X)) &= \Tr(X^\dagger \left((1-\lambda)\rho X + \lambda X \rho\right))\\
        &= (1-\lambda)\Tr(\rho XX^\dagger) + \lambda \Tr(\rho X^\dagger X),
    \end{align}
    which is strictly positive since $\rho \succ 0$. 
\end{proof}

Let $\rho \succ 0$, and we denote the left and right multiplication by $\rho$ as $\mathbf{L}_\rho$, $\mathbf{R}_\rho$, respectively:
\begin{align}
    \mathbf{L}_\rho(X) := \rho X,\qquad \mathbf{R}_\rho(X) := X \rho.
\end{align}
The modular operator is defined as
\begin{align*}
    \mathbf{\Delta}_\rho(X) := \rho X \rho^{-1} = \mathbf{L}_\rho \mathbf{R}^{-1}_\rho (X).
\end{align*}

\begin{lemma}\label{lem:fourer_representation_lambda_score}
    Suppose that $\rho \propto e^{-\beta H}$. 
    The $\lambda$-score of $\rho$ with respect to $\partial_a$ can be represented as 
    \begin{align}
        S_{\lambda, a} = \frac{1}{2\pi}\int^\infty_{-\infty} g_{\lambda}(t) e^{iHt}A_a e^{-iHt}\,\ud t,
    \end{align}
    where the kernal $g_\lambda(t)$ has the Fourier transform 
    \begin{align}\label{eqn:g_lambda_ft}
        \hat{g}_{\lambda}(\nu) = \I \frac{e^{-\beta \nu}-1}{\lambda + (1-\lambda)e^{-\beta \nu}}.
    \end{align}
\end{lemma}
\begin{proof}
    First, note that the super-operators $\mathscr{I}_{\rho, \lambda}$ and $\partial_a$ can be represented as 
    \begin{align}
        \mathscr{I}_{\rho, \lambda} = \left((1-\lambda) \mathbf{\Delta}_\rho + \lambda I\right) \mathbf{R}_\rho,\qquad  \partial_a \rho &= \I (\mathbf{\Delta}_\rho - I) \mathbf{R}_\rho(A_a), 
    \end{align}
    therefore,
    \begin{align*}
        S_{\lambda, a} = \mathscr{I}^{-1}_{\rho,\lambda}(\partial_a \rho) = \I\frac{\mathbf{\Delta}_\rho - I}{(1-\lambda) \mathbf{\Delta}_\rho + \lambda I}(A_a)\label{eq:score-fractional-rep} = \I \sum_{\nu}  \frac{e^{-\beta \nu}-1}{\lambda + (1-\lambda)e^{-\beta \nu}} (A_a)_{\nu},
    \end{align*}
and the weighted operator Fourier transform can be implemented using the kernel $g(t)$ such that its Fourier transform is given in~\cref{eqn:g_lambda_ft}.
\end{proof}

\begin{remark}
    In the special cases where $\lambda = 1/2$ and $\lambda=1$, the Fourier coefficients reduces to 
    \begin{align}
        \hat{g}_{1/2}(\nu) = -2i\tanh(\beta \nu/2),\qquad \hat{g}_1(\nu) = i(e^{-\beta\nu}-1).
    \end{align} 
\end{remark}

Now, we generalize the classical score matching framework to quantum learning and derive a quantum Hyv\"arinen formula. 
This formula enables sampling-based evaluation of the score-matching objective.
\begin{definition}\label{def:rho-lambda-weighted-inner-product}
    Let $\rho \succ 0$. We define the $(\rho, \lambda)$-weighted inner product:
    \begin{align}
        \langle X, Y\rangle_{\rho,\lambda} := \Tr(X^\dagger \mathscr{I}_{\rho,\lambda}(Y)) = \Tr\left(\rho\left(\lambda X^\dagger Y + (1-\lambda)YX^\dagger\right)\right).
    \end{align}
    And the associated 2-norm is given by
    \begin{align}
        \|X\|^2_{\rho,\lambda} := \Tr\left(\rho\left(\lambda X^\dagger X + (1-\lambda)XX^\dagger\right)\right).
    \end{align}
\end{definition}

Given two density operators $\rho, \sigma \succ 0$ and a frame of derivations $\{A_a\}_{a\in \cA}$, the $\lambda$-Fisher divergence is defined as 
\begin{align}
    \mathcal{F}_{\lambda,\cA}(\rho\|\sigma) := \frac{1}{2}\sum_{a\in \cA}\left\|\mathscr{I}^{-1}_{\rho,\lambda}(\partial_a \rho) - \mathscr{I}^{-1}_{\sigma,\lambda}(\partial_a \sigma)\right\|^2_{\sigma,\lambda}\,.
\end{align}

Suppose we have a parametrized density operator $\rho(\theta)$ and a data ensemble $\sigma$ such that $\rho(\theta), \sigma  \succ 0$ for any $\theta$. 
A natural formulation of quantum score matching requires minimizing the quantum Fisher divergence. For a fixed $\lambda \in [0,1]$, we define
\begin{align}\label{eqn:q_F_div}
    J_Q(\theta) := \mathcal{F}_{\lambda,\cA}(\rho(\theta)\|\sigma) = \frac{1}{2}\sum_{a\in \cA}\left\|S_{a}(\theta) - S^*_a\right\|^2_{\sigma,\lambda}\, ,
\end{align}
where 
\begin{align}
    S_a(\theta) := \mathscr{I}^{-1}_{\rho(\theta),\lambda}(\partial_a \rho(\theta)),\qquad S^*_a := \mathscr{I}^{-1}_{\sigma,\lambda}(\partial_a \sigma).
\end{align}

\begin{lemma}[Quantum Hyv\"arinen formula]\label{lem:quantum-hyvarinen}
    For every $\lambda\in[0,1]$,
    \begin{align}\label{eqn:quantum-hyvarinen}
        J_Q(\theta)
        = \frac{1}{2}\sum_a\Tr\left(\sigma\left(
            \lambda S_a^\dagger(\theta)S_a(\theta)
            +(1-\lambda)S_a(\theta)S_a^\dagger(\theta)
            +\partial_a\bigl(S_a(\theta)+S_a^\dagger(\theta)\bigr)
        \right)\right) + C_\lambda,
    \end{align}
    where 
    \begin{align}
        C_\lambda
        := \frac{1}{2}\sum_a\|S_a^*\|_{\sigma,\lambda}^2
        = \frac{1}{2}\sum_a\Tr\left(\sigma\left(
            \lambda(S_a^*)^\dagger S_a^*
            +(1-\lambda)S_a^*(S_a^*)^\dagger
        \right)\right)
    \end{align}
    is independent of the parametrized operator $\rho(\theta)$.
\end{lemma}
\begin{proof}
    Fix $a$ and abbreviate $S_a(\theta)$ by $S_a$. Expanding the
    corresponding summand in~\cref{eqn:q_F_div} separates the two pure
    terms from the cross term:
    \begin{align}
        \frac{1}{2}\|S_a-S_a^*\|_{\sigma,\lambda}^2
        &= \frac{1}{2}\|S_a\|_{\sigma,\lambda}^2 + \frac{1}{2}\|S_a^*\|_{\sigma,\lambda}^2 -\frac{1}{2}T_a,\label{eqn:hyvarinen-expansion}\\
        \text{where}\quad T_a
        &:= \langle S_a, S^*_a\rangle_{\sigma,\lambda}+\langle S^*_a, S_a\rangle_{\sigma,\lambda}.\label{eqn:hyvarinen-cross-term}
    \end{align}
    It remains to simplify $T_a$. Let the adjoint operator of the derivation operator under the $(\sigma, \lambda)$-weighted inner product be $\partial^\sharp_a$ which satisfies
    \begin{equation}
        \braket{X, \partial_a Y}_{\sigma, \lambda} = \braket{\partial^\sharp_a X, Y}_{\sigma, \lambda}.
    \end{equation}
    Note that the regular adjoint under the Hilbert--Schmidt inner product is $\partial_a^\dagger = - \partial_a$. From \cref{def:rho-lambda-weighted-inner-product}, we see that
    \begin{equation}
        \partial_a^\sharp = \mathscr{I}_{\sigma, \lambda}^{-1} \partial_a^\dagger \mathscr{I}_{\sigma, \lambda} \quad \text{ and }\quad \partial_a \mathscr{I}_{\sigma, \lambda} = \mathscr{I}_{\sigma, \lambda} \partial_a + \mathscr{I}_{\partial_a \sigma, \lambda}.
    \end{equation}
    Therefore
    \begin{equation}
        \partial_a^\sharp = - \partial_a - \mathscr{I}_{\sigma, \lambda}^{-1} \mathscr{I}_{\partial_a \sigma, \lambda}.
    \end{equation}
    Now, note that the following identity holds
    \begin{equation}
        - \partial_a^\sharp(I) = \partial_a(I) + \mathscr{I}_{\sigma, \lambda}^{-1} \mathscr{I}_{\partial_a \sigma, \lambda}(I) = \mathscr{I}_{\sigma, \lambda}^{-1}(\partial_a \sigma) = S_a^*.
    \end{equation}
    We can apply the integral-by-parts formula through the adjoint operator
    \begin{align}
        \braket{S_a, S_a^*}_{\sigma, \lambda} &= - \braket{S_a, \partial_a^\sharp(I)}_{\sigma, \lambda} = - \braket{\partial_a (S_a), I}_{\sigma, \lambda} = - \Tr(\sigma (\partial_a(S_a))^\dagger) = - \Tr(\sigma \partial_a(S_a^\dagger)),\label{eq:inner-prod-Sa-1}\\
        \braket{S_a^*, S_a}_{\sigma, \lambda} &= - \braket{\partial_a^\sharp(I), S_a}_{\sigma, \lambda} = - \braket{I, \partial_a(S_a)}_{\sigma, \lambda} = - \Tr(\sigma \partial_a(S_a))\label{eq:inner-prod-Sa-2}.
    \end{align}
    Combining~\cref{eq:inner-prod-Sa-1} and~\cref{eq:inner-prod-Sa-2}, we obtain 
    \begin{align}\label{eqn:Ta-partial-S-relation}
        T_a=-\Tr\left(\sigma\,
            \partial_a\bigl(S_a+S_a^\dagger\bigr)\right).
    \end{align}
    Substituting this into~\cref{eqn:hyvarinen-expansion} and summing over
    $a$ proves~\cref{eqn:quantum-hyvarinen}.
\end{proof}

In the commutative setting, our construction reduces to the classical score-matching formalism.
To see this, we use $\widetilde{\cdot}$ to denote the classical counterparts of these operators. Suppose $\wt{\sigma}(x)$ is a probability density. Then
\begin{equation}
    (\wt{\mathscr{I}}_{\sigma, \lambda} f)(x) = \wt{\sigma}(x) f(x) \quad \text{and}\quad \braket{f, g}_{\sigma, \lambda} = \int f^\dagger (x) g(x) \sigma(x) \ud x.
\end{equation}
Then, the adjoint operator becomes the regular integration by part with the score trick
\begin{equation}
    \wt{\partial_a^\sharp}f(x) = - \frac{\partial f(x)}{\partial x_a} - f(x) \frac{\partial \log \wt{\sigma}(x)}{\partial x_a} \quad \text{and} \quad \wt{\nabla^\sharp} = - \nabla - (\nabla \log \wt{\sigma}). 
\end{equation}
From this correspondence, it can also be seen that
\begin{equation}
    - \wt{\nabla^\sharp}(1)(x) = \nabla \log \wt{\sigma}(x) = s^*(x).
\end{equation}
In classical setting, $s(x; \theta) = \nabla \log q(x; \theta)$ and 
\begin{equation}
    \braket{s_\theta, s^*}_{\sigma, \lambda} = - \braket{s_\theta, \wt{\nabla^\sharp}(1)}_{\sigma, \lambda} = - \braket{\nabla s_\theta, 1}_{\sigma, \lambda} = - \int \Delta \log q(x; \theta) \wt{\sigma}(x) \ud x.
\end{equation}
This exactly reduces to the classical Hyv\"arinen formula.

However, we note that our $\lambda$-score is not equivalent to the $\partial_a$ derivative of the ``logarithm'' of a density operator. In particular, if $\rho = e^{-H}/\Tr(e^{-H})$ for some Hamiltonian $H$, the matrix logarithm of $\rho$ is simply $\log(\rho) = -H - \log(\Tr(e^{-H}))$. In general,
\begin{align}
    S_{\lambda, a} \neq \partial_a (\log(\rho)) = - \partial_a H.
\end{align}
In~\cref{append:matrix-logarithm-recover}, we discuss an alternative quantum score defined by the Bogoliubov--Kubo--Mori (BKM) logarithmic derivative, which recovers the derivative of the matrix logarithm exactly. A formal Hyv\"arinen identity still holds, but its quadratic term depends nonlinearly on the unknown target state and does not in general reduce to an expectation of a candidate-dependent observable.

Below, we give the global identifiability condition for quantum score matching using the $\lambda$-score. Recall that $\cS'$ is the joint commutant of all operators in $\cS$.

\begin{theorem}[Global identifiability]
    \label{thm:global-identifiability}
    For $\lambda\in [0,1]$, let $\{A_a\}_{a\in \cA}$ be a derivation frame and $\sigma \propto e^{-\beta H^*}$ be a thermal state.
    Using this derivation frame and target state, we define $J_Q(\theta)$ as in~\cref{eqn:score-matching-objective}, let $\hat{\theta}$ be a global minimizer of $J_Q(\theta)$. 
    Define 
    \begin{align}\label{eqn:K-sigma-lambda-a}
        K^\sigma_{\lambda, a} := \left(\lambda \mathbf{\Delta}^{-1/2}_\sigma + (1-\lambda) \mathbf{\Delta}^{1/2}_\sigma\right)^{-1}(A_a)\,.
    \end{align}
    If $\{K^\sigma_{\lambda,a}:a\in\cA\}'=\CC I$, we have $\rho(\hat{\theta}) = \sigma$.
\end{theorem}
\begin{proof}
    Since $\theta^*$ is feasible, $J_Q(\hat{\theta})=J_Q(\theta^*)=0$. Since $\sigma \succ 0$, this implies that the $\lambda$-score of $\rho(\hat{\theta})$ and $\sigma$ matches on every direction $a \in \cA$. In what follows, we write $\rho=\rho(\hat{\theta})$ for simplicity, and denote $L_a$ as the identitcal score in the direction $a$, so
    \begin{align}
        L_a := S_a(\hat{\theta}) = S^*_a,\quad \forall~ a\in \cA.
    \end{align}
    Rearranging the score equations for $\sigma$ and $\rho$ gives
    \begin{align}\label{eq:score-equations-sigma-rho}
        \widetilde B_{\lambda,a}^{\sigma}\sigma=\sigma\widetilde A_{\lambda,a}^{\sigma},
        \qquad
        \widetilde B_{\lambda,a}^{\sigma}\rho=\rho\widetilde A_{\lambda,a}^{\sigma},
    \end{align}
    where
    \begin{align}
        \widetilde A_{\lambda,a}^{\sigma}:=A_a+\I(1-\lambda)L_a,
        \qquad
        \widetilde B_{\lambda,a}^{\sigma}:=A_a-\I\lambda L_a.
    \end{align}
    Rearranging \cref{eq:score-equations-sigma-rho} and conjugating by $\sigma^{1/2}$, we obtain
    \begin{align}
        [\sigma^{-1/2}\rho \sigma^{-1/2}, \sigma^{1/2}\widetilde{A}_{\lambda,a}^{\sigma} \sigma^{-1/2}] = 0,\quad \forall~a\in \cA.
    \end{align}
    By~\cref{eq:score-fractional-rep}, we can further simplify 
    \begin{align}
        \sigma^{1/2}\widetilde{A}_{\lambda,a}^{\sigma}\sigma^{-1/2} &=  \sigma^{1/2} \left(I -(1-\lambda)\frac{\mathbf{\Delta}_\sigma - I}{(1-\lambda)\mathbf{\Delta}_\sigma + \lambda I}\right)(A_a) \sigma^{-1/2}\\
         &= \left(\lambda \mathbf{\Delta}^{-1/2}_\sigma + (1-\lambda) \mathbf{\Delta}^{1/2}_\sigma\right)^{-1}(A_a) =: K^\sigma_{\lambda,a}.
    \end{align}
    In other words, $J_Q(\hat{\theta})=0$ implies that for every $a \in \cA$, there is 
    \begin{align}\label{eq:consequence-zero-score}
        [\sigma^{-1/2}\rho \sigma^{-1/2}, K^\sigma_{\lambda, a}] = 0.
    \end{align}
    If $\{K^\sigma_{\lambda,a}:a\in\cA\}'=\CC I$, by~\cref{eq:consequence-zero-score}, there exists a $c \in \mathbb{C}$ such that
        \begin{equation}
            \sigma^{-1/2}\rho \sigma^{-1/2} = c I \implies \rho = c\sigma.
        \end{equation}
    Since $\Tr(\rho)=\Tr(\sigma)=1$, we must have $c = 1$ and this implies $\rho = \sigma$.
\end{proof}

At the two endpoints $\lambda = 0, 1$, the commutant condition is significantly simplified, and we show that a set of single-qubit Pauli operators suffices. 

\begin{lemma}[Endpoint logarithmic derivatives]
\label{lem:endpoint-consistency}
    For $\lambda=1$ and $\lambda=0$,
    \begin{align}
        K_{1,a}^{\sigma}
        =\sigma^{1/2}A_a\sigma^{-1/2},
        \qquad
        K_{0,a}^{\sigma}
        =\sigma^{-1/2}A_a\sigma^{1/2}.
    \end{align}
    Hence $\{A_a:a\in\cA\}'=\CC I$ guarantees global Fisher consistency at
    either endpoint.  In particular, the single-qubit operators
    $\{X_x,Z_x\}_{x=1}^{n}$ form such a frame.
\end{lemma}
\begin{proof}
    At either endpoint, the transformed frame is obtained from the
    original frame by one common similarity transformation.  For example,
    \begin{align*}
        \{K_{1,a}^{\sigma}\}_a'=\sigma^{1/2}\{A_a\}_a'\sigma^{-1/2}.
    \end{align*}
    Thus a scalar commutant remains scalar.  The local $X_x$ and $Z_x$
    generate the full matrix algebra, so their commutant is $\CC I$.
\end{proof}

Finally, we prove that the target parameter $\theta^*$ is always a first-order stationary point of $J_Q(\theta)$, regardless of the choice of derivation frame.

\begin{lemma}[First-order stationarity]\label{lem:first-order-stationary}
    Suppose that $H^*=H(\theta^*)$ for a $\theta^* \in \Theta$. Let $J_Q(\theta)$ be the same as in~\cref{eqn:score-matching-objective}. 
    Then we have
    \begin{align}
        J_Q(\theta^*) = 0,\qquad \nabla_\theta J_Q(\theta^*) = 0.
    \end{align} 
\end{lemma}
\begin{proof}
    Since $H(\theta^*)=H^*$, we have $\rho(\theta^*)=\sigma$ and hence
    $S_a(\theta^*)=S_a^*$ for every $a\in\cA$. Substitution into
    \cref{eqn:score-matching-objective} gives $J_Q(\theta^*)=0$.
    Since the Gibbs family is smooth and full-rank, the score map is
    differentiable. Differentiating the squared score distance with respect
    to any parameter $\theta_j$ yields
    \begin{align}
        \partial_{\theta_j}J_Q(\theta)=\sum_{a\in\cA}
        \left\langle
            \partial_{\theta_j}S_a(\theta),
            S_a(\theta)-S_a^*
        \right\rangle_{\sigma}.
    \end{align}
    Evaluating at $\theta=\theta^*$ makes every summand vanish, so
    $\partial_{\theta_j}J_Q(\theta^*)=0$ for all $j$.
\end{proof}

\subsection{The Bogoliubov--Kubo--Mori logarithmic derivative}
\label{append:matrix-logarithm-recover}

For a full-rank density operator $\rho$, define the BKM map and score by~\cite{petz1993bogoliubov},
\begin{equation}
    \Omega_\rho(O):=\int_0^1\rho^s O\rho^{1-s}\ud s,
    \qquad K_a^\rho:=\partial_a\log\rho.
    \label{eqn:bkm-map-score}
\end{equation}
Duhamel's formula gives $\partial_a\rho=\Omega_\rho(K_a^\rho)$,
so $K_a^\rho$ is the logarithmic derivative associated with
$\Omega_\rho$. For a Gibbs state $\rho_H=e^{-\beta H}/Z$, the scalar
normalization disappears under $\partial_a=-\I[A_a,\cdot]$, giving
\begin{equation}
    K_a^{\rho_H}=-\beta\partial_aH=\I\beta[A_a,H].
    \label{eqn:bkm-hamiltonian-score}
\end{equation}
This identity holds for every finite $\beta>0$, without the nonlinear
function of $\mathbf{ad}_H$ appearing in the SLD score. In particular,
the BKM score is linear in the parameters of a Pauli Hamiltonian.

To construct a score-matching objective, the logarithmic derivative
must be paired with a choice of inner product. The corresponding BKM
inner product is
\begin{equation}
    \langle X,Y\rangle_{\sigma,\mathrm{BKM}}
    :=\Tr(X^\dagger\Omega_\sigma(Y)).
\end{equation}
Writing $K_a(\theta):=K_a^{\rho(\theta)}$ and $K_a^*:=K_a^\sigma$, a
natural alternative to the objective in the main text is
\begin{equation}
    J_{\mathrm{BKM}}(\theta)
    :=\frac12\sum_a
      \|K_a(\theta)-K_a^*\|_{\sigma,\mathrm{BKM}}^2.
    \label{eqn:bkm-score-matching}
\end{equation}
There is, in fact, an exact integration-by-parts identity. Since
$\Omega_\sigma(K_a^*)=\partial_a\sigma$ and the trace of a commutator
vanishes,
\begin{equation}
    \langle K_a(\theta),K_a^*\rangle_{\sigma,\mathrm{BKM}}
    =\Tr(K_a(\theta)\partial_a\sigma)
    =-\Tr(\sigma\partial_aK_a(\theta)).
\end{equation}
Expanding the square therefore eliminates the target score:
\begin{equation}
    J_{\mathrm{BKM}}(\theta)
    =\frac12\sum_a\Tr(K_a(\theta)\Omega_\sigma(K_a(\theta)))
      +\sum_a\Tr(\sigma\partial_aK_a(\theta))+C_\sigma,
    \label{eqn:bkm-formal-hyvarinen}
\end{equation}
where $C_\sigma$ is independent of $\theta$.

The difficulty of evaluating the BKM score-matching objective comes from the first term in~\cref{eqn:bkm-formal-hyvarinen}.
It contains the correlation
\[
    \Tr(K\Omega_\sigma(K))
    =\int_0^1\Tr(K\sigma^sK\sigma^{1-s})\ud s,
\]
which depends nonlinearly on the unknown state $\sigma$.
For the SLD inner product, the analogous
quantity is simply $\Tr(\sigma K^2)$, an expectation of a known
observable. The BKM expression does not in general admit this
reduction, even though $K_a(\theta)$ itself is easy to construct.

\section{Gradient estimation via importance sampling}
\label{append:gradient-estimation-proof}

We first give an explicit implementation of the estimator in
\cref{thm:informal-gradient-estimation-cost}. For each measurement shot, the
algorithm uses classical randomness to choose a unitary circuit, runs a
Hadamard test on one fresh copy of $\sigma$, and rescales its binary
outcome. The proof below bounds the bias from truncating the evolution
times and the number of shots needed to estimate the mean.

\subsection{The gradient estimation algorithm}

\paragraph{Classical preprocessing and time sampling.}
Fix a coordinate $j$ and abbreviate
$B_a:=\partial_aH_\theta$ and $Q_{a,j}:=\partial_aP_j$.
Expand $A_a,B_a,Q_{a,j}$ in Hermitian Pauli products, with their real
coefficients stored classically. For any such expansion
$O=\sum_r c_rV_r$ with $\ell(O)>0$, let $\mathrm{Pauli}(O)$ denote the
following classical draw: choose $r$ with probability $|c_r|/\ell(O)$
and return $(\chi,V)=(\operatorname{sgn}(c_r),V_r)$. Thus
\[
    O=\ell(O)\mathbb E(\chi V).
\]
Zero coefficients and choices with zero total weight are always omitted.
For a Pauli frame, $Q_{a,j}$ is either zero or twice a signed Pauli
product, so its nonzero Pauli draw is deterministic.

Let $\nu$ be the law of the random time $\xi$ in
\cref{eqn:probablistic_S_a}. For a cutoff $R>0$, define
\begin{align}
    z_0(R)&:=\mathbb E\mathbf1_{\{|\xi|\leq R\}},&
    \nu_{0,R}(\ud u)&:=\frac{\mathbf1_{\{|u|\leq R\}}}{z_0(R)}\nu(\ud u),\nonumber\\
    z_1(R)&:=\mathbb E(|\xi|\mathbf1_{\{|\xi|\leq R\}}),&
    \nu_{1,R}(\ud u)&:=\frac{|u|\mathbf1_{\{|u|\leq R\}}}{z_1(R)}\nu(\ud u).
    \label{eqn:gradient-algorithm-time-laws}
\end{align}
Both distributions can be sampled by rejection. To sample $\nu_{0,R}$,
draw $t\sim q_\beta$ and $u\sim\mathrm{Uniform}[-t,t]$ until $|u|\leq R$.
To sample $\nu_{1,R}$, additionally accept a draw from $\nu_{0,R}$ with
probability $|u|/R$, repeating until acceptance. These rejections are
entirely classical and consume no copies of $\sigma$ or Hamiltonian
evolution. The normalizing factors are
$z_0(R)=1-p_R$ and $z_1(R)=c_U\beta-r_R$, with the convergent expressions
for $p_R,r_R$ derived in
\cref{eqn:gradient-sampling-time-tail,eqn:gradient-sampling-weighted-tail}.

The coefficients that determine the branch probabilities are
\begin{align}
    s_{a,R}&:=\beta\mathfrak b_a z_0(R),&
    t^{(1)}_{a,j,R}&:=\beta\mathfrak c_{a,j}z_0(R),\nonumber\\
    t^{(2)}_{a,j,R}&:=2\beta\mathfrak b_a z_1(R),&
    t_{a,j,R}&:=t^{(1)}_{a,j,R}+t^{(2)}_{a,j,R},\nonumber\\
    L_j(R)&:=\sum_{a\in\cA}(s_{a,R}+2\mathfrak a_a)t_{a,j,R}.
    \label{eqn:gradient-algorithm-masses}
\end{align}
All of these quantities depend only on the candidate Hamiltonian and the
chosen frame. In particular, preprocessing does not require measurements
of the target state.

\paragraph{Sampling the elementary factors.}
Write $S_{a,R}$ and $T^{(d)}_{a,j,R}$ for the expectations in
\cref{eqn:probablistic_S_a,eqn:probablistic_T_a_1,eqn:probablistic_T_a_2_final}
with an indicator $\mathbf1_{\{|\xi|\leq R\}}$ inserted, without
renormalizing them, and set $T_{a,j,R}=T^{(1)}_{a,j,R}+T^{(2)}_{a,j,R}$.
For the parameter-shift branch, define
\begin{equation}
    W_{j,\eta}(u,s;V):=
    \tau^{H_\theta}_{(1-s)u}\!\left(
        \mathcal R_{j,\eta}(\tau^{H_\theta}_{su}(V))\right).
    \label{eqn:gradient-algorithm-shifted-unitary}
\end{equation}
\Cref{tab:gradient-elementary-samplers} specifies a unitary-valued sample
$\mathsf U_F$ for each elementary factor $F$, including its coefficient
sign. Each row satisfies $F=m_F\mathbb E\mathsf U_F$, where $m_F$ is the
weight in that row. For the last row, also draw independent
$s\sim\mathrm{Uniform}[0,1]$ and a fair sign $\eta\in\{-1,1\}$.

\begin{table}[H]
    \centering
    \small
    \begin{tabular}{@{}lllll@{}}
        \toprule
        Factor $F$ & Weight $m_F$ & Pauli draw $(\chi,V)$ & Time law & Sample $\mathsf U_F$ \\
        \midrule
        $A_a$ & $\mathfrak a_a$ & $\mathrm{Pauli}(A_a)$ & --- & $\chi V$ \\
        $S_{a,R}$ & $s_{a,R}$ & $\mathrm{Pauli}(B_a)$ & $\nu_{0,R}$ & $-\chi\tau_u^{H_\theta}(V)$ \\
        $T^{(1)}_{a,j,R}$ & $t^{(1)}_{a,j,R}$ & $\mathrm{Pauli}(Q_{a,j})$ & $\nu_{0,R}$ & $-\chi\tau_u^{H_\theta}(V)$ \\
        $T^{(2)}_{a,j,R}$ & $t^{(2)}_{a,j,R}$ & $\mathrm{Pauli}(B_a)$ & $\nu_{1,R}$ & $-\chi\operatorname{sgn}(u)\eta W_{j,\eta}(u,s;V)$ \\
        \bottomrule
    \end{tabular}
    \caption{Classical sampling rules for the elementary unitary factors.
    The Pauli, time, insertion-time, and coin draws are independent.}
    \label{tab:gradient-elementary-samplers}
\end{table}

To sample a unitary $\mathsf U_T$ representing $T_{a,j,R}$, choose
$d\in\{1,2\}$ with probability $t^{(d)}_{a,j,R}/t_{a,j,R}$ and use
the corresponding row of \cref{tab:gradient-elementary-samplers}.
Then $T_{a,j,R}=t_{a,j,R}\mathbb E\mathsf U_T$.
The factor $|u|$ in the time law $\nu_{1,R}$ is essential: it absorbs
the magnitude of the Duhamel coefficient into $t^{(2)}_{a,j,R}$,
leaving only $\operatorname{sgn}(u)\eta$ in the sampled unitary.

\paragraph{Combining factors and measuring one shot.}
For a fixed frame index $a$, expand its contribution to the truncated
gradient as
\[
    \frac12S_{a,R}T_{a,j,R}+\frac12T_{a,j,R}S_{a,R}
      -\I A_aT_{a,j,R}+\I T_{a,j,R}A_a.
\]
Its four product branches are listed in
\cref{tab:gradient-product-samplers}. First draw the frame index with
probability
\begin{equation}
    \PP(a)=\frac{(s_{a,R}+2\mathfrak a_a)t_{a,j,R}}{L_j(R)},
    \label{eqn:gradient-algorithm-frame-law}
\end{equation}
then draw a row with probability equal to its weight divided by
$(s_{a,R}+2\mathfrak a_a)t_{a,j,R}$. Draw the two factors in that row
independently, even when both use the same Pauli decomposition or time
law. In the table, $\mathsf U_S,\mathsf U_A,\mathsf U_T$ denote fresh
samples for $S_{a,R},A_a,T_{a,j,R}$, respectively.

\begin{table}[H]
    \centering
    \begin{tabular}{@{}lll@{}}
        \toprule
        Product branch & Weight & Sampled unitary $\mathsf U$ \\
        \midrule
        $S_{a,R}T_{a,j,R}/2$ & $s_{a,R}t_{a,j,R}/2$ & $\mathsf U_S\mathsf U_T$ \\
        $T_{a,j,R}S_{a,R}/2$ & $s_{a,R}t_{a,j,R}/2$ & $\mathsf U_T\mathsf U_S$ \\
        $-\I A_aT_{a,j,R}$ & $\mathfrak a_a t_{a,j,R}$ & $-\I\mathsf U_A\mathsf U_T$ \\
        $\I T_{a,j,R}A_a$ & $\mathfrak a_a t_{a,j,R}$ & $\I\mathsf U_T\mathsf U_A$ \\
        \bottomrule
    \end{tabular}
    \caption{The four algebraic branches for one frame index. Their
    weights sum to $(s_{a,R}+2\mathfrak a_a)t_{a,j,R}$.}
    \label{tab:gradient-product-samplers}
\end{table}

These choices give $\mathbb E\mathsf U=\mathcal G_{j,R}/L_j(R)$.
Prepare an ancilla in $|+\rangle$ and the system in $\sigma$, apply
controlled-$\mathsf U$, and measure Pauli $X$ on the ancilla. The outcome
$Y\in\{-1,1\}$ has conditional mean
$\mathbb E(Y\mid\mathsf U)=\Re\Tr(\sigma\mathsf U)$.
Consequently, $Z=L_j(R)Y$ estimates the truncated gradient.
The known phases $\pm1,\pm\I$ in the tables are implemented as phase
gates on the control ancilla. Every other factor is a conjugated Pauli
$CVC^\dagger$, so its controlled circuit is implemented by
$(I\otimes C)\operatorname{ctrl}(V)(I\otimes C^\dagger)$.
In particular, the Hamiltonian evolutions within $C$ are ordinary forward
and backward evolutions; only the Pauli insertion is controlled.

\begin{algorithm}[H]
    \caption{Estimating one gradient coordinate by parameter-shift sampling}
    \label{alg:randomized-gradient-coordinate}
    \KwIn{Candidate $H_\theta$, coordinate $j$, inverse temperature $\beta$,
    frame and Pauli decompositions, copies of $\sigma$, error $\varepsilon>0$,
    and failure probability $0<\delta<1$.}
    Compute $L_j$ from \cref{eqn:gradient-sampling-coordinate-mass}\;
    \If{$\varepsilon\geq L_j$}{\textbf{return} $0$\;}
    Set $R=(2\beta/\pi)\log(2+12L_j/\varepsilon)$\;
    Compute $z_0(R),z_1(R)$ and the weights in
    \cref{eqn:gradient-algorithm-masses}\;
    \If{$L_j(R)=0$}{\textbf{return} $0$\;}
    Set $N=\left\lceil8L_j^2(R)\varepsilon^{-2}\log(2/\delta)\right\rceil$\;
    \For{$k=1,\ldots,N$}{
        Draw $a$ according to \cref{eqn:gradient-algorithm-frame-law}\;
        Draw a product branch from \cref{tab:gradient-product-samplers}
        proportionally to its weight\;
        Draw the response type $d$ with probability
        $t^{(d)}_{a,j,R}/t_{a,j,R}$ and sample its elementary unitary
        using \cref{tab:gradient-elementary-samplers}\;
        Independently sample the companion score or frame factor from
        \cref{tab:gradient-elementary-samplers}, and form $\mathsf U_k$
        in the selected product order, including its phase\;
        Run the ancilla-$X$ Hadamard test for $\mathsf U_k$ on one fresh
        copy of $\sigma$, obtaining $Y_k\in\{-1,1\}$\;
        Set $Z_k=L_j(R)Y_k$\;
    }
    \KwOut{$\widehat g_j=N^{-1}\sum_{k=1}^N Z_k$.}
\end{algorithm}

The choices of $R$ and $N$ assign at most $\varepsilon/2$ to truncation
bias and $\varepsilon/2$ to statistical error, respectively. The
algorithm estimates one coordinate using one copy per shot. To estimate
all $m$ coordinates with simultaneous additive error $\varepsilon$,
run it for each $j$ with failure budget $\delta/m$; a union bound then
gives $\|\widehat g-\nabla_\theta J_Q(\theta)\|_\infty\leq\varepsilon$
with probability at least $1-\delta$, and the resources add over the
coordinates.

\subsection{Correctness and resource bounds}

To state the resource bounds, fix Pauli decompositions $O=\sum_r c_r P_r$ and let $\ell(O):=\sum_r|c_r|$.
For operators $A_a$, $\partial_aH_\theta$, and $\partial_aP_j$, we write 
\begin{equation}
    \mathfrak{a}_a:=\ell(A_a),\qquad \mathfrak{b}_a:=\ell(\partial_aH_\theta),\qquad
    \mathfrak{c}_{a,j}:=\ell(\partial_aP_j).
    \label{eqn:gradient-sampling-local-masses}
\end{equation}
Define
\begin{equation}
    L_j:=\sum_{a\in\cA}(\beta \mathfrak{b}_a+2\mathfrak{a}_a)(\beta \mathfrak{c}_{a,j}+2c_U\beta^2 \mathfrak{b}_a),
    \qquad c_U:=\frac{7\zeta(3)}{\pi^3}.
    \label{eqn:gradient-sampling-coordinate-mass}
\end{equation}

\begin{theorem}[Gradient estimation]
\label{thm:parameter-shift-gradient-cost}
    Fix $\theta$, $\beta>0$, a coordinate $j$, $0<\delta<1$, and $\varepsilon>0$.
    Assume exact classical sampling of the distributions below, access to $e^{\I tH_\theta}$ for the requested positive and negative times $t$, and implementations of controlled Pauli insertions and the rotations $e^{\pm\I\pi P_j/4}$.
    With success probability at least $1-\delta$, we can estimate $\partial_{\theta_j}J_Q(\theta)$ to additive error at most $\varepsilon$ using
    \begin{enumerate}
        \item total number of copies of the thermal state $\sigma$:
        $N_{\rm copy}=\mathcal{O}\!\left(\dfrac{L_j^2}{\varepsilon^2}\log\dfrac{2}{\delta}\right)$;
        \item maximal total evolution time of $H_\theta$ in one measurement circuit:
        $T_{\max}=\mathcal{O}\!\left(\beta\log\left(2+\dfrac{L_j}{\varepsilon}\right)\right)$;
        \item total evolution time of $H_\theta$ over all measurement circuits:
        \[
            T_{\rm tot}=\mathcal{O}\!\left(
                \frac{\beta L_j^2}{\varepsilon^2}
                \log\frac{2}{\delta}\,
                \log\left(2+\frac{L_j}{\varepsilon}\right)
            \right).
        \]
    \end{enumerate}
\end{theorem}

\begin{remark}
    When $\mathfrak{a}_a,\mathfrak{b}_a,\mathfrak{c}_{a,j}$ are all absolute constants independent of $|\cA|$ and $\beta$, one has $L_j=\mathcal{O}(|\cA|\beta^3)$.
    In this case, both the number of samples of $\sigma$ and the total evolution time of $H_\theta$ scale polynomially in parameters $|\cA|$, $\beta$, and $1/\varepsilon$.
\end{remark}

\begin{proof}[Proof of~\cref{thm:parameter-shift-gradient-cost}]
    We first control the random times. For $R\geq0$, conditioning on $t$ and expanding $1/\sinh x=2\sum_{k\geq0}e^{-(2k+1)x}$ gives
    \begin{align}
        p_R:=\PP(|\xi|>R)
        &=\int_R^\infty q_\beta(t)\left(1-\frac Rt\right)\ud t
          =\frac8{\pi^2}\sum_{k=0}^\infty
             \frac{e^{-(2k+1)\pi R/\beta}}{(2k+1)^2}
          \leq e^{-\pi R/\beta},\label{eqn:gradient-sampling-time-tail}\\
        \mu_1:=\mathbb E|\xi|
        &=\int_0^\infty p_R\ud R
          =\frac{7\zeta(3)}{\pi^3}\beta=c_U\beta,\label{eqn:gradient-sampling-time-moment}\\
        r_R:=\mathbb E(|\xi|\mathbf1_{\{|\xi|>R\}})
        &=Rp_R+\int_R^\infty p_v\ud v
          \leq\left(R+\frac\beta\pi\right)e^{-\pi R/\beta}.
          \label{eqn:gradient-sampling-weighted-tail}
    \end{align}
    All exchanges of integrals and sums above follow from nonnegativity; the identity at $R=0$ also verifies normalization. Unitary conjugation preserves operator norm, so the score expectation is absolutely convergent. The Duhamel integrand for its derivative is bounded in norm by $2|\xi|\|\partial_aH_\theta\|$, since $\|P_j\|=1$. The finite first moment therefore justifies differentiation under the expectation, locally uniformly in $\theta$, and all subsequent operator integrals and products are absolutely convergent.

    Next, expand the three expectations in
    \cref{eqn:probablistic_S_a,eqn:probablistic_T_a_1,eqn:probablistic_T_a_2_final} using the chosen Pauli decompositions. Every resulting branch is unitary: the response branches are conjugations
    \begin{equation}
        W_{j,\eta}(u,s;V)
        :=C_{j,\eta}(u,s)V C_{j,\eta}^\dagger(u,s),
        \qquad
        C_{j,\eta}(u,s):=e^{\I(1-s)uH_\theta}
            e^{\eta\I\pi P_j/4}e^{\I suH_\theta}.
        \label{eqn:gradient-sampling-shifted-branch}
    \end{equation}
    The total absolute coefficient masses of these displayed decompositions are
    \begin{equation}
        s_a:=\beta \mathfrak b_a,\qquad
        t_{a,j}:=\beta \mathfrak c_{a,j}+2\beta \mathfrak b_a\mu_1.
        \label{eqn:gradient-sampling-score-response-masses}
    \end{equation}
    In particular, the parameter-shift term contributes $2\beta \mathfrak b_a\mu_1$: the fair-coin and $s$ distributions each have unit mass, while the factor $\xi\eta$ contributes $\mathbb E|\xi|$. Unitary conjugation changes neither the coefficients nor their absolute mass.
    In the anticommutator in \cref{eqn:exact-gradient-observable-final}, use independent random times for the $S_a$ and $T_{a,j}$ factors. The two product orders, each with coefficient $1/2$, have total mass $s_at_{a,j}$. Similarly, the two orders in $-\I[A_a,T_{a,j}]$ have total mass $2\mathfrak a_at_{a,j}$. Thus the uncombined gradient branches have total mass $L_j$, and $\|\mathcal G_j\|\leq L_j$. If $\varepsilon\geq L_j$, returning zero proves the claim, so assume henceforth $0<\varepsilon<L_j$.

    To bound every circuit's evolution time, truncate each occurrence of $\xi$ to $|\xi|\leq R$. The resulting operators $S_{a,R}$ and $T_{a,j,R}$ retain the indicators inside the expectations, without renormalizing them. Their branch masses are
    \begin{equation}
        s_{a,R}=\beta \mathfrak b_a(1-p_R),\qquad
        t_{a,j,R}=\beta \mathfrak c_{a,j}(1-p_R)+2\beta \mathfrak b_a(\mu_1-r_R).
        \label{eqn:gradient-sampling-truncated-factor-masses}
    \end{equation}
    Their omitted masses, which also bound their operator-norm errors, are
    \begin{equation}
        d_{S,a}=s_ap_R,\qquad
        d_{T,a,j}=\beta \mathfrak c_{a,j}p_R+2\beta \mathfrak b_ar_R.
        \label{eqn:gradient-sampling-omitted-masses}
    \end{equation}
    Let $\mathcal G_{j,R}$ be obtained by replacing $S_a,T_{a,j}$ by these truncated operators in \cref{eqn:exact-gradient-observable-final}. Using
    $ST-S_RT_R=(S-S_R)T+S_R(T-T_R)$ for each product gives
    \begin{align}
        \|\mathcal G_j-\mathcal G_{j,R}\|
        &\leq\sum_a\left(d_{S,a}t_{a,j}
                     +(s_{a,R}+2\mathfrak a_a)d_{T,a,j}\right)\nonumber\\
        &\leq L_j\left(p_R+\max\left(p_R,\frac{r_R}{\mu_1}\right)\right)
         \leq 6L_j e^{-\pi R/(2\beta)}.
        \label{eqn:gradient-sampling-truncation-bias}
    \end{align}
    For the last step, put $x=\pi R/\beta$ and use
    $1/(\pi c_U)<2$ to bound the bracket by $3(1+x)e^{-x}\leq6e^{-x/2}$.
    Consequently, the explicit cutoff
    \begin{equation}
        R:=\frac{2\beta}{\pi}\log\left(2+\frac{12L_j}{\varepsilon}\right)
        =\mathcal{O}\!\left(\beta\log\left(2+\frac{L_j}{\varepsilon}\right)\right)
        \label{eqn:gradient-sampling-cutoff}
    \end{equation}
    ensures $|\Tr(\sigma(\mathcal G_j-\mathcal G_{j,R}))|\leq\varepsilon/2$.

    We now importance-sample the retained unitary branches. Write
    \begin{equation}
        \mathcal G_{j,R}=\int c_j(z)V_j(z)\ud\nu_j(z),\qquad
        L_j(R):=\int|c_j(z)|\ud\nu_j(z)
             =\sum_a(s_{a,R}+2\mathfrak a_a)t_{a,j,R}\leq L_j,
        \label{eqn:gradient-sampling-truncated-branches}
    \end{equation}
    where $\nu_j$ is a reference measure on the branch labels $z$, which record the frame index, Pauli terms, product order, response type, random times, insertion time $s$, and coin $\eta$. Every displayed branch is retained separately, so the mass identity holds even when different branches cancel as operators. The finite choices can be sampled proportionally to their absolute coefficients. For a score or explicit derivative factor, the time law is that of $\xi$ conditioned on $|\xi|\leq R$. For a parameter-shift response factor, it is instead the law proportional to $|\xi|\mathbf1_{\{|\xi|\leq R\}}$, with normalizing factor $\mu_1-r_R$; $s$ and $\eta$ retain their uniform laws. These normalizing factors are included in \cref{eqn:gradient-sampling-truncated-factor-masses}, and the sign of $\xi\eta$ is included in the branch coefficient. The times belonging to distinct factors are sampled independently conditional on the finite branch choices. This is precisely the sampling procedure specified in \cref{tab:gradient-elementary-samplers,tab:gradient-product-samplers} and used in \cref{alg:randomized-gradient-coordinate}.

    If $L_j(R)=0$, return zero, with error at most $\varepsilon/2$ by the truncation bound. Otherwise sample $z$ from $|c_j(z)|\ud\nu_j(z)/L_j(R)$ and perform a real-part Hadamard test for $e^{\I\arg c_j(z)}V_j(z)$ on one fresh copy of $\sigma$. Its outcome $Y\in\{-1,1\}$ has conditional mean
    $\Re\Tr(\sigma e^{\I\arg c_j(z)}V_j(z))$. Thus $Z:=L_j(R)Y$ satisfies
    \begin{equation}
        \mathbb E Z=\Tr(\sigma\mathcal G_{j,R})=:g_{j,R},
        \qquad |Z|\leq L_j(R).
        \label{eqn:gradient-sampling-single-shot}
    \end{equation}
    Here $\mathcal G_{j,R}$ is Hermitian, so its expectation is real. Take
    \begin{equation}
        N:=\left\lceil\frac{8L_j^2(R)}{\varepsilon^2}
                         \log\frac2\delta\right\rceil
        \label{eqn:gradient-sampling-shot-count}
    \end{equation}
    independent shots and return their sample mean $\widehat g_j$. Hoeffding's inequality gives
    \begin{equation}
        \PP\left(|\widehat g_j-g_{j,R}|>\frac\varepsilon2\right)
        \leq2\exp\left(-\frac{N\varepsilon^2}{8L_j^2(R)}\right)
        \leq\delta.
    \end{equation}
    Combining this with the deterministic bias proves the claimed error guarantee. Each shot consumes one copy of $\sigma$, and the ceiling in $N$ is absorbed by the stated bound because $\varepsilon<L_j$.

    Finally, the Hadamard tests require only the time evolutions and controlled Pauli insertions stated in the theorem. For any conjugated insertion $CVC^\dagger$, its controlled implementation is
    \begin{equation}
        \operatorname{ctrl}(CVC^\dagger)
        =(I\otimes C)\operatorname{ctrl}(V)(I\otimes C^\dagger).
    \end{equation}
    This applies both to $\tau_u^{H_\theta}(V)$ and to the parameter-shift branch in \cref{eqn:gradient-sampling-shifted-branch}. Each costs $2|u|$ in total absolute Hamiltonian-evolution time; splitting $u$ into $su$ and $(1-s)u$ does not increase this sum. A product branch contains at most one score factor and one response factor. It therefore costs at most $2(|u_1|+|u_2|)\leq4R$, with at most six Hamiltonian-evolution segments, two controlled Pauli insertions, and two parameter-shift rotations. Each individual evolution segment has duration at most $R$. Hence $T_{\max}\leq4R$ and $T_{\rm tot}\leq4RN$, which yield the two evolution-time bounds. The final scaling of $L_j$ follows by expanding \cref{eqn:gradient-sampling-coordinate-mass} under the stated uniform bounds.
\end{proof}

\section{High-temperature Hamiltonian learning}

Throughout this appendix, use the model and single-qubit frame of
\cref{sec:high-temperature-learning}: $H(\theta)=\sum_{j=1}^m\theta_jP_j$,
$\Theta=[-1,1]^m$, and $\{A_a:a\in\cA\}=\{X_x,Z_x:x\in[n]\}$.
The $P_j$ are distinct nonidentity Pauli strings of weight at most $k$,
and each support $X_j:=\operatorname{supp}(P_j)$ intersects at most $D$
other term supports.  Write $Y_a:=\operatorname{supp}(A_a)$.
Unless explicitly stated otherwise, we use the SLD notation $S_a$,
$T_{a,j}=\partial_{\theta_j}S_a$, and
$U_{a,ij}=\partial_{\theta_i}\partial_{\theta_j}S_a$.
The target is $\sigma=\rho(\theta^*)$, and
$\Delta S_a(\theta):=S_a(\theta)-S_a(\theta^*)$.

\begin{definition}[Interaction distance]\label{defn:interaction-dist}
    The fixed family $\{P_j\}_{j=1}^m$ defines the undirected
    \textit{site interaction graph} $\mathfrak G=(\mathcal V,\mathcal E)$,
    with $\mathcal V=[n]$: distinct sites are adjacent if they occur in
    a common support $X_j$.  All model terms are included, even when
    their coefficients vanish at the current parameter $\theta$.
    Let $d(x,y)$ be its graph distance, with $d(x,y)=\infty$ for sites
    in different connected components.  For nonempty subsets $X,Y$,
    and integers $r\geq0$, set
    \begin{align}
        d(X,Y)&:=\min_{x\in X,y\in Y}d(x,y),
        &
        B_r(X)&:=\{y\in\mathcal V:d(y,X)\leq r\}.
    \end{align}
    We abbreviate $d(x,X):=d(\{x\},X)$ and $B_r(x):=B_r(\{x\})$.
\end{definition}

\subsection{Lieb--Robinson bounds}

\begin{lemma}[Lieb-Robinson bound for low-intersection Hamiltonians]
    \label{lem:low-intersection-lieb-robinson}
    Let $H=\sum_jH_j$ be an $n$-qubit Hamiltonian, where $H_j$ is
    supported on a prescribed set $\Lambda_j$ of at most $k$ qubits
    and $\|H_j\|\leq1$.  Assume that each $\Lambda_j$ intersects at
    most $D$ other interaction supports.  Let $d(X,Y)$ be the distance
    in the site graph whose edges join sites in a common $\Lambda_j$.
    Define $\tau_t^H(O):=e^{\I tH}Oe^{-\I tH}$ and
    $v_\mu:=2(D+1)e^\mu$. For every $\mu>0$, $s,t\in\RR$, and
    operators $O_X,B_Y$ supported on nonempty sets $X,Y$, respectively,
    \begin{align}
        \|[\tau_t^H(O_X),\tau_s^H(B_Y)]\|_\infty
        \leq
        2|X|\|O_X\|_\infty\|B_Y\|_\infty
        e^{v_\mu|t-s|-\mu d(X,Y)}.
        \label{eqn:low-intersection-LR}
    \end{align}
    For $H=H(\theta)$, the bound is uniform in $n,m$ and
    $\theta\in\Theta$. If the support sizes are bounded by $s_0$,
    the prefactor $2|X|$ may be replaced by $C_{\mathrm{LR}}=2s_0$.
\end{lemma}
\begin{proof}
    We give the interaction-chain proof explicitly, specializing the standard Lieb--Robinson argument to low-intersection Hamiltonians~\cite{nachtergaele2006propagation,nachtergaele2019quasi}.
    The interaction graph on qubits joins two distinct sites whenever
    they belong to a common $\Lambda_j$; throughout the proof,
    $d(X,Y)$ is the induced graph distance.  Fix $B_Y$, and for a set
    of sites $S$ define
    \begin{align}
        C_{B_Y}(S,t)
        :=
        \sup_{\substack{O\neq0\\ \operatorname{supp}(O)\subseteq S}}
        \frac{\|[\tau_t^H(O),B_Y]\|}
        {\|O\|}.
        \label{eqn:LR-commutator-profile}
    \end{align}
    Also put
    \begin{align}
        H_S:=\sum_{j:\Lambda_j\cap S\neq\varnothing}H_j.
    \end{align}
    If $O$ is supported on $S$, then
    $[H-H_S,O]=0$.  Therefore, for
    $f(t):=[\tau_t^H(O),B_Y]$, differentiation and the Jacobi identity
    give an equation of the form
    \begin{align}
        f'(t)
        =
        \I[\tau_t^H(H_S),f(t)]
        +\I[[\tau_t^H(H_S),B_Y],\tau_t^H(O)].
        \label{eqn:LR-differential-equation}
    \end{align}
    The first term generates norm-preserving unitary conjugations.
    Variation of constants, the triangle inequality, and
    $\|[P,Q]\|\leq2\|P\|\|Q\|$ therefore imply
    \begin{align}
        C_{B_Y}(S,t)
        \leq
        C_{B_Y}(S,0)
        +2\sum_{j:\Lambda_j\cap S\neq\varnothing}
        \|H_j\|
        \int_0^{|t|}C_{B_Y}(\Lambda_j,s)\,\ud s.
        \label{eqn:LR-recursion}
    \end{align}
    At time zero, operators with disjoint supports commute, so
    \begin{align}
        C_{B_Y}(S,0)
        \leq
        2\|B_Y\|
        \mathbf{1}_{\{S\cap Y\neq\varnothing\}}.
    \end{align}

    Iterating~\cref{eqn:LR-recursion} yields the absolutely convergent
    interaction-chain expansion
    \begin{align}
        C_{B_Y}(X,t)
        \leq
        2\|B_Y\|
        \sum_{p=0}^\infty
        \frac{(2|t|)^p}{p!}\,a_p(X,Y),
        \label{eqn:LR-iteration}
    \end{align}
    where
    $a_0(X,Y):=\mathbf{1}_{\{X\cap Y\neq\varnothing\}}$, and for
    $p\geq1$,
    \begin{align}
        a_p(X,Y)
        :=
        \sum_{\substack{
        j_1,\ldots,j_p:\;
        \Lambda_{j_1}\cap X\neq\varnothing,\;
        \Lambda_{j_\ell}\cap\Lambda_{j_{\ell+1}}\neq\varnothing
        \ (1\leq\ell<p),\\
        \Lambda_{j_p}\cap Y\neq\varnothing}}
        \prod_{\ell=1}^p\|H_{j_\ell}\|.
        \label{eqn:LR-chain-coefficient}
    \end{align}
    Thus only chains of pairwise-overlapping interaction supports that
    connect $X$ to $Y$ contribute.

    Write $r:=d(X,Y)$.  Any contributing chain of length $p$ produces
    a path of at most $p$ edges in the interaction graph from $X$ to
    $Y$; hence $a_p(X,Y)=0$ for $p<r$.  The first interaction support
    in a chain has at most $|X|(D+1)$ choices.  Indeed, all local terms
    containing a fixed site overlap one another, so there can be at
    most $D+1$ of them.  Once $j_\ell$ is fixed, the next index has at
    most $D+1$ choices: $j_{\ell+1}=j_\ell$ or one of the at most $D$
    other terms overlapping $H_{j_\ell}$.  Using also $\|H_j\|\leq1$,
    we obtain, for $p\geq r$,
    \begin{align}
        a_p(X,Y)
        \leq
        |X|(D+1)^p.
        \label{eqn:LR-chain-count}
    \end{align}
    The same final bound holds for $p=0$ when $r=0$.  Consequently,
    \begin{align}
        C_{B_Y}(X,t)
        &\leq
        2|X|\|B_Y\|
        \sum_{p\geq r}
        \frac{(2(D+1)|t|)^p}{p!}
        \nonumber \leq
        2|X|\|B_Y\|
        e^{-\mu r}
        \sum_{p=0}^\infty
        \frac{(2(D+1)e^\mu|t|)^p}{p!}
        \nonumber\\
        &=
        2|X|\|B_Y\|
        e^{2(D+1)e^\mu|t|-\mu r}.
        \label{eqn:LR-exponential-tilt}
    \end{align}
    Multiplying by $\|O_X\|$ and recalling $r=d(X,Y)$ proves
    \cref{eqn:low-intersection-LR} for $s=0$. If $d(X,Y)=\infty$, there is no
    interaction chain connecting the two supports, and every
    coefficient in~\cref{eqn:LR-iteration} vanishes.
    Finally, unitary invariance gives
    \begin{align*}
        \|[\tau_t^H(O_X),\tau_s^H(B_Y)]\|_\infty
        =\|[\tau_{t-s}^H(O_X),B_Y]\|_\infty,
    \end{align*}
    which proves \cref{eqn:low-intersection-LR} for arbitrary $s,t$.
\end{proof}

\begin{lemma}[Three-observable tree bound]
\label{lem:appendix-triple-LR}
    Fix $\mu_0>0$, put $\nu=\mu_0/2$, and let
    $v_{\mu_0}=2(D+1)e^{\mu_0}$ be the velocity in
    \cref{lem:low-intersection-lieb-robinson}.  Define
\begin{align}
    \mathfrak F_\nu(X,Y,Z)
    :=
    e^{-\nu(d(X,Y)+d(Y,Z))}
    +e^{-\nu(d(X,Z)+d(Z,Y))}
    +e^{-\nu(d(Y,X)+d(X,Z))}.
    \label{eqn:appendix-tree-weight}
\end{align}
    For local observables supported on at most $s_0$ sites, there is
    a constant $C$, depending only on $s_0,\mu_0,k,D$, such that
    \begin{align}
        \|[\tau_{t_1}^H(O_X),
        [\tau_{t_2}^H(V_Y),\tau_{t_3}^H(W_Z)]]\|_\infty
        \leq
        C e^{v_{\mu_0}(|t_1|+|t_2|+|t_3|)}
        \mathfrak F_\nu(X,Y,Z)
        \prod_{L\in\{O,V,W\}}\|L\|_\infty.
        \label{eqn:appendix-triple-LR}
    \end{align}
\end{lemma}

\begin{proof}
    Write $A:=\tau_{t_1}^H(O_X)$, $B:=\tau_{t_2}^H(V_Y)$, and $C:=\tau_{t_3}^H(W_Z)$. Applying \cref{lem:low-intersection-lieb-robinson} to the inner commutator gives
    \begin{align*}
        \|[A,[B,C]]\|_\infty
        \leq
        2\|A\|_\infty\|[B,C]\|_\infty
        \leq
        C e^{v_{\mu_0}|t_2-t_3|}
        e^{-\mu_0d(Y,Z)}
        \prod_{L\in\{O,V,W\}}\|L\|_\infty.
    \end{align*}
    This bound captures the separation of $Y$ and $Z$ but not that of $X$ from the inner commutator. The Jacobi identity
    \begin{align*}
        [A,[B,C]]
        =
        [[A,B],C]+[B,[A,C]]
    \end{align*}
    and a second application of \cref{lem:low-intersection-lieb-robinson} give the independent bound
    \begin{align*}
        \|[A,[B,C]]\|_\infty
        \leq
        C\prod_{L\in\{O,V,W\}}\|L\|_\infty
        \left(
        e^{v_{\mu_0}|t_1-t_2|-\mu_0d(X,Y)}
        +
        e^{v_{\mu_0}|t_1-t_3|-\mu_0d(X,Z)}
        \right).
    \end{align*}
    Set $S:=|t_1|+|t_2|+|t_3|$, $x:=e^{-\mu_0d(Y,Z)}$, $y:=e^{-\mu_0d(X,Y)}$, and $z:=e^{-\mu_0d(X,Z)}$. Since $|t_i-t_j|\leq S$, the two estimates above imply
    \begin{align*}
        \|[A,[B,C]]\|_\infty
        \leq
        C e^{v_{\mu_0}S}
        \prod_{L\in\{O,V,W\}}\|L\|_\infty
        \min(x,y+z).
    \end{align*}
    Finally,
    \begin{align*}
        \min(x,y+z)\leq\sqrt{xy}+\sqrt{xz}
    \end{align*}
    produces the tree weights $e^{-\nu(d(X,Y)+d(Y,Z))}$ and $e^{-\nu(d(X,Z)+d(Z,Y))}$ with $\nu=\mu_0/2$. Adding the third nonnegative tree weight $e^{-\nu(d(Y,X)+d(X,Z))}$ gives the symmetric function $\mathfrak F_\nu$ and proves the claim.
\end{proof}

\subsection{Proof of~\Cref{thm:uniform-local-pauli-high-temperature}}
\label{append:local-pauli-frame-high-temperature}

\begin{proof}[Proof of~\Cref{thm:uniform-local-pauli-high-temperature}]
    We prove the scalar-commutant condition for every
    $\lambda\in[0,1]$, with $K^\sigma_{\lambda,a}$ as defined in
    \cref{eqn:K-sigma-lambda-a}; at $\lambda=1/2$, this is precisely
    $K^\sigma_a$ from~\cref{eqn:K-sigma-sld}.
    Index the single-qubit frame by $\cA=[n]\times\{X,Z\}$, with
    $A_{(x,X)}=X_x$, $A_{(x,Z)}=Z_x$, and
    $Y_a:=\operatorname{supp}(A_a)=\{x\}$ for $a=(x,p)$.
    Note that $\mathbf{\Delta}_{\sigma} = \exp(-\beta \mathbf{ad}_{H^*})$. We write $q := 2\lambda -1 \in [-1,1]$, and
    \begin{align}
        K^\sigma_{\lambda,a}= \left(\cosh(\frac{\beta}{2}\mathbf{ad}_{H^*}) + q\sinh(\frac{\beta}{2}\mathbf{ad}_{H^*})\right)^{-1}(A_a).
    \end{align}
    Conjugating $K^\sigma_{\lambda,a}$ by $e^{q \beta H^*/2}$, we define
    \begin{align}
        \widehat{K}_{\lambda,a} := e^{q \beta H^*/2} K^\sigma_{\lambda,a} e^{-q \beta H^*/2} = g_q\left(\frac{\beta}{2} \mathbf{ad}_{H^*}\right)(A_a),
    \end{align}
    where
    \begin{align}
        g_q(y) := \frac{\exp(qy)}{\cosh(y)+q\sinh(y)}.
    \end{align}
    The similarity transformation preserves whether the joint commutant
    is $\mathbb{C}I$. For $q=\pm1$, $g_q(y)=1$ identically, so
    $\widehat K_{\lambda,a}=A_a$ and the joint commutant is already
    $\CC I$ at every temperature. It remains to treat $|q|<1$.
    Moreover, $g_q(0)=1$. Our goal is to show that
    $g_q(\frac{\beta}{2}\mathbf{ad}_{H^*})(A_a)-A_a$ is uniformly small
    in the high-temperature regime.

    We define the residual operator
    \begin{align}
        R_a := \widehat{K}_{\lambda,a} - A_a.
    \end{align}
    Set $s = \operatorname{arctanh}(q)$. With the convention
    $g_q(y)=\int_{\RR}\varphi(t)e^{\I ty}\,\ud t$, its inverse Fourier transform is
    \begin{align}
        \varphi(t) = \frac{c_q e^{\I st}}{2\cos(\frac{\pi}{2}(q-\I t))},\qquad c_q := \cosh(s)e^{-qs} \le 1.
    \end{align}
    The weighted operator Fourier transform gives
    \begin{align}
        \widehat{K}_{\lambda,a} = \int_{\R} \varphi(t) \tau^{\frac{\beta}{2}H^*}_{t}(A_a)~\d t,\qquad \text{where}\quad \tau_t^{G}(A) := e^{\I Gt}A e^{-\I Gt}.
    \end{align}
    For each $|q|<1$, the kernel $\varphi$ is integrable. Using
    $g_q(0)=1$, we have
    \begin{align}
        R_a &= \int_{\R} \varphi(t)\left(\tau^{\frac{\beta}{2}H^*}_{t}(A_a) - A_a\right)~\d t\\
        &= \frac{\I\beta}{2}\int_{\mathbb{R}} \varphi(t)~\int^t_0 \tau^{\frac{\beta}{2}H^*}_u \left([H^*, A_a]\right)~\d u~\d t,
    \end{align}
    where the last step uses Duhamel's formula.
    Use the site interaction graph of the fixed Pauli family from
    \cref{defn:interaction-dist}. At most $D+1$ terms contain any site,
    so this graph has maximum degree at most $(D+1)(k-1)$. Thus
    \begin{align}\label{eq:Br-size-estimate}
        |B_r(Y_a)| \le b^r,\quad b = \max(2, 1 + (D+1)(k-1)).
    \end{align}
    
    For a set of sites $W\subseteq\mathcal V$, let $\Pi_W$ be normalized partial trace over
    $W^c$, tensored with the identity on $W^c$.  Single-site Haar averaging
    and a telescoping product give
    \begin{align}
        \|O-\Pi_W(O)\|_\infty
        \leq
        \sum_{x\notin W}
        \sup_{U_x\text{ unitary}}\|[O,U_x]\|_\infty.
        \label{eqn:appendix-conditional-expectation}
    \end{align}
    Indeed, single-site Haar averages $\mathbb E_x$ commute and satisfy
    $\Pi_W=\prod_{x\notin W}\mathbb E_x$.  Telescoping this product reduces
    the left-hand side to
    $\sum_{x\notin W}\|O-\mathbb E_x(O)\|_\infty$, and each summand is at most
    $\sup_{U_x}\|[O,U_x]\|_\infty$.

    Define $E_{a,0} := \Pi_{B_0(Y_a)} (R_a)$ and for $r \ge 1$,
    \begin{align}
        E_{a,r} := \Pi_{B_r(Y_a)} (R_a) - \Pi_{B_{r-1}(Y_a)}(R_a).
    \end{align}
    These operators telescope to $R_a$, which is supported in the
    connected component containing $Y_a$. We next derive the required
    shell bound. At most $D+1$ interaction terms meet $Y_a$, so
    \begin{align}
        \operatorname{supp}([H^*,A_a])\subseteq B_1(Y_a),
        \qquad
        \|[H^*,A_a]\|\leq2(D+1).
    \end{align}
    Choose $\mu=4\log b$ in \Cref{lem:low-intersection-lieb-robinson}.  Since
    $|B_1(Y_a)|\leq b$ and
    $d(B_1(Y_a),y)\geq d(Y_a,y)-1$, it gives, for operators $U_y$
    supported on the single site $y$,
    \begin{align}
        \sup_{\|U_y\|\leq1}
        \bigl\|[\tau_u^{\frac{\beta}{2}H^*}([H^*,A_a]),U_y]\bigr\|
        \leq
        4(D+1)b^5
        e^{(D+1)b^4\beta|u|}b^{-4d(Y_a,y)}.
    \end{align}
    The explicit formula for $\varphi$ implies, uniformly in $|q|<1$, $|\varphi(t)| \leq\frac{1}{2|\sinh(\pi t/2)|}$, and
    \begin{align}
        \int_{\R}|t|e^{\pi|t|/4}|\varphi(t)|\,\d t \leq
        \int_0^\infty\frac{t e^{\pi t/4}}{\sinh(\pi t/2)}\,\d t <4.
    \end{align}
    Hence, whenever
    $\beta\leq\pi/(4(D+1)b^4)$, Duhamel's formula yields
    \begin{align}
        \sup_{\|U_y\|\leq1}\|[R_a,U_y]\|
        &\leq
        2(D+1)b^5\beta b^{-4d(Y_a,y)}
        \int_{\R}|t|e^{\pi|t|/4}|\varphi(t)|\,\d t
        \nonumber\\
        &\leq8(D+1)b^5\beta b^{-4d(Y_a,y)}.
        \label{eqn:uniform-pauli-residual-LR}
    \end{align}
    For $r\geq1$, contractivity and the conditional-expectation
    inequality now give
    \begin{align}\label{eq:expansion-E-ar}
        \|E_{a,r}\|
        &\leq \|R_a-\Pi_{B_{r-1}(Y_a)}(R_a)\|
        \nonumber\\
        &\leq
        \sum_{y\notin B_{r-1}(Y_a)}
        \sup_{\|U_y\|\leq1}\|[R_a,U_y]\|
        \nonumber\\
        &\leq
        8(D+1)b^5\beta
        \sum_{\ell\geq r}b^\ell b^{-4\ell}
        =\frac{8(D+1)b^5}{1-b^{-3}}\,\beta b^{-3r}.
    \end{align}
    The same Duhamel estimate without the commutator gives
    $\|E_{a,0}\|\leq\|R_a\|\leq4(D+1)\beta$.  Thus, for every
    $r\geq0$,
    \begin{align}
        \|E_{a,r}\| \le C_{\rm LR} \beta b^{-3r},
        \qquad
        C_{\rm LR} := \frac{8(D+1)b^5}{1-b^{-3}}.
    \end{align}
    We now compare the perturbation with the unperturbed frame. Define
    \begin{align}
        \Phi_0(X)
        :=\left([X,A_a]\right)_{a\in\cA},\qquad  \widehat\Phi_{\beta,\lambda}(X) :=\left([X,\widehat K_{\lambda,a}]\right)_{a\in\cA}.
    \end{align}
    Equip the output space with the direct-sum Hilbert--Schmidt norm
    $\|(V_a)_a\|_{2,\oplus}^2:=\sum_a\|V_a\|_2^2$, where
    $\|V\|_2^2=\Tr(V^\dagger V)$.
    For fixed $r$, a site $x$ lies in $B_r(Y_a)$ for at most $2b^r$ frame labels $a$ (here we count both $X$ and $Z$ on a single site). Applying~\Cref{lem:local-pauli-dirichlet-bound} and then Minkowski's inequality gives
    \begin{align}
        \|(
        \widehat\Phi_{\beta,\lambda}-\Phi_0)(X)\|_{2,\oplus}
        &\leq \sum_{r\geq0}
        C_{\mathrm{LR}} \beta \,b^{-3r} \sqrt{2b^r}\,\|\Phi_0(X)\|_{2,\oplus}
        \nonumber\\
        &=\widetilde{C}_{\mathrm{LR}}\beta \,\|\Phi_0(X)\|_{2,\oplus},
        \label{eqn:uniform-pauli-relative-bound}
    \end{align}
    where $\widetilde{C}_{\rm LR} := \sqrt{2}C_{\rm LR}/(1-b^{-5/2})$.
    The non-identity Pauli strings diagonalize $\Phi_0^\dagger\Phi_0$, whose smallest nonzero eigenvalue is $4$. Thus every traceless $X$ satisfies
    \begin{align}
        \|\Phi_0(X)\|_{2,\oplus} \geq2\|X\|_2.
        \label{eqn:uniform-pauli-infinite-temperature-gap}
    \end{align}
    Recall that $b\geq2$, so
    $(1-b^{-3})(1-b^{-5/2})\geq1/2$.  The choice
    \begin{align}
        \beta< \beta_0 := \frac{1}{32\sqrt{2}(D+1)b^5} \leq \min\left(
        \frac{1}{2\widetilde{C}_{\rm LR}},
        \frac{\pi}{4(D+1)b^4}
        \right)
    \end{align}
    therefore guarantees both the Lieb--Robinson estimate above and the perturbative condition.  Then
    \cref{eqn:uniform-pauli-relative-bound,eqn:uniform-pauli-infinite-temperature-gap}
    yields
    \begin{align}
        \|\widehat\Phi_{\beta,\lambda}(X)\|_{2,\oplus} \geq 2\left(1-\widetilde{C}_{\rm LR} \beta \right)\|X\|_2 \geq\|X\|_2.
    \end{align}
    Hence the common commutant of the $\widehat K_{\lambda,a}$ contains no nonzero traceless operator and is $\CC I$.  
    It follows from the similarity transformation that
    \begin{align}
        \{K^\sigma_{\lambda,a}:a\in\cA\}' =\CC I.
        \label{eqn:uniform-pauli-scalar-commutant}
    \end{align}
    Taking $\lambda=1/2$ and applying
    \cref{lem:global-identifiability-sld} proves the claimed equality
    $\rho(\widehat\theta)=\sigma$ for every global minimizer of $J_Q$.
    The general $\lambda$-score conclusion follows from
    \cref{thm:global-identifiability}.
\end{proof}

\begin{lemma}[Local Pauli Dirichlet bound]
    \label{lem:local-pauli-dirichlet-bound}
    Let $S\subseteq[n]$, let $B$ be an operator supported on $S$, and
    let $X$ be an arbitrary $n$-qubit operator.  Then
    \begin{align}
        \|[X,B]\|_2^2
        \leq
        \|B\|_\infty^2
        \sum_{x\in S}
        \left(
        \|[X,X_x]\|_2^2+\|[X,Z_x]\|_2^2
        \right).
        \label{eqn:uniform-pauli-local-Dirichlet-bound}
    \end{align}
\end{lemma}
\begin{proof}
    The claim is immediate when $S=\varnothing$, so assume $S$ is nonempty.
    Let
    \begin{align}
        \mathbb{E}_{S^c}(X)
        :=\frac{I_S}{2^{|S|}}\otimes\Tr_S(X)
    \end{align}
    be the Hilbert--Schmidt orthogonal projection onto the subspace of
    operators that are identity on $S$, and set
    $X_\perp:=X-\mathbb{E}_{S^c}(X)$.  Since $B$ is supported on $S$,
    $[\mathbb{E}_{S^c}(X),B]=0$.  Therefore, by the Schatten-norm
    inequality $\|MN\|_2\leq\|M\|_2\|N\|_\infty$,
    \begin{align}
        \|[X,B]\|_2
        =\|[X_\perp,B]\|_2
        \leq 2\|B\|_\infty\|X_\perp\|_2.
        \label{eqn:local-pauli-commutator-step}
    \end{align}
    It remains to estimate $\|X_\perp\|_2$.  Expand $X$ in the
    $n$-qubit Pauli basis.  The quadratic form
    \begin{align}
        \mathcal{D}_S(X)
        :=\sum_{x\in S}
        \left(
        \|[X,X_x]\|_2^2+\|[X,Z_x]\|_2^2
        \right)
    \end{align}
    is diagonal in this basis.  It vanishes on a Pauli string precisely
    when that string is identity on every site in $S$.  Otherwise,
    choose $x\in S$ at which the string has a nonidentity Pauli factor.
    If that factor is $X_x$ or $Z_x$, exactly one of its commutators with
    $X_x,Z_x$ has squared Hilbert--Schmidt norm equal to four times the
    squared norm of the string; if it is $Y_x$, both do. Consequently,
    the smallest nonzero eigenvalue of $\mathcal{D}_S$ is $4$, and hence
    \begin{align}
        \mathcal{D}_S(X)\geq4\|X_\perp\|_2^2.
        \label{eqn:local-pauli-dirichlet-gap}
    \end{align}
    Squaring~\cref{eqn:local-pauli-commutator-step} and applying
    \cref{eqn:local-pauli-dirichlet-gap} proves
    \cref{eqn:uniform-pauli-local-Dirichlet-bound}.
\end{proof}

\subsection{Shared-copy gradient estimator: proof of~\Cref{thm:shared-copy-gradient-estimation}}

\subsubsection{Quasi-locality of the gradient observable}

\begin{lemma}[Exponential summability and tree convolutions]
\label{lem:appendix-exponential-summability}
    Let $b$ be the growth constant in \cref{eq:Br-size-estimate}.
    For $c>\log b$ and nonempty supports $X$ of uniformly bounded
    size,
    \begin{align}
        \sum_a e^{-c d(Y_a,X)}
        +\sum_j e^{-c d(X_j,X)}
        \leq C_c.
        \label{eqn:appendix-label-summability}
    \end{align}
    For $\alpha>\log b$, $0<\zeta<\alpha-\log b$, and nonempty
    supports $X,Y$ of uniformly bounded size,
    \begin{align}
        \sum_x e^{-\alpha(d(x,X)+d(x,Y))}
        &\leq C e^{-\zeta d(X,Y)},
        \label{eqn:appendix-convolution}\\
        \sum_a\mathfrak F_\alpha(Y_a,X,Y)
        &\leq C e^{-\zeta d(X,Y)}.
        \label{eqn:appendix-tree-convolution}
    \end{align}
    Here $\mathfrak F_\alpha$ is the tree weight in
    \cref{eqn:appendix-tree-weight}.  Let $\ell(X,Y,Z)$ be the minimum
    of the three path lengths in that definition.  For
    $0<\eta<\alpha-\log b$,
    \begin{align}
        \sup_i\sum_{a,j}
        e^{\eta\ell(Y_a,X_i,X_j)}
        \mathfrak F_\alpha(Y_a,X_i,X_j)
        <\infty.
        \label{eqn:appendix-tree-sum}
    \end{align}
    All constants are independent of $n,m,\theta$ and depend only on
    the displayed exponents, $k,D$, and the support-size bounds.
    We set $e^{-c\infty}=0$ and omit infinite-distance terms from
    positively weighted sums.
\end{lemma}

\begin{proof}
    The ball bound gives $|B_r(X)|\leq|X|b^r$, so, for $c>\log b$,
    \begin{align*}
        \sum_x e^{-c d(x,X)}
        \leq |X|\sum_{r\geq0}(be^{-c})^r
        =\frac{|X|}{1-be^{-c}}.
    \end{align*}
    Each site carries two frame labels and belongs to at most $D+1$
    term supports.  For each $j$, choose a site in $X_j$ minimizing
    its distance to $X$.  Summing over these sites proves
    \cref{eqn:appendix-label-summability}.

    Since $d(X,Y)\leq d(x,X)+d(x,Y)$,
    \begin{align*}
        \sum_x e^{-\alpha(d(x,X)+d(x,Y))}
        &\leq e^{-\zeta d(X,Y)}
        \sum_x e^{-(\alpha-\zeta)(d(x,X)+d(x,Y))}\\
        &\leq e^{-\zeta d(X,Y)}
        \sum_x e^{-(\alpha-\zeta)d(x,X)}.
    \end{align*}
    The last sum is bounded because $\alpha-\zeta>\log b$,
    proving \cref{eqn:appendix-convolution}.  In
    $\mathfrak F_\alpha(Y_a,X,Y)$, the two trees containing the
    edge $(X,Y)$ are bounded after summing over $a$ by
    \cref{eqn:appendix-label-summability}.  The tree centered at
    $Y_a$ is bounded by twice \cref{eqn:appendix-convolution}.
    This proves \cref{eqn:appendix-tree-convolution}.

    Set $c=\alpha-\eta>\log b$.  For each tree path length $L$,
    $\ell\leq L$, so $e^{\eta\ell}e^{-\alpha L}\leq e^{-cL}$.
    The tree centered at $X_i$ therefore contributes at most
    \begin{align*}
        \left(\sum_a e^{-c d(Y_a,X_i)}\right)
        \left(\sum_j e^{-c d(X_i,X_j)}\right).
    \end{align*}
    Both factors are uniformly bounded by
    \cref{eqn:appendix-label-summability}.  For the tree centered at
    $X_j$, sum over $a$ first and then over $j$; for the tree centered
    at $Y_a$, sum over $j$ first and then over $a$.  The same bound
    applies in both cases, proving \cref{eqn:appendix-tree-sum}.
    Terms with infinite distance vanish because the corresponding
    supports lie in different connected components.
\end{proof}

\begin{lemma}[SLD response bounds]
\label{lem:appendix-SLD-response}
Use the random time $\xi$ from \cref{eqn:probablistic_S_a}: first sample
$t>0$ with density $q_\beta$ in
\cref{eqn:appendix-random-time-density}, then sample
$\xi\mid t\sim\mathrm{Uniform}[-t,t]$.  Define
\begin{align}
    \mathcal Q_H(O)&:=\mathbb E_\xi\tau_\xi^H(O),
    &
    B_a(H)&:=\partial_aH=\I[H,A_a],
    \label{eqn:appendix-response-definitions}
\end{align}
so that $S_a(\theta)=-\beta\mathcal Q_{H(\theta)}(B_a(H(\theta)))$.
We write $\mathcal Q_H^{(1)}(O;V)$ and
$\mathcal Q_H^{(2)}(O;V,W)$ for the first and second directional
Hamiltonian derivatives of $\mathcal Q_H(O)$, holding $O$ fixed.

Fix $\mu_0=7\log b$ and put $\nu=\mu_0/2$.  For local observables
whose support sizes are bounded in terms of $k,D$, there is a constant
$C$, depending only on these bounds and $k,D$, such that for
$0<\beta\leq\pi/(4v_{\mu_0})$,
\begin{align}
    \|[\mathcal Q_H(O_X),B_Z]\|_\infty
    &\leq C\|O_X\|_\infty\|B_Z\|_\infty e^{-\mu_0d(X,Z)},
    \label{eqn:appendix-Q-commutator}\\
    \|\mathcal Q_H^{(1)}(O_X;V_Y)\|_\infty
    &\leq C\beta\|O_X\|_\infty\|V_Y\|_\infty e^{-\mu_0d(X,Y)},
    \label{eqn:appendix-Q-first-bound}\\
    \|[\mathcal Q_H^{(1)}(O_X;V_Y),B_Z]\|_\infty
    &\leq C\beta\|O_X\|_\infty\|V_Y\|_\infty\|B_Z\|_\infty \mathfrak F_\nu(X,Y,Z),
    \label{eqn:appendix-Q-first-commutator}\\
    \|\mathcal Q_H^{(2)}(O_X;V_Y,W_Z)\|_\infty
    &\leq C\beta^2\|O_X\|_\infty\|V_Y\|_\infty\|W_Z\|_\infty \mathfrak F_\nu(X,Y,Z).
    \label{eqn:appendix-Q-second-bound}
\end{align}
Moreover, for any bounded local observable $O$ with uniformly bounded support,
\begin{align}
    \|\mathcal Q_H(O)-O\|_\infty\leq C\beta\|O\|_\infty.
    \label{eqn:appendix-Q-minus-I}
\end{align}
\end{lemma}

\begin{proof}
For every integer $p\geq0$ and $c\geq0$, we have
\begin{align}
    \mathbb E\bigl(|\xi|^pe^{c|\xi|}\bigr)
    \leq C_p\beta^p
    \label{eqn:appendix-kernel-moments}
\end{align}
whenever $c\beta\leq\pi/2$.  This follows directly from
\begin{align*}
    \mathbb E(|\xi|^pe^{c|\xi|})
    \leq
    4\beta^p\int_0^\infty
    \frac{x^{p+1}e^{c\beta x}}{\sinh(\pi x)}\,\ud x.
\end{align*}
Duhamel's formula gives
\begin{align}
    \mathcal Q_H^{(1)}(O;V)
    =
    \I\mathbb E\left(
        \xi\int_0^1
        \tau_{(1-s)\xi}^H
        ([V,\tau_{s\xi}^H(O)])\,\ud s
    \right).
    \label{eqn:appendix-first-response}
\end{align}
The second derivative is the sum over the two orders of $V,W$ of
\begin{align}
    -\mathbb E\left(
    \xi^2\int_{0\leq s_2\leq s_1\leq1}
    \tau_{(1-s_1)\xi}^H
    \left(
    [V,\tau_{(s_1-s_2)\xi}^H
    ([W,\tau_{s_2\xi}^H(O)])]
    \right)\ud s_2\ud s_1
    \right).
    \label{eqn:appendix-second-response}
\end{align}
For the commutator of the first derivative with $B_Z$, conjugating
the integrand by $\tau_{-(1-s)\xi}^H$ leaves three times
$0,s\xi,-(1-s)\xi$, whose absolute values sum to $|\xi|$.
For the second derivative, removing the outer conjugation leaves the
nested commutator
$[V, [\tau_{(s_1-s_2)\xi}^H(W),\tau_{s_1\xi}^H(O)]]$;
its three absolute times sum to at most $2|\xi|$.
Thus \cref{lem:low-intersection-lieb-robinson,lem:appendix-triple-LR} require kernel moments
with $c\leq2v_{\mu_0}$.  Applying these bounds inside the integrals and
using \cref{eqn:appendix-kernel-moments} yields
\begin{align*}
    \|[\mathcal Q_H(O_X),B_Z]\|_\infty
    &\leq C\|O_X\|_\infty\|B_Z\|_\infty
    e^{-\mu_0d(X,Z)},\\
    \|\mathcal Q_H^{(1)}(O_X;V_Y)\|_\infty
    &\leq C\beta\|O_X\|_\infty\|V_Y\|_\infty
    e^{-\mu_0d(X,Y)},\\
    \|[\mathcal Q_H^{(1)}(O_X;V_Y),B_Z]\|_\infty
    &\leq C\beta\|O_X\|_\infty\|V_Y\|_\infty\|B_Z\|_\infty
    \mathfrak F_\nu(X,Y,Z),\\
    \|\mathcal Q_H^{(2)}(O_X;V_Y,W_Z)\|_\infty
    &\leq C\beta^2\|O_X\|_\infty\|V_Y\|_\infty\|W_Z\|_\infty
    \mathfrak F_\nu(X,Y,Z).
\end{align*}
The assumed threshold gives $2v_{\mu_0}\beta\leq\pi/2$, so all
constants remain independent of the volume.  Finally,
\begin{align*}
    \|\mathcal Q_H(O)-O\|_\infty
    &\leq \mathbb E|\xi|\,\|[H,O]\|_\infty\\
    &\leq C\beta\,2|\operatorname{supp}(O)|(D+1)\|O\|_\infty,
\end{align*}
because only terms meeting $\operatorname{supp}(O)$ contribute to the
commutator.  The first inequality follows by integrating
$\tau_\xi^H(O)-O=\I\int_0^\xi\tau_s^H([H,O])\,\ud s$.
\end{proof}

\begin{proposition}[Uniform score-response bounds]
\label{prop:uniform-score-response}
    Put $Q_{a,j}=\partial_aP_j$ and $\mu_0=7\log b$, and choose
\begin{align}
    \nu=\frac{\mu_0}{2},
    \qquad
    \zeta=\frac{\nu+\log b}{2},
    \qquad
    \kappa=\zeta-\log b=\frac{\nu-\log b}{2}.
    \label{eqn:appendix-exponent-ledger}
\end{align}
Since $\mu_0=7\log b$, we have
$\nu=\frac72\log b$, $\zeta=\frac94\log b$, and
$\kappa=\frac54\log b>\log b$; in particular,
$2\log b<\zeta<\nu-\log b$.
In the positively weighted sums below, infinite-distance terms are
omitted; the corresponding responses vanish across disconnected
components of the interaction graph.
There is $C=\poly(k,D)>0$ such that for
$0<\beta\leq\min(1,\pi/(8(D+1)b^7))$,
\begin{align}
    T_{a,j}
    &=
    -\beta Q_{a,j}+R_{a,j},
    \label{eqn:appendix-T-leading}\\
    \sup_j\sum_a e^{\kappa d(Y_a,X_j)}\|R_{a,j}\|_\infty
    +\sup_a\sum_j e^{\kappa d(Y_a,X_j)}\|R_{a,j}\|_\infty
    &\leq C\beta^2,
    \label{eqn:appendix-R-sum}\\
    \sup_i\sum_{a,j}
    e^{\kappa\ell(Y_a,X_i,X_j)}\|U_{a,ij}\|_\infty
    &\leq C\beta^2,
    \label{eqn:appendix-U-sum}\\
    \sup_a\|S_a(\theta)-S_a(\theta^*)\|_\infty
    &\leq C\beta.
    \label{eqn:appendix-score-difference}
\end{align}
In addition, for every single-site $O_x$ with
$\|O_x\|_\infty\leq1$,
\begin{align}
    \sum_a\|[T_{a,j},O_x]\|_\infty
    &\leq C\beta e^{-\zeta d(X_j,x)},
    \label{eqn:appendix-T-commutator-sum}\\
    \|[S_a,O_x]\|_\infty
    &\leq C\beta e^{-\mu_0d(Y_a,x)},
    \label{eqn:appendix-S-commutator}\\
    \|T_{a,j}\|_\infty
    &\leq C\beta e^{-\mu_0d(Y_a,X_j)}.
    \label{eqn:appendix-T-pointwise}
\end{align}
\end{proposition}

\begin{proof}
    Differentiating \cref{eqn:probablistic_S_a} gives the exact identities
\begin{align}
    T_{a,j}
    &=
    -\beta\left(
        \mathcal Q_H(Q_{a,j})
        +\mathcal Q_H^{(1)}(B_a;P_j)
    \right),
    \label{eqn:appendix-T-formula}\\
    U_{a,ij}
    &=
    -\beta\left(
        \mathcal Q_H^{(1)}(Q_{a,j};P_i)
        +\mathcal Q_H^{(1)}(Q_{a,i};P_j)
        +\mathcal Q_H^{(2)}(B_a;P_i,P_j)
    \right).
    \label{eqn:appendix-U-formula}
\end{align}
    Only terms meeting $Y_a$ contribute to $B_a(H)$, so
    \begin{align*}
        \operatorname{supp}(B_a)&\subseteq B_1(Y_a),
        &\|B_a\|_\infty&\leq2(D+1).
    \end{align*}
    Also $Q_{a,j}=\partial_aP_j$ vanishes unless $Y_a\subseteq X_j$,
    and, whenever it is nonzero, has norm $2$ and support
    $X_j\subseteq B_1(Y_a)$.
    Replacing any of these supports by $Y_a$ in a decay estimate costs
    at most a factor $e^{\mu_0}=b^7$ per incident tree edge, since
    $d(B_1(Y_a),X)\geq d(Y_a,X)-1$.

    Subtracting the leading term from \cref{eqn:appendix-T-formula} gives
    \begin{align*}
        R_{a,j}
        =
        -\beta\left(
            \mathcal Q_H(Q_{a,j})-Q_{a,j}
            +\mathcal Q_H^{(1)}(B_a;P_j)
        \right).
    \end{align*}
    The first difference vanishes unless $Y_a$ meets $X_j$ and is then
    $\mathcal{O}(\beta)$ by \cref{eqn:appendix-Q-minus-I}.  The second term is
    $\mathcal{O}(\beta e^{-\mu_0d(Y_a,X_j)})$ by
    \cref{eqn:appendix-Q-first-bound}.  Hence
    \begin{align*}
        \|R_{a,j}\|_\infty
        \leq C\beta^2e^{-\mu_0d(Y_a,X_j)},
    \end{align*}
    and \cref{eqn:appendix-label-summability} proves both sums in \cref{eqn:appendix-R-sum}.  The same bound together with the local leading term proves \cref{eqn:appendix-T-pointwise}.

    In \cref{eqn:appendix-U-formula}, the first two terms are bounded by
    $C\beta^2\mathfrak F_\nu(Y_a,X_i,X_j)$ using
    \cref{eqn:appendix-Q-first-bound}; the final term is
    $\mathcal{O}(\beta^3\mathfrak F_\nu)$ by
    \cref{eqn:appendix-Q-second-bound}.  The tree sum
    \cref{eqn:appendix-tree-sum} proves \cref{eqn:appendix-U-sum}.
    Equation \cref{eqn:appendix-score-difference} follows immediately
    from \cref{eqn:probablistic_S_a}, because $B_a$ is uniformly
    local and bounded.

    Finally, commuting \cref{eqn:appendix-T-formula} with $O_x$ and using
    \cref{eqn:appendix-Q-commutator,eqn:appendix-Q-first-commutator,eqn:appendix-tree-convolution}
    proves \cref{eqn:appendix-T-commutator-sum}. Integrating the ordinary
    LR bound in \cref{eqn:probablistic_S_a} gives
    \cref{eqn:appendix-S-commutator}.

    The constants in these bounds can be chosen polynomially in $k,D$: the LR prefactors and the norms and support sizes of $B_a$ are polynomial, every support shift costs a fixed power of $b$, and the geometric sums have ratios bounded away from one because
    $b\geq2$ and all exponent gaps are fixed positive multiples of $\log b$. The moment constants involve only $p=0,1,2$ and are
    absolute after rescaling by $\beta$.
\end{proof}

\begin{theorem}[Quasi-locality of the exact gradient]
\label{thm:high-temperature-gradient-shells}
    Suppose that every $P_j$ has weight at most $k$ and intersects at most $D$ other Hamiltonian terms.
    Define
    \begin{align}
        \beta_0 :=
        \min\left(
            1,
            \frac{\pi}{8(D+1)b^7}
        \right).
        \label{eqn:gradient-shell-beta-threshold}
    \end{align}
    There exist $C_{\mathrm{sh}}, C_{\mathrm{tail}}=\poly(k,D)$ such that, for every $0<\beta\leq\beta_0$, with $\kappa=\frac54\log b>\log b$,
    \begin{align}
        \mathcal G_j(\theta)
        =
        \sum_{r=0}^{\infty}G_{j,r}(\theta),
        \qquad
        \operatorname{supp}G_{j,r}(\theta)\subseteq B_r(X_j),
        \qquad
        \|G_{j,r}(\theta)\|_\infty
        \leq
        C_{\mathrm{sh}}\beta e^{-\kappa r}.
        \label{eqn:high-temperature-gradient-shells}
    \end{align}
    Equivalently, the order-$R$ partial sum $\mathcal G_j^{(R)}:=\sum_{r=0}^R G_{j,r}$ is strictly supported on $B_R(X_j)$ and approximates the exact gradient observable with exponentially small error,
    \begin{align}
        \operatorname{supp}\mathcal G_j^{(R)}
        &\subseteq B_R(X_j),
        &
        \|\mathcal G_j-\mathcal G_j^{(R)}\|_\infty
        &\leq
        C_{\mathrm{tail}}\beta e^{-\kappa R},
        \label{eqn:high-temperature-gradient-truncation}
    \end{align}
    The bounds are uniform in $n,m,j$, and $\theta\in\Theta$.
\end{theorem}
\begin{proof}[Proof of \cref{thm:high-temperature-gradient-shells}]
    Let $O_x$ be any single-site observable of norm at most one.
    Commuting \cref{eqn:qsm-gradient-observable} with $O_x$ gives
    \begin{align*}
        [\mathcal G_j,O_x]
        =
        \sum_a\Bigl(
        &\frac12\{S_a,[T_{a,j},O_x]\}
        +\frac12\{[S_a,O_x],T_{a,j}\}\\
        &-\I[A_a,[T_{a,j},O_x]]
        +\I[T_{a,j},[A_a,O_x]]
        \Bigr).
    \end{align*}
    The first term is $\mathcal{O}(\beta^2e^{-\zeta d(X_j,x)})$ by
    \cref{eqn:appendix-T-commutator-sum}.  The second combines
    \cref{eqn:appendix-S-commutator,eqn:appendix-T-pointwise} and the
    convolution bound \cref{eqn:appendix-convolution}.  The third is
    $\mathcal{O}(\beta e^{-\zeta d(X_j,x)})$.  In the last term,
    $[A_a,O_x]$ vanishes unless $Y_a=\{x\}$, so
    \cref{eqn:appendix-T-pointwise} applies directly.  Consequently,
    \begin{align}
        \sup_{\|O_x\|_\infty\leq1}
        \|[\mathcal G_j(\theta),O_x]\|_\infty
        \leq
        C\beta e^{-\zeta d(X_j,x)}.
        \label{eqn:appendix-gradient-commutator}
    \end{align}

    Apply \cref{eqn:appendix-conditional-expectation} with
    $W=B_r(X_j)$.  Summing \cref{eqn:appendix-gradient-commutator} over
    the complement and using \cref{eq:Br-size-estimate} gives
    \begin{align}
        \|\mathcal G_j-
        \Pi_{B_r(X_j)}(\mathcal G_j)\|_\infty
        \leq
        C\beta e^{-(\zeta-\log b)r}
        =
        C\beta e^{-\kappa r}.
    \end{align}
    Set
    $\widetilde{\mathcal G}_{j,r}=\Pi_{B_r(X_j)}(\mathcal G_j)$,
    $G_{j,0}=\widetilde{\mathcal G}_{j,0}$, and
    $G_{j,r}=\widetilde{\mathcal G}_{j,r}
    -\widetilde{\mathcal G}_{j,r-1}$ for $r\geq1$.  These Hermitian operators
    have the required supports and telescope to $\mathcal G_j$ in
    operator norm.  The preceding approximation bound gives
    $\|G_{j,r}\|_\infty\leq C\beta e^{-\kappa r}$ after changing the
    constant.  The $r=0$ bound follows from
    $\sum_a\|T_{a,j}\|_\infty=\mathcal{O}(\beta)$ and
    $\sup_a\|S_a\|_\infty=\mathcal{O}(\beta)$.
\end{proof}

\subsubsection{Algorithm and complexity analysis}
\label{append:shared_copy_grad_estimate}

\begin{lemma}[Coloring overlapping neighborhoods]
\label{lem:appendix-neighborhood-coloring}
    For fixed $r$, let $\mathscr O_r$ be the graph on the term labels $1,\ldots,m$ in which $i$ and $j$ are adjacent when $B_r(X_i)\cap B_r(X_j)\neq\varnothing$. Then
    \begin{align}
        \deg_{\mathscr O_r}(j)
        \leq
        (D+1)|B_{2r}(X_j)|
        \leq
        (D+1)k b^{2r}.
        \label{eqn:appendix-neighborhood-degree}
    \end{align}
    Consequently, $\mathscr O_r$ admits a proper coloring with
    \begin{align}
        \chi_r
        \leq
        (1+(D+1)k)b^{2r}.
        \label{eqn:appendix-neighborhood-coloring}
    \end{align}
\end{lemma}

\begin{proof}
    Fix $j$ and write $N_r(j):=\{i\neq j:B_r(X_i)\cap B_r(X_j)\neq\varnothing\}$. If $i\in N_r(j)$, choose $z\in B_r(X_i)\cap B_r(X_j)$. There are $x_i\in X_i$ and $x_j\in X_j$ with $d(x_i,z)\leq r$ and $d(z,x_j)\leq r$, hence $x_i\in X_i\cap B_{2r}(X_j)$. It follows that
    \begin{align}
        N_r(j)
        \subseteq
        \bigcup_{x\in B_{2r}(X_j)}
        \{i:x\in X_i\}.
        \label{eqn:appendix-neighborhood-cover}
    \end{align}
    For a fixed qubit $x$, all terms with $x\in X_i$ mutually overlap. If there are $L$ such terms, each has at least $L-1$ neighbors in the term-overlap graph, so $L-1\leq D$ and $L\leq D+1$. Therefore
    \begin{align}
        |N_r(j)|
        \leq
        \sum_{x\in B_{2r}(X_j)}
        \#\{i:x\in X_i\}
        \leq
        (D+1)|B_{2r}(X_j)|.
    \end{align}
    Finally, $|X_j|\leq k$ and \cref{eq:Br-size-estimate} give
    \begin{align}
        |B_{2r}(X_j)|
        \leq
        \sum_{x\in X_j}|B_{2r}(x)|
        \leq
        kb^{2r}.
    \end{align}
    This proves \cref{eqn:appendix-neighborhood-degree}.  Greedy coloring
    uses at most one more color than the maximum degree, and
    $1+(D+1)kb^{2r}\leq(1+(D+1)k)b^{2r}$, proving
    \cref{eqn:appendix-neighborhood-coloring}.
\end{proof}

\begin{algorithm}[H]
\caption{Shared-copy gradient estimator}
\label{alg:shared-copy-gradient}
\KwIn{$N_{\rm copy}$ copies of $\sigma$, a fixed parameter $\theta$, target accuracy $\tau$, failure probability $\delta$, and shell operators $G_{j,r}(\theta)$ from \cref{thm:high-temperature-gradient-shells}.}
Set $a_r=C_{\mathrm{sh}}\beta e^{-\kappa r}$ and choose $R_\tau=\mathcal{O}(1+(\log(C\beta/\tau))_+)$ so that $\sum_{r>R_\tau}a_r\leq\tau/2$\;
\For{$r=0,\ldots,R_\tau$}{
Partition the labels into $\chi_r$ color classes
$\mathcal C_{r,1},\ldots,\mathcal C_{r,\chi_r}$ so that, for every color
$c$ and all distinct $j,j'\in\mathcal C_{r,c}$,
$B_r(X_j)\cap B_r(X_{j'})=\varnothing$\;
}
Set $S_{R_\tau}:=\sum_{s=0}^{R_\tau}a_s\sqrt{\chi_s}$ and $p_r=a_r\sqrt{\chi_r}/S_{R_\tau}$\;
\For{$u=1,\ldots,N_{\rm copy}$}{
Sample one radius $r$ with probability $p_r$, then sample one color $c$ uniformly from $\{1,\ldots,\chi_r\}$\;
Initialize the $u$th record by setting $Z_j^{(u)}=0$ for all $j$\;
Simultaneously measure the tensor-product POVM $\bigotimes_{j\in\mathcal C_{r,c}}\{E_{j,r}^+,E_{j,r}^-\}$, where $E_{j,r}^{\pm}=\frac12(I\pm G_{j,r}/a_r)$, and record all outcomes $Y_{j,r}^{(u)}\in\{-1,1\}$\;
Set $Z_j^{(u)}=(a_r\chi_r/p_r)Y_{j,r}^{(u)}$ for every $j\in\mathcal C_{r,c}$\;
}
Set $L=\lceil8\log(2m/\delta)\rceil$ and $s=\lfloor N_{\rm copy}/L\rfloor$; assume $s\geq\max(1,\lceil16C_{\mathrm{var}}\beta^2/\tau^2\rceil)$\;
Partition the first $sL$ records into $L$ blocks $I_1,\ldots,I_L$ of size $s$\;
Set $\widehat g_j=\operatorname{median}_{\ell\in\{1,\ldots,L\}}\bigl(|I_\ell|^{-1}\sum_{u\in I_\ell}Z_j^{(u)}\bigr)$ for every $j$\;
\KwOut{$\widehat g=(\widehat g_1,\ldots,\widehat g_m)$.}
\end{algorithm}

\begin{proof}[Proof of~\Cref{thm:shared-copy-gradient-estimation}]
    The choice of $R_\tau$ makes the truncation bias at most $\tau/2$. By \cref{lem:appendix-neighborhood-coloring}, the labels admit a coloring with
    \begin{align}
        \chi_r
        \leq
        (1+(D+1)k)b^{2r}.
        \label{eqn:gradient-shell-coloring}
    \end{align}
    The POVMs in each color class are compatible because their supports are disjoint, and they are valid because $\|G_{j,r}\|_\infty\leq a_r$. For fixed $j$ and $r$, the color containing $j$ is selected with probability $p_r/\chi_r$, while $\mathbb E(Y_{j,r}^{(u)}\mid r,j\text{ selected})=\Tr(\sigma G_{j,r})/a_r$. Therefore each record has
    \begin{align*}
        \mathbb E Z_j^{(u)}
        =
        \sum_{r=0}^{R_\tau}
        \frac{p_r}{\chi_r}
        \frac{a_r\chi_r}{p_r}
        \frac{\Tr(\sigma G_{j,r})}{a_r}
        =
        \Tr\left(\sigma\sum_{r=0}^{R_\tau}G_{j,r}\right).
    \end{align*}
    The same selection probability gives
    \begin{align*}
        \mathbb E (Z_j^{(u)})^2
        \leq
        \sum_{r=0}^{R_\tau}
        \frac{p_r}{\chi_r}
        \left(\frac{a_r\chi_r}{p_r}\right)^2
        =
        \sum_{r=0}^{R_\tau}\frac{a_r^2\chi_r}{p_r}.
    \end{align*}
    Substituting $p_r=a_r\sqrt{\chi_r}/S_{R_\tau}$ gives
    \begin{align*}
        \sum_{r=0}^{R_\tau}
        \frac{a_r^2\chi_r}{p_r}
        &=
        \sum_{r=0}^{R_\tau}
        a_r^2\chi_r
        \frac{S_{R_\tau}}{a_r\sqrt{\chi_r}}\\
        &=
        S_{R_\tau}
        \sum_{r=0}^{R_\tau}a_r\sqrt{\chi_r}
        =
        S_{R_\tau}^2.
    \end{align*}
    The decay--growth bound then gives
    \begin{align*}
        S_{R_\tau}
        &\leq
        C_{\mathrm{sh}}\beta
        \sqrt{1+(D+1)k}
        \sum_{r\geq0}(be^{-\kappa})^r\\
        &=\mathcal{O}(\beta).
    \end{align*}
    The implicit constants in this proof depend only on $k$ and $D$.
    The records are independent across $u$. Let $C_{\mathrm{var}}$ be a constant such that $\operatorname{Var}(Z_j^{(u)})\leq C_{\mathrm{var}}\beta^2$ for every $j$. If a block $I_\ell$ contains $s$ records, its coordinatewise average
    \begin{align*}
        \overline Z_{j,\ell}
        :=
        \frac1s\sum_{u\in I_\ell}Z_j^{(u)}
    \end{align*}
    satisfies $\operatorname{Var}(\overline Z_{j,\ell})\leq C_{\mathrm{var}}\beta^2/s$. Chebyshev's inequality therefore gives
    \begin{align*}
        \PP\left(
            |\overline Z_{j,\ell}-\mathbb EZ_j^{(u)}|>\frac{\tau}{2}
        \right)
        \leq
        \frac{4C_{\mathrm{var}}\beta^2}{s\tau^2}.
    \end{align*}
    Taking $s=\max(1,\lceil16C_{\mathrm{var}}\beta^2/\tau^2\rceil)$ makes one block bad with probability at most $1/4$. The coordinatewise median can be bad only if at least half of the $L$ independent blocks are bad, so Hoeffding's inequality gives
    \begin{align*}
        \PP\left(
            |\widehat g_j-\mathbb EZ_j^{(u)}|>\frac{\tau}{2}
        \right)
        \leq
        e^{-L/8}.
    \end{align*}
    With $L=\lceil8\log(2m/\delta)\rceil$, this probability is at most $\delta/(2m)$. A union bound controls all $m$ coordinates simultaneously, while the truncation contributes the remaining error $\tau/2$. The total number of copies is
    \begin{align*}
        N_{\rm copy}=sL
        =
        \mathcal{O}\left(
            \left(1+\frac{\beta^2}{\tau^2}\right)
            \log\frac{2m}{\delta}
        \right),
    \end{align*}
    proving \cref{eqn:shared-copy-gradient-cost}.
\end{proof}

\subsection{End-to-end sample complexity}
\label{append:high-temp-sample-complexity-proof}

\begin{lemma}[Uniform local thermal bias]
\label{lem:local-thermal-bias}
    Suppose $H(\theta)=\sum_{j=1}^m\theta_jP_j$, where the known,
    distinct, nonidentity Pauli terms have weight at most $k$ and each
    overlaps at most $D$ other terms. For every $\beta\geq0$, every
    operator $O$ on at most $q$ qubits satisfies
    \begin{align}
        \left|
            \Tr(\rho(\theta)O)-2^{-n}\Tr O
        \right|
        \leq
        2q(D+1)\beta\|O\|_\infty.
        \label{eqn:appendix-local-thermal-bias}
    \end{align}
    The bound is uniform in $n,m$, and $\theta\in\Theta=[-1,1]^m$.
\end{lemma}

\begin{proof}
    Let $S=\operatorname{supp}(O)$ and separate the terms meeting $S$:
    \begin{align*}
        V:=\sum_{j:\operatorname{supp}(P_j)\cap S\neq\varnothing}
            \theta_jP_j,\qquad
        H_0:=H(\theta)-V.
    \end{align*}
    At most $D+1$ terms contain any one site, since such terms all
    overlap one another. Therefore $\|V\|_\infty\leq q(D+1)$.
    As $H_0$ acts trivially on $S$, its Gibbs state has maximally mixed
    marginal on $S$.

    Interpolate between $H_0$ and $H(\theta)$ by setting
    \begin{align*}
        H_t:=H_0+tV,\qquad
        \omega_t:=\frac{e^{-\beta H_t}}{\Tr(e^{-\beta H_t})},
        \qquad 0\leq t\leq1.
    \end{align*}
    Duhamel's formula and differentiation of the normalization give
    \begin{align*}
        \frac{\d\omega_t}{\d t}
        =-\beta\int_0^1
            \omega_t^s\bigl(V-\Tr(\omega_tV)I\bigr)
            \omega_t^{1-s}\,\ud s.
    \end{align*}
    Write $\|B\|_p:=(\Tr|B|^p)^{1/p}$ for the Schatten $p$-norm,
    $1\leq p<\infty$, with $|B|=(B^\dagger B)^{1/2}$; the
    $p=\infty$ case is the spectral norm. Schatten H\"older's inequality implies
    \begin{align*}
        \|\omega_t^s V\omega_t^{1-s}\|_1
        \leq\|\omega_t^s\|_{1/s}\|V\|_\infty
            \|\omega_t^{1-s}\|_{1/(1-s)}
        =\|V\|_\infty,
    \end{align*}
    with the endpoint cases understood by continuity. Consequently,
    $\|\d\omega_t/\d t\|_1\leq2\beta\|V\|_\infty$.
    Since $\omega_1=\rho(\theta)$ and
    $\Tr(\omega_0O)=2^{-n}\Tr O$, trace duality yields
    \begin{align*}
        |\Tr(\rho(\theta)O)-2^{-n}\Tr O|
        &\leq\|O\|_\infty
            \int_0^1\left\|\frac{\d\omega_t}{\d t}\right\|_1\ud t\\
        &\leq2\beta q(D+1)\|O\|_\infty.
    \end{align*}
\end{proof}

We next express the gradient and Hessian in terms of the score responses.
For these identities, allow a general $\lambda\in[0,1]$ and use the
$\lambda$-scores and weighted inner product defined
in~\cref{append:lambda_score}. Set
\begin{align*}
    \Delta S_{\lambda,a}(\theta)
    &:=S_{\lambda,a}(\theta)-S_{\lambda,a}(\theta^*),\\
    T^j_{\lambda,a}(\theta)
    &:=\partial_{\theta_j}S_{\lambda,a}(\theta),
    &
    U^{ij}_{\lambda,a}(\theta)
    &:=\partial_{\theta_i}\partial_{\theta_j}S_{\lambda,a}(\theta).
\end{align*}
Write $J_{Q,\lambda}(\theta):=\frac12\sum_a\|\Delta S_{\lambda,a}(\theta)\|^2_{\sigma,\lambda}$,
so that $J_{Q,1/2}=J_Q$. Since the target $\sigma=\rho(\theta^*)$ is
fixed, differentiating this expression gives the exact identities
\begin{align}
    g_{\lambda,j}(\theta)
    &:=\partial_{\theta_j}J_{Q,\lambda}(\theta)
    =\sum_a\Re\left\langle
        T^j_{\lambda,a}(\theta),\Delta S_{\lambda,a}(\theta)
    \right\rangle_{\sigma,\lambda},
    \label{eqn:general-lambda-objective-gradient}\\
    \mathsf H_{\lambda,ij}(\theta)
    &:=\partial_{\theta_i}\partial_{\theta_j}J_{Q,\lambda}(\theta)
    \nonumber\\
    &=\sum_a\Re\left(
        \left\langle T^i_{\lambda,a}(\theta),T^j_{\lambda,a}(\theta)
        \right\rangle_{\sigma,\lambda}
        +\left\langle U^{ij}_{\lambda,a}(\theta),\Delta S_{\lambda,a}(\theta)
        \right\rangle_{\sigma,\lambda}
    \right).
    \label{eqn:general-lambda-objective-hessian}
\end{align}

For the remainder of the proof, specialize to $\lambda=1/2$ and write
$S_a=S_{1/2,a}$, $T_{a,j}=T^j_{1/2,a}$, and
$U_{a,ij}=U^{ij}_{1/2,a}$, together with
$g=g_{1/2}=\nabla_\theta J_Q$ and
$\mathsf H=\mathsf H_{1/2}=\nabla_\theta^2J_Q$.
For Hermitian $X,Y$, abbreviate
$\langle X,Y\rangle_\sigma:=\langle X,Y\rangle_{\sigma,1/2}
=\frac12\Tr(\sigma\{X,Y\})$.
The commutator Gram matrix from~\cref{sec:high-temperature-learning}
can be written as
\begin{align}
    \Gamma_{i,j}
    &:=2^{-n}\sum_a\Tr\bigl(Q_{a,i}Q_{a,j}\bigr),
    &Q_{a,j}&:=\partial_a(P_j)=\I[P_j,A_a].
    \label{eqn:infinite-temperature-score-gram}
\end{align}
For the Pauli model and frame used here, \cref{eqn:score-gram-diagonal}
shows that $\Gamma$ is diagonal and
\begin{align}
    \Gamma_{\min}:=\min_j\Gamma_{j,j}\geq4,
    \qquad
    \max_j\Gamma_{j,j}\leq8k,
    \qquad
    \kappa(\Gamma)\leq2k.
    \label{eqn:score-gram-diagonal-bounds}
\end{align}

To see the leading high-temperature behavior, expand the SLD multiplier
from~\cref{eqn:g_lambda_ft}:
\begin{align*}
    \phi_{1/2}(z)=-2\I\tanh(z/2)=-\I z+\mathcal{O}(z^3).
\end{align*}
At fixed finite volume, uniformly for
$\theta,\theta^*\in\Theta$, this gives
\begin{align*}
    T_{a,j}(\theta)&=-\beta Q_{a,j}+\mathcal{O}(\beta^3),
    &U_{a,ij}(\theta)&=\mathcal{O}(\beta^3),\\
    \Delta S_a(\theta)
    &=-\beta\sum_j(\theta_j-\theta_j^*)Q_{a,j}+\mathcal{O}(\beta^3).
\end{align*}
Since $\sigma=2^{-n}I+\mathcal{O}(\beta)$ at fixed volume, substituting these
expansions into
\cref{eqn:general-lambda-objective-gradient,eqn:general-lambda-objective-hessian}
yields
\begin{align}
    g(\theta)
    &=\beta^2\Gamma(\theta-\theta^*)+\mathcal{O}(\beta^3),
    &\mathsf H(\theta)&=\beta^2\Gamma+\mathcal{O}(\beta^3).
    \label{eqn:formal-sld-gradient-hessian}
\end{align}
The following lemma makes the Hessian estimate uniform in the system
size and turns it into a contraction bound for the learning update.

\begin{lemma}[Uniform high-temperature Hessian]
\label{lem:finite-high-temperature-hessian}
    Let $g=\nabla_\theta J_Q$ and $\mathsf H=\nabla_\theta^2J_Q$, with
    the target $\theta^*$ fixed; write $g(\theta;\theta^*)$ when its
    target dependence is needed. Let $\Gamma$ and $\Gamma_{\min}$ be
    as in~\cref{eqn:infinite-temperature-score-gram,eqn:score-gram-diagonal-bounds}.
    There exist a constant $C_{\mathrm H}$ and a threshold
    $\beta_0\in(0,1]$, satisfying
    \begin{align*}
        C_{\mathrm H}=\poly(k,D),
        \qquad
        \beta_0=\Omega\bigl(1/\poly(k,D)\bigr),
    \end{align*}
    such that, for $0<\beta\leq\beta_0$,
    \begin{align}
        \sup_{\theta,\theta^*\in\Theta}
        \left\|
            \mathsf H(\theta)-\beta^2\Gamma
        \right\|_{\infty\to\infty}
        \leq
        C_{\mathrm H}\beta^3.
        \label{eqn:finite-high-temperature-hessian}
    \end{align}
    Let
    \begin{align}
        q
        &:=
        \frac{C_{\mathrm H}\beta}{\Gamma_{\min}}.
        \label{eqn:finite-high-temperature-contraction-factor}
    \end{align}
    After decreasing $\beta_0$ if necessary, $q<1$.  For each fixed
    target $\theta^*\in\Theta$, the projected population update
    \begin{align}
        \mathcal T_{\theta^*}(\theta)
        :=
        \Pi_\Theta\left(
            \theta-\beta^{-2}\Gamma^{-1}g(\theta)
        \right)
        \label{eqn:high-temperature-population-update}
    \end{align}
    is a contraction: for all $\theta,\eta\in\Theta$,
    \begin{align}
        \left\|
            \mathcal T_{\theta^*}(\theta)
            -
            \mathcal T_{\theta^*}(\eta)
        \right\|_\infty
        \leq
        q\|\theta-\eta\|_\infty.
        \label{eqn:high-temperature-population-pairwise-contraction}
    \end{align}
    In particular, $\mathcal T_{\theta^*}(\theta^*)=\theta^*$
    and
    \begin{align}
        \left\|
            \mathcal T_{\theta^*}(\theta)-\theta^*
        \right\|_\infty
        \leq
        q\|\theta-\theta^*\|_\infty.
        \label{eqn:high-temperature-population-contraction}
    \end{align}
\end{lemma}

\begin{proof}
    The fixed-volume expansion in \cref{eqn:formal-sld-gradient-hessian} is not by itself uniform in an extensive
    system, because the Hessian contains sums over $a$ and the parameter labels.  In the
    notation of this subsection, \cref{prop:uniform-score-response} gives
    \begin{align}
        T_{a,i}(\theta)
        &=
        -\beta\partial_a(P_i)+R_{a,i}(\theta),
        \label{eqn:finite-score-response-expansion}\\
        \sup_i\sum_a\|R_{a,i}(\theta)\|_\infty
        +
        \sup_a\sum_i\|R_{a,i}(\theta)\|_\infty
        &\leq
        C\beta^2,\nonumber\\
        \sup_i\sum_{a,j}
        \|U_{a,ij}(\theta)\|_\infty
        &\leq
        C\beta^2,\nonumber\\
        \sup_a\|\Delta S_a(\theta)\|_\infty
        &\leq
        C\beta.
        \label{eqn:uniform-response-summary}
    \end{align}
    These row- and column-summability bounds are uniform in
    $\theta,\theta^*\in\Theta$ and have constants independent of
    $n$ and $m$.

    Specializing \cref{eqn:general-lambda-objective-hessian} to the SLD,
    the contribution obtained by retaining the leading term in both
    score responses is
    \begin{align}
        L_{i,j}
        :=
        \frac{\beta^2}{2}
        \sum_a
        \Tr\left(
            \sigma
            \left\{
                \partial_a(P_i),\partial_a(P_j)
            \right\}
        \right).
        \label{eqn:finite-temperature-score-gram-term}
    \end{align}
    For fixed $i$, only a number of pairs $(a,j)$ bounded in terms of
    $k$ and $D$ contribute.  Each anticommutator in
    \cref{eqn:finite-temperature-score-gram-term} has support on at most
    $2k$ qubits and has norm at most $8$.
    Hence \cref{lem:local-thermal-bias} allows us to replace
    its expectation in $\sigma=\rho(\theta^*)$ by its normalized
    trace with row-sum error at most $C\beta^3$, where
    $C=\poly(k,D)$.  The normalized-trace
    contribution is exactly
    \begin{align*}
        \beta^2 2^{-n-1}\sum_a
        \Tr\left(
            \left\{
                \partial_a(P_i),\partial_a(P_j)
            \right\}
        \right)
        =
        \beta^2\Gamma_{i,j}.
    \end{align*}

    To bound the response remainders, note that locality of the Pauli
    strings and of the derivation frame gives
    \begin{align*}
        \sup_i\sum_a\|\partial_a(P_i)\|_\infty
        &\leq
        4k,
        &
        \sup_a\sum_j\|\partial_a(P_j)\|_\infty
        &\leq
        2(D+1).
    \end{align*}
    Combining these estimates with
    \cref{eqn:uniform-response-summary}, the two mixed
    $\partial_a(P)$--$R$ contributions have row-sum norm
    $\mathcal{O}(\beta^3)$, while the $R$--$R$ contribution is $\mathcal{O}(\beta^4)$.
    The remaining term in the exact Hessian satisfies
    \begin{align*}
        \sup_i\sum_j
        \left|
            \sum_a
            \left\langle
                U_{a,ij}(\theta),
                \Delta S_a(\theta)
            \right\rangle_\sigma
        \right|
        &\leq
        \left(
            \sup_i\sum_{a,j}
            \|U_{a,ij}(\theta)\|_\infty
        \right)
        \left(
            \sup_a
            \|\Delta S_a(\theta)\|_\infty
        \right)
        \nonumber\\
        &\leq
        C\beta^3.
    \end{align*}
    Enlarging the constant if necessary proves
    \cref{eqn:finite-high-temperature-hessian}.

    It remains to prove the contraction claim.  For a fixed target, let
    \begin{align*}
        F(\theta)
        :=
        \theta-\beta^{-2}\Gamma^{-1}g(\theta)
    \end{align*}
    denotes the unprojected update.  Since $\Theta$ is convex, the
    fundamental theorem of calculus and
    \cref{eqn:finite-high-temperature-hessian} give, for
    $\theta,\eta\in\Theta$,
    \begin{align*}
        F(\theta)-F(\eta)
        &=
        \int_0^1
        \left(
            I-\beta^{-2}\Gamma^{-1}
            \mathsf H\bigl(\eta+t(\theta-\eta)\bigr)
        \right)
        (\theta-\eta)\,\ud t,\\
        \|F(\theta)-F(\eta)\|_\infty
        &\leq
        \frac{C_{\mathrm H}\beta}{\Gamma_{\min}}
        \|\theta-\eta\|_\infty.
    \end{align*}
    Here we used that $\Gamma$ is diagonal and therefore
    $\|\Gamma^{-1}\|_{\infty\to\infty}=\Gamma_{\min}^{-1}$.
    Coordinatewise projection onto $\Theta$ is nonexpansive in
    $\ell_\infty$, which proves
    \cref{eqn:high-temperature-population-pairwise-contraction}.
    Finally, \cref{eqn:general-lambda-objective-gradient} gives
    $g(\theta^*)=0$ because
    $\Delta S_a(\theta^*)=0$.  Thus the update fixes
    $\theta^*$, and
    \cref{eqn:high-temperature-population-contraction} follows.
\end{proof}

\begin{lemma}[One-step statistical error]
\label{lem:one-step-high-temperature-error}
    Under \cref{lem:finite-high-temperature-hessian}, let
    $\theta^{(t+1)}$ be the update in
    \cref{eqn:stochastic-high-temperature-update}.  If
    \begin{align}
        \|\widehat g_t
        -g(\theta^{(t)};\theta^*)\|_\infty
        \leq\tau_t,
    \end{align}
    then
    \begin{align}
        \|\theta^{(t+1)}-\theta^*\|_\infty
        \leq
        q\|\theta^{(t)}-\theta^*\|_\infty
        +
        \frac{\tau_t}{\Gamma_{\min}\beta^2}.
        \label{eqn:one-step-high-temperature-error}
    \end{align}
\end{lemma}

\begin{proof}
    Insert and subtract the population update
    \cref{eqn:high-temperature-population-update}.  Its contribution
    contracts by \cref{eqn:high-temperature-population-contraction}.
    The remaining term is bounded by
    \begin{align*}
        \left\|
            \beta^{-2}\Gamma^{-1}
            \bigl(
                \widehat g_t
                -g(\theta^{(t)};\theta^*)
            \bigr)
        \right\|_\infty
        \leq
        \frac{\tau_t}{\Gamma_{\min}\beta^2}.
    \end{align*}
\end{proof}

\begin{lemma}[Geometric stochastic contraction]
\label{lem:geometric-stochastic-contraction}
    Under \cref{lem:one-step-high-temperature-error}, let $q<r<1$,
    let $0<\varepsilon<R_0$, and suppose
    $\|\theta^{(0)}-\theta^*\|_\infty\leq R_0$.  Choose
    \begin{align}
        \tau_t
        =
        \Gamma_{\min}\beta^2
        (r-q)R_0r^t.
        \label{eqn:nonresonant-gradient-schedule}
    \end{align}
    If every gradient estimate meets its prescribed tolerance, then
    \begin{align}
        \|\theta^{(t)}-\theta^*\|_\infty
        \leq
        R_0r^t.
        \label{eqn:stochastic-geometric-contraction}
    \end{align}
    Hence
    \begin{align}
        T
        =
        \left\lceil
            \frac{\log(R_0/\varepsilon)}
            {\log(1/r)}
        \right\rceil
        \label{eqn:high-temperature-iteration-count}
    \end{align}
    iterations suffice for coefficient error at most $\varepsilon$.
\end{lemma}

\begin{proof}
    The claim holds at $t=0$.  If it holds at $t$, then
    \cref{eqn:one-step-high-temperature-error,eqn:nonresonant-gradient-schedule}
    give
    \begin{align*}
        \|\theta^{(t+1)}-\theta^*\|_\infty
        \leq
        qR_0r^t+(r-q)R_0r^t
        =
        R_0r^{t+1}.
    \end{align*}
    Induction proves \cref{eqn:stochastic-geometric-contraction}, and
    the stated value of $T$ makes its right-hand side at most
    $\varepsilon$.
\end{proof}

\begin{proof}[Proof of~\Cref{thm:end-to-end-qsm-upper-bound}]
    Choose $\beta_0$ below the thresholds in
    \cref{lem:finite-high-temperature-hessian,thm:high-temperature-gradient-shells}
    and small enough that $q\leq q_0:=1/2$ and $\beta_0\leq1/(6k)$. Conditioned on the preceding
    rounds, $\theta^{(t)}$ is fixed, so
    \cref{thm:shared-copy-gradient-estimation} applies to the adaptive
    gradient observable at that round, using fresh copies of $\sigma$. It gives gradient error at most
    $\tau_t$ with conditional failure probability $\delta_t$ using
    \begin{align}
        N_t
        =
        \mathcal{O}\left(
            \frac{\beta^2}{\tau_t^2}
            \log\frac{2m}{\delta_t}
        \right)
        \label{eqn:finite-gradient-estimation-cost}
    \end{align}
    copies: the tolerances below satisfy $\tau_t\leq\frac34\Gamma_{\min}\beta^2\leq6k\beta^2\leq\beta$, so the additive $1$ in~\cref{eqn:shared-copy-gradient-cost} is absorbed.

    Since $\theta^*\in[-1,1]^m$, the initialization satisfies
    $\|\theta^{(0)}-\theta^*\|_\infty\leq R_0$ with $R_0=1$.
    Choose $r=(1+q_0)/2=3/4$ and use
    \cref{eqn:nonresonant-gradient-schedule}.  For the $T$ in
    \cref{eqn:high-temperature-iteration-count}, assign
    \begin{align}
        \delta_t
        =
        \delta\,
        \frac{(1-r^2)r^{2(T-1-t)}}{1-r^{2T}},
        \qquad
        0\leq t<T.
    \end{align}
    These probabilities sum to $\delta$.  Substituting the tolerance schedule in \cref{eqn:nonresonant-gradient-schedule} into \cref{eqn:finite-gradient-estimation-cost} makes the copy costs $N_t$ grow geometrically.  With $s=T-1-t$,
    \begin{align*}
        \sum_{t=0}^{T-1}N_t
        &=
        \mathcal{O}\left(
            \frac{1}
            {\Gamma_{\min}^2\beta^2(r-q)^2
            R_0^2r^{2(T-1)}}
            \sum_{s=0}^{T-1}
            r^{2s}
            \left(
                \log\frac{2m}{\delta}
                +\mathcal{O}(s+1)
            \right)
        \right).
    \end{align*}
    The two finite sums are bounded by their infinite counterparts, $\sum_{s\geq0}r^{2s}=(1-r^2)^{-1}$ and $\sum_{s\geq0}(s+1)r^{2s}=(1-r^2)^{-2}$. Since $r=3/4$, both are absolute constants independent of $T$. Since $R_0r^{T-1}>\varepsilon$ and
    $r-q\geq r-q_0 = (1-q_0)/2$, the preceding display is
    \cref{eqn:end-to-end-qsm-copy-bound}.  On the event that all
    gradient estimates succeed,
    \cref{lem:geometric-stochastic-contraction} gives the desired
    coefficient error; the chosen failure probabilities make this
    event occur with probability at least $1-\delta$.
\end{proof}

\section{Information-theoretic lower bounds for learning Gibbs states}
\label{sec:information-lower-bounds}

We show that a Hamiltonian coefficient becomes exponentially harder to
resolve when it changes only a thermally suppressed population.  This
effect already appears in a noninteracting Pauli model.  Its
high-temperature limit recovers the minimax rate
$\Omega(\log(M/\delta)/(\beta^2\epsilon^2))$.  Both conclusions concern
the Gibbs-state statistical experiment and therefore apply
independently of the estimator used.

Consider
\begin{align}
    H(\theta)=\sum_{j=1}^M\theta_jP_j,
    \qquad
    \rho_\theta=\frac{e^{-\beta H(\theta)}}
    {\Tr(e^{-\beta H(\theta)})},
    \qquad
    \theta\in[-1,1]^M ,
    \label{eqn:lower-bound-model}
\end{align}
where the $P_j$ are fixed Pauli strings and $\beta>0$ is known.  An
estimator based on $m$ copies has accuracy $\epsilon$ and failure
probability $\delta$ if
\begin{align}
    \PP_\theta\left(
        \norm{\widehat\theta-\theta}_\infty\leq\epsilon
    \right)\geq1-\delta
    \label{eqn:lower-bound-guarantee}
\end{align}
uniformly over the parameter class.  The results below allow arbitrary
collective measurements on the $m$ copies.

\begin{theorem}[Information-theoretic sample lower bound]
    \label{thm:minimax-sample-lower-bound}
    For any $\beta>0$, $M\geq2$, $0<\epsilon<1/4$, and
    $0<\delta\leq1/4$, there exists a noninteracting one-local Pauli
    model with $M$ unknown coefficients for which every estimator
    satisfying~\cref{eqn:lower-bound-guarantee} requires
    \begin{align}
        m
        =
        \Omega\left(
            \frac{e^{\beta(2-8\epsilon)}}
            {\beta^2\epsilon^2}
            \log\frac{M}{\delta}
        \right).
        \label{eqn:hidden-sector-minimax-lower-bound}
    \end{align}
    For $\beta=\mathcal{O}(1)$, this reduces to the high-temperature
    minimax rate
    $m=\Omega(\log(M/\delta)/(\beta^2\epsilon^2))$.
\end{theorem}

The proof and the explicit one-qubit construction are given in
\cref{sec:thermally-hidden-directions}.  The two-qubit construction of
Ref.~\cite{haah2024learning} is discussed separately in
\cref{append:hkt-hidden-sector}.

We first identify the thermal mechanism that controls the relative
entropy between nearby Gibbs states and then apply it to the
noninteracting model in~\cref{thm:minimax-sample-lower-bound}.

\subsection{Relative entropy along a Hamiltonian path}

The statistical distance between two Gibbs states is controlled by the
thermal fluctuations of the Hamiltonian direction that separates them.
For a full-rank state $\rho$ and a Hermitian operator $V$, define the
Bogoliubov--Kubo--Mori covariance
\begin{align}
    \operatorname{Cov}^{\mathrm{BKM}}_\rho(V,V)
    :=
    \int_0^1
        \Tr\left(\rho^uV\rho^{1-u}V\right)\ud u
    -\Tr^2(\rho V) .
    \label{eqn:bkm-covariance}
\end{align}

\begin{lemma}[Gibbs chord identity]
    \label{lem:gibbs-chord}
    Let
    \begin{align*}
        H_s=H_0+sV,
        \qquad
        \rho_s=\frac{e^{-\beta H_s}}{\Tr(e^{-\beta H_s})},
        \qquad 0\leq s\leq1 .
    \end{align*}
    Then
    \begin{align}
        D(\rho_1\Vert\rho_0)
        =
        \beta^2\int_0^1
        s\,\operatorname{Cov}^{\mathrm{BKM}}_{\rho_s}(V,V)
        \ud s .
        \label{eqn:gibbs-chord}
    \end{align}
\end{lemma}

\begin{proof}
    Write
    \begin{align*}
        Z(s):=\Tr(e^{-\beta H_s}),
        \qquad
        \psi(s):=\log Z(s).
    \end{align*}
    Duhamel's formula gives
    \begin{align*}
        \frac{\ud}{\ud s}e^{-\beta H_s}
        =
        -\beta\int_0^1
        e^{-\beta(1-u)H_s}V
        e^{-\beta uH_s}\ud u.
    \end{align*}
    Taking the trace and using cyclicity gives
    \begin{align}
        \psi'(s)
        &=-\beta\Tr(\rho_sV),&
        \psi''(s)
        &=\beta^2
        \operatorname{Cov}^{\mathrm{BKM}}_{\rho_s}(V,V).
        \label{eqn:log-partition-chord-derivatives}
    \end{align}
    More explicitly,
    \begin{align*}
        \frac{Z''(s)}{Z(s)}
        &=
        \beta^2\int_0^1
        \Tr(\rho_s^uV\rho_s^{1-u}V)\ud u,\\
        \left(\frac{Z'(s)}{Z(s)}\right)^2
        &=
        \beta^2\Tr^2(\rho_sV),
    \end{align*}
    whose difference is $\psi''(s)$.

    Finally,
    \begin{align*}
        \log\rho_s-\log\rho_0
        =
        -\beta sV-\psi(s)+\psi(0),
    \end{align*}
    and hence
    \begin{align*}
        D(\rho_s\Vert\rho_0)
        &=
        -\beta s\Tr(\rho_sV)-\psi(s)+\psi(0)\\
        &=
        s\psi'(s)-\psi(s)+\psi(0).
    \end{align*}
    Its derivative is $s\psi''(s)$.  Integrating from $0$ to $1$
    proves~\cref{eqn:gibbs-chord}.
\end{proof}

The chord identity is exact and does not assume a small perturbation.
It exposes the relevant mechanism directly: a Hamiltonian direction is
hard to learn when its thermal variance remains small along the path
between the two hypotheses.

\subsection{Thermally hidden Hamiltonian directions}
\label{sec:thermally-hidden-directions}

This suppression occurs when a perturbation acts only on a sector with
exponentially small Gibbs weight.

\begin{theorem}[Thermally hidden directions]
    \label{thm:thermally-hidden-directions}
    Let $P$ and $Q$ be orthogonal projectors such that
    \begin{align}
        H_0P&=E_0P,
        &
        \operatorname{rank}(P)&=g,
        &
        [Q,H_0]&=0,
        &
        QH_0Q&\succeq(E_0+\Delta)Q,
        &
        \operatorname{rank}(Q)&=r .
        \label{eqn:hidden-sector-assumptions}
    \end{align}
    Suppose that $V=QVQ$, $\norm{V}_\infty\leq v$, and
    $H_1=H_0+\tau V$ with $\abs{\tau}v<\Delta$.  Then
    \begin{align}
        D(\rho_1\Vert\rho_0)
        \leq
        \frac{\beta^2\tau^2v^2}{2}
        \frac{r}{g}
        e^{-\beta(\Delta-\abs{\tau}v)} .
        \label{eqn:hidden-sector-relative-entropy}
    \end{align}
\end{theorem}

\begin{proof}
    Apply~\cref{lem:gibbs-chord} to
    $H_s=H_0+s\tau V$.  Diagonalize
    $\rho_s=\sum_m p_m\ket{m}\!\bra{m}$.  In this basis,
    \begin{align*}
        \int_0^1
        \Tr(\rho_s^uV\rho_s^{1-u}V)\ud u
        =
        \sum_{m,n}
        \left(
            \int_0^1p_m^up_n^{1-u}\ud u
        \right)
        \abs{V_{mn}}^2.
    \end{align*}
    Weighted AM--GM gives
    \begin{align*}
        p_m^up_n^{1-u}
        \leq
        up_m+(1-u)p_n.
    \end{align*}
    Integrating over $u$ and summing over $m,n$ therefore yields
    \begin{align*}
        \int_0^1
        \Tr(\rho_s^uV\rho_s^{1-u}V)\ud u
        \leq
        \frac12\sum_{m,n}(p_m+p_n)\abs{V_{mn}}^2
        =
        \Tr(\rho_sV^2).
    \end{align*}
    Since the centering term in the BKM covariance is nonnegative and
    $V^2\preceq v^2Q$, we obtain
    \begin{align}
        \operatorname{Cov}^{\mathrm{BKM}}_{\rho_s}(V,V)
        \leq
        \Tr(\rho_sV^2)
        \leq
        v^2\Tr(Q\rho_s).
        \label{eqn:hidden-sector-bkm-bound}
    \end{align}
    Since $Q$ reduces both $H_0$ and $V$, it also reduces $H_s$.  On
    this sector,
    \begin{align*}
        QH_sQ
        =
        QH_0Q+s\tau V
        \succeq
        \left(E_0+\Delta-\abs{\tau}v\right)Q.
    \end{align*}
    Thus
    \begin{align*}
        \Tr(Qe^{-\beta H_s})
        \leq
        r e^{-\beta(E_0+\Delta-\abs{\tau}v)}.
    \end{align*}
    The perturbation vanishes on $P$, so the partition function obeys
    \begin{align*}
        \Tr(e^{-\beta H_s})
        \geq
        \Tr(Pe^{-\beta H_s})
        =
        g e^{-\beta E_0}.
    \end{align*}
    Dividing the last two bounds gives, uniformly along the chord,
    \begin{align}
        \Tr(Q\rho_s)
        \leq
        \frac{r}{g}
        e^{-\beta(\Delta-\abs{\tau}v)} .
        \label{eqn:hidden-sector-weight}
    \end{align}
    Substitution into~\cref{eqn:gibbs-chord}, including the chord
    direction $\tau V$, proves the result.
\end{proof}

The theorem separates the perturbation size from its thermal visibility. The polynomial factor $\beta^2\tau^2$ is the usual local
statistical scale, while the exponential factor is the Gibbs weight of
the sector that carries the perturbation.

\begin{corollary}
    \label{cor:hidden-sector-sample-lower-bound}
    Suppose that the Hamiltonians in
    \cref{thm:thermally-hidden-directions} correspond to parameters
    $\theta_0$ and $\theta_1$ satisfying
    $\norm{\theta_1-\theta_0}_\infty>2\epsilon$.  If
    \cref{eqn:lower-bound-guarantee} holds at both endpoints with
    $0<\delta<1/2$, then
    \begin{align}
        m
        \geq
        \frac{2g}{r}
        \frac{e^{\beta(\Delta-\abs{\tau}v)}}
        {\beta^2\tau^2v^2}
        d_{\mathrm{bin}}(1-\delta\Vert\delta),
        \label{eqn:hidden-sector-sample-lower-bound}
    \end{align}
    where
    \begin{align*}
        d_{\mathrm{bin}}(p\Vert q)
        =p\log\frac{p}{q}+(1-p)\log\frac{1-p}{1-q}.
    \end{align*}
    If the model contains $M$ disjoint copies of this perturbation,
    then for $0<\delta\leq1/4$,
    \begin{align}
        m
        =
        \Omega\left(
            \frac{g}{r}
            \frac{e^{\beta(\Delta-\abs{\tau}v)}}
            {\beta^2\tau^2v^2}
            \log\frac{M}{\delta}
        \right).
        \label{eqn:hidden-sector-disjoint-lower-bound}
    \end{align}
\end{corollary}

\begin{proof}
    The estimator induces a binary test between the two endpoints.
    Because the two $\epsilon$-balls are disjoint, its acceptance
    probability for $\theta_1$ is at least $1-\delta$ and at most
    $\delta$ for $\theta_0$.  Data processing under this test and
    tensorization of relative entropy, followed by monotonicity of the
    binary relative entropy in these two probabilities, give
    \begin{align*}
        mD(\rho_1\Vert\rho_0)
        \geq d_{\mathrm{bin}}(1-\delta\Vert\delta).
    \end{align*}
    This proves~\cref{eqn:hidden-sector-sample-lower-bound}.

    For the second claim, take the reference hypothesis with every
    block unperturbed and let hypothesis $j$ perturb only block $j$.
    Write $\alpha$ for the right-hand side of
    \cref{eqn:hidden-sector-relative-entropy}.  Relative entropy
    tensorizes, so
    \begin{align*}
        D(\rho_j\Vert\rho_0)\leq\alpha
    \end{align*}
    for every alternative.  Their parameter vectors are separated by
    more than $2\epsilon$, so a successful estimator identifies $j$.
    If $J$ is uniform on the $M$ alternatives, Fano's inequality and
    the Holevo bound give
    \begin{align*}
        \log M-h_{\mathrm{bin}}(\delta)
        -\delta\log(M-1)
        &\leq
        I(J;\widehat J)\\
        &\leq
        \inf_{\omega}
        \frac{1}{M}\sum_{j=1}^{M}
        D(\rho_j^{\otimes m}\Vert\omega)\\
        &\leq
        \frac{1}{M}\sum_{j=1}^{M}
        D(\rho_j^{\otimes m}\Vert\rho_0^{\otimes m})
        \leq
        m\alpha.
    \end{align*}
    Here
    \(h_{\mathrm{bin}}(x)=-x\log x-(1-x)\log(1-x)\).
    For $\delta\leq1/4$, this yields
    $m\alpha=\Omega(\log M)$.  The binary bound also gives
    $m\alpha=\Omega(\log(1/\delta))$.  Since
    \begin{align*}
        \max\left(\log M,\log\frac1\delta\right)
        \geq
        \frac12\log\frac{M}{\delta},
    \end{align*}
    the two bounds combine to give
    \cref{eqn:hidden-sector-disjoint-lower-bound}.
\end{proof}

\begin{proof}[Proof of \cref{thm:minimax-sample-lower-bound}]
    On one qubit, take $h_0=Z$ and
    $h_1=(1-4\epsilon)Z$.  Their coefficient distance is
    $4\epsilon$.  Let $P$ project onto the ground state of $h_0$ and
    $Q=I-P$.  Since adding a scalar does not change a Gibbs state,
    $h_1$ is statistically equivalent to
    \begin{align*}
        h_1-4\epsilon I=h_0-8\epsilon Q .
    \end{align*}
    Thus~\cref{thm:thermally-hidden-directions} applies with
    $g=r=v=1$, $\Delta=2$, and $\tau=-8\epsilon$.  Placing $M$
    independent copies of this block and applying
    \cref{eqn:hidden-sector-disjoint-lower-bound} proves the claim.
\end{proof}

\subsection{The construction of Ref.~\cite{haah2024learning} as a thermally hidden direction}
\label{append:hkt-hidden-sector}

The lower-bound construction in Ref.~\cite{haah2024learning} provides a simple instance of
\cref{thm:thermally-hidden-directions}.  On two qubits, consider
\begin{align}
    h_x
    &=
    -Z_1
    +\left(x-\frac12\right)Z_2
    -\left(x+\frac12\right)Z_1Z_2 \notag\\
    &=
    \operatorname{diag}(-2,0,1+2x,1-2x),
    \qquad 0\leq x\leq\frac12 .
    \label{eqn:hkt-block}
\end{align}
The perturbation
\begin{align}
    h_x-h_0=x(Z_2-Z_1Z_2)
    \label{eqn:hkt-perturbation}
\end{align}
acts only on the subspace spanned by the last two computational-basis
states.  This sector has rank two and lies three energy units above the
ground state at $x=0$.  Since
$\norm{Z_2-Z_1Z_2}_\infty=2$,
\cref{thm:thermally-hidden-directions} gives
\begin{align}
    D(\rho_x\Vert\rho_0)
    \leq
    4\beta^2x^2e^{-\beta(3-2x)} .
    \label{eqn:hkt-finite-radius-bound}
\end{align}

Equation~\eqref{eqn:hkt-finite-radius-bound} keeps the dependence on
the perturbation radius.  In the local regime $\beta x=o(1)$, it
reveals the information scale $e^{-3\beta}$.  If the result is required
uniformly over $0<x\leq1/2$, then $3-2x\geq2$ and one obtains the
weaker but uniform bound
\begin{align}
    D(\rho_x\Vert\rho_0)
    \leq
    4\beta^2x^2e^{-2\beta}.
    \label{eqn:hkt-uniform-bound}
\end{align}
The familiar $e^{2\beta}$ sample lower bound is the inverse of this
uniform information scale.  Its smaller exponent reflects the finite
range of perturbations covered by the statement, rather than a
limitation of the underlying construction.

\section{Numerical simulation details}\label{app:finite-shot-numerical-methods}

\subsection{Population quantities and finite-shot sampling}

We simulate learning for the eight-qubit Hamiltonian and local Pauli frame specified in \cref{eq:numerical-inhomogeneous-tfim}. The target Gibbs state and the population QSM objective are evaluated by exact diagonalization of the full Hamiltonian. We then differentiate the objective to obtain the population gradient at each optimization step. We consider the inverse temperatures $\beta=0.2,0.4,\ldots,1.6$ and the total per-iteration measurement budgets $N=10^3,3\times10^3,10^4,3\times10^4,10^5,$ and $10^6$. For every temperature, measurement budget, and initialization, we perform $100$ independent optimization runs with independently sampled measurement outcomes.

We simulate the finite-shot gradient estimator in \cref{alg:randomized-gradient-coordinate} from its measurement statistics. At iteration $t$, let $g_{j,t}:=(\nabla J_Q(\theta^{(t)}))_j$ be the $j$th population-gradient coordinate. One shot of the randomized Hadamard-test procedure returns $Y_{j,t}\in\{-1,1\}$. Its rescaled output is $L_j(R;\theta^{(t)})Y_{j,t}$, where $L_j(R;\theta^{(t)})$ is the coefficient mass in \cref{eqn:gradient-algorithm-masses} evaluated at the current candidate Hamiltonian.

Each sampled circuit branch contains one or two random Hamiltonian-evolution times. We truncate every such time draw to $[-R,R]$ using the distributions in \cref{eqn:gradient-algorithm-time-laws}. Define the numerical parameter domain $\Theta_{\mathrm{num}}:=[0,2]^{15}$, and let $g_{j,R}(\theta)$ denote the gradient coordinate obtained from the truncated representation. For each temperature, we choose $R$ so that
\begin{equation}
\sup_{\theta\in\Theta_{\mathrm{num}}}\max_j\left|g_{j,R}(\theta)-g_j(\theta)\right|\leq10^{-4}.
\end{equation}
To isolate finite-measurement fluctuations, we center the simulated estimator at the exact population gradient $g_{j,t}$; the finite cutoff enters only through the measurement range. Write $L_{j,t}:=L_j(R;\theta^{(t)})$. The target Gibbs state is fixed, whereas $\theta^{(t)}$ depends on the preceding measurement outcomes and determines the circuit used at iteration $t$. We therefore state the law of the next measurement outcome conditionally on the current iterate,
\begin{equation}\label{eq:numerical-binomial-gradient}
\PP(Y_{j,t}=1\mid\theta^{(t)})=\frac{1+g_{j,t}/L_{j,t}}{2},\qquad \widehat g_{j,t}=\frac{L_{j,t}}{N_{j,t}}\sum_{r=1}^{N_{j,t}}Y_{j,t}^{(r)}.
\end{equation}
Consequently, conditional on the current parameters,
\begin{equation}\label{eq:numerical-gradient-moments}
\mathbb E(\widehat g_{j,t}\mid\theta^{(t)})=g_{j,t},\qquad \operatorname{Var}(\widehat g_{j,t}\mid\theta^{(t)})=\frac{L_{j,t}^{2}-g_{j,t}^{2}}{N_{j,t}}.
\end{equation}
The total budget $N$ is divided among the $15$ gradient coordinates in proportion to their current measurement ranges $L_{j,t}$, with integer rounding chosen so that $\sum_jN_{j,t}=N$. Every measurement uses a fresh copy of the target Gibbs state, and no coherent amplitude amplification is used.

\subsection{Optimization protocol}

We use projected, preconditioned gradient descent. For the parameter order $(J_1,\ldots,J_7,h_1^x,\ldots,h_8^x)$, the high-temperature commutator Gram matrix is $\Gamma=\operatorname{diag}(8I_7,4I_8)$ and sets the leading curvature scale. The unprojected step at iteration $t$ is $\widetilde d_t=\eta_t\beta^{-2}\Gamma^{-1}\widehat{\nabla J_Q}(\theta^{(t)})$. We set $d_t=\widetilde d_t$ when $\|\widetilde d_t\|_2\leq0.05\|\theta^*\|_2$ and otherwise rescale it so that $\|d_t\|_2=0.05\|\theta^*\|_2$. The projected update is
\begin{equation}\label{eq:numerical-projected-update}
\theta^{(t+1)}=\Pi_{\Theta_{\mathrm{num}}}(\theta^{(t)}-d_t).
\end{equation}
We consider far and local initializations. The far initialization sets $J_i^{(0)}=0.5$ and $(h_i^x)^{(0)}=1$. For each local trajectory, we independently draw a standard Gaussian vector $v\in\mathbb R^{15}$ and set $\theta^{(0)}=\theta^*+0.05\|\theta^*\|_2v/\|v\|_2$. Thus every local initialization lies at relative distance $e_{\mathrm{rel}}(\theta^{(0)})=0.05$ from the target.

The two initializations use different initial learning-rate schedules. Local runs start within the target neighborhood and use $T=200$ updates with $\eta_t=0.5/(1+t/10)$. Far runs use $T=300$ updates. Their first $60$ updates use a temperature-dependent constant rate $\eta_{\mathrm{far}}$,
\begin{center}
\begin{tabular}{@{}crrrrrrrr@{}}
\toprule
$\beta$ & $0.2$ & $0.4$ & $0.6$ & $0.8$ & $1.0$ & $1.2$ & $1.4$ & $1.6$ \\
\midrule
$\eta_{\mathrm{far}}$ & $0.25$ & $0.5$ & $1$ & $2$ & $4$ & $8$ & $16$ & $32$
\\
\bottomrule
\end{tabular}
\end{center}
After the first $60$ updates, far runs use $\eta_t=0.5/(1+(t-60)/10)$. The final output is the averaged estimate $\bar\theta=(T-t_0+1)^{-1}\sum_{t=t_0}^{T}\theta^{(t)}$, with $t_0=100$ for local runs and $t_0=150$ for far runs. At fixed temperature and initialization, we keep the optimization schedule unchanged across all shot budgets.

\subsection{IBM quantum-hardware experiment}\label{app:ibm-hardware-implementation}

The experiments were conducted on the IBM \texttt{ibm\_pittsburgh} processor on September 20, 2026.

For the hardware experiment in \cref{sec:ibm-hardware-experiment}, we allocate the randomized circuit instances between the $J$ and $h$ gradient coordinates in proportion to their coefficient masses $L_j(R;\theta^{(t)})$ in \cref{eqn:gradient-algorithm-masses}. Within each instance, the sampled Pauli terms, evolution times, and target energy eigenstate remain fixed for all $32$ shots; these choices are sampled independently for the next instance. If $K_j$ instances are allocated to coordinate $j$, the gradient estimate is
\begin{equation}\label{eq:ibm-measured-gradient}
\widehat g_j=\frac{L_j(R;\theta^{(t)})}{K_j}\sum_{\ell=1}^{K_j}s_{j\ell}\frac{n_{0,j\ell}-n_{1,j\ell}}{32},
\end{equation}
where $s_{j\ell}\in\{-1,1\}$ is the sign of the coefficient of the sampled term and $n_{b,j\ell}$ counts ancilla outcome $b$ for instance $\ell$.

To execute the measurements in parallel, we randomly assign the circuit instances from all runs to $19$ disjoint five-qubit paths on the processor. Each processor circuit therefore measures up to $19$ instances simultaneously. We use a control-free implementation in which the ancilla controls only Pauli insertions, while Hamiltonian evolutions remain uncontrolled, following the control-free design principle in Ref.~\cite{dong2022ground}. The second-order Trotterization uses a maximal time step of $0.25$. We truncate the random evolution times so that the omitted integral tail is bounded by $10^{-4}$ uniformly over the candidate parameter domain.

We use the projected, preconditioned update in \cref{eq:numerical-projected-update}, with $\Gamma=\operatorname{diag}(24,16)$, learning rate $\eta_t=0.5/(1+t/10)$, and projection onto $\Theta_{\mathrm{num}}=[0,2]^2$. The step-norm cap remains $0.05\|\theta^*\|_2$. We report the raw parameter iterates. \Cref{tab:ibm-epoch-trajectories} lists the mean parameters and mean relative error at every update across the five independent runs at each temperature.

\begingroup
\small
\setlength{\tabcolsep}{9pt}
\begin{longtable}{@{}>{\raggedleft\arraybackslash}p{3em}*{3}{>{\centering\arraybackslash}p{\dimexpr(0.95\textwidth-5em-10\tabcolsep)/6\relax}}@{\qquad}*{3}{>{\centering\arraybackslash}p{\dimexpr(0.95\textwidth-5em-10\tabcolsep)/6\relax}}@{}}
\caption{IBM hardware learning trajectories from epoch $0$ (the common initialization) through epoch $45$. Entries show the mean $\pm$ one sample standard deviation over five independent runs. Relative parameter errors are computed for each run using \cref{eq:numerical-relative-learning-error} and shown as percentages.}\label{tab:ibm-epoch-trajectories}\\
\toprule
& \multicolumn{3}{c}{$\beta=0.2$} & \multicolumn{3}{c}{$\beta=0.4$}\\
\cmidrule(lr){2-4}\cmidrule(l){5-7}
Epoch & $\bar J$ & $\bar h$ & $\overline e_{\mathrm{rel}}$ (\%) & $\bar J$ & $\bar h$ & $\overline e_{\mathrm{rel}}$ (\%)\\
\midrule
\endfirsthead
\toprule
& \multicolumn{3}{c}{$\beta=0.2$} & \multicolumn{3}{c}{$\beta=0.4$}\\
\cmidrule(lr){2-4}\cmidrule(l){5-7}
Epoch & $\bar J$ & $\bar h$ & $\overline e_{\mathrm{rel}}$ (\%) & $\bar J$ & $\bar h$ & $\overline e_{\mathrm{rel}}$ (\%)\\
\midrule
\endhead
\midrule
\endfoot
\bottomrule
\endlastfoot
0 & $0.50\,\pm\,0.00$ & $0.50\,\pm\,0.00$ & $62.0\,\pm\,0.0$ & $0.50\,\pm\,0.00$ & $0.50\,\pm\,0.00$ & $62.0\,\pm\,0.0$\\
1 & $0.52\,\pm\,0.02$ & $0.59\,\pm\,0.00$ & $57.4\,\pm\,0.6$ & $0.54\,\pm\,0.02$ & $0.57\,\pm\,0.02$ & $57.2\,\pm\,0.4$\\
2 & $0.56\,\pm\,0.02$ & $0.66\,\pm\,0.02$ & $52.6\,\pm\,0.9$ & $0.57\,\pm\,0.06$ & $0.62\,\pm\,0.06$ & $54.4\,\pm\,2.2$\\
3 & $0.62\,\pm\,0.04$ & $0.71\,\pm\,0.03$ & $48.4\,\pm\,1.0$ & $0.62\,\pm\,0.06$ & $0.68\,\pm\,0.08$ & $50.1\,\pm\,2.9$\\
4 & $0.65\,\pm\,0.03$ & $0.78\,\pm\,0.07$ & $44.7\,\pm\,3.0$ & $0.65\,\pm\,0.06$ & $0.76\,\pm\,0.09$ & $45.5\,\pm\,3.1$\\
\addlinespace[3pt]
5 & $0.69\,\pm\,0.04$ & $0.85\,\pm\,0.06$ & $40.0\,\pm\,3.2$ & $0.67\,\pm\,0.05$ & $0.83\,\pm\,0.08$ & $41.6\,\pm\,2.8$\\
6 & $0.67\,\pm\,0.06$ & $0.88\,\pm\,0.10$ & $39.2\,\pm\,4.7$ & $0.70\,\pm\,0.10$ & $0.87\,\pm\,0.10$ & $39.3\,\pm\,4.7$\\
7 & $0.69\,\pm\,0.05$ & $0.96\,\pm\,0.10$ & $34.4\,\pm\,4.8$ & $0.73\,\pm\,0.07$ & $0.91\,\pm\,0.10$ & $36.4\,\pm\,4.5$\\
8 & $0.71\,\pm\,0.05$ & $1.01\,\pm\,0.06$ & $31.9\,\pm\,3.3$ & $0.72\,\pm\,0.05$ & $0.95\,\pm\,0.09$ & $34.7\,\pm\,3.6$\\
9 & $0.72\,\pm\,0.09$ & $1.01\,\pm\,0.07$ & $31.5\,\pm\,4.8$ & $0.74\,\pm\,0.04$ & $1.02\,\pm\,0.09$ & $30.4\,\pm\,3.5$\\
\addlinespace[3pt]
10 & $0.73\,\pm\,0.07$ & $1.02\,\pm\,0.11$ & $30.3\,\pm\,7.2$ & $0.78\,\pm\,0.07$ & $0.99\,\pm\,0.07$ & $31.0\,\pm\,3.2$\\
11 & $0.77\,\pm\,0.11$ & $1.02\,\pm\,0.08$ & $29.5\,\pm\,6.0$ & $0.83\,\pm\,0.07$ & $1.06\,\pm\,0.08$ & $26.7\,\pm\,3.6$\\
12 & $0.80\,\pm\,0.09$ & $1.06\,\pm\,0.08$ & $27.0\,\pm\,4.4$ & $0.87\,\pm\,0.10$ & $1.05\,\pm\,0.05$ & $26.2\,\pm\,1.5$\\
13 & $0.81\,\pm\,0.09$ & $1.07\,\pm\,0.09$ & $26.5\,\pm\,5.0$ & $0.90\,\pm\,0.08$ & $1.05\,\pm\,0.08$ & $26.0\,\pm\,4.4$\\
14 & $0.81\,\pm\,0.07$ & $1.08\,\pm\,0.07$ & $25.5\,\pm\,4.1$ & $0.92\,\pm\,0.11$ & $1.07\,\pm\,0.10$ & $24.9\,\pm\,5.7$\\
\addlinespace[3pt]
15 & $0.82\,\pm\,0.04$ & $1.08\,\pm\,0.03$ & $25.4\,\pm\,1.3$ & $0.90\,\pm\,0.07$ & $1.09\,\pm\,0.09$ & $23.5\,\pm\,5.1$\\
16 & $0.85\,\pm\,0.03$ & $1.09\,\pm\,0.07$ & $24.5\,\pm\,3.5$ & $0.92\,\pm\,0.07$ & $1.11\,\pm\,0.09$ & $22.6\,\pm\,4.8$\\
17 & $0.86\,\pm\,0.04$ & $1.11\,\pm\,0.10$ & $23.2\,\pm\,4.5$ & $0.92\,\pm\,0.08$ & $1.14\,\pm\,0.07$ & $20.8\,\pm\,3.4$\\
18 & $0.90\,\pm\,0.06$ & $1.08\,\pm\,0.05$ & $23.9\,\pm\,2.6$ & $0.94\,\pm\,0.10$ & $1.16\,\pm\,0.08$ & $19.8\,\pm\,4.4$\\
19 & $0.90\,\pm\,0.04$ & $1.11\,\pm\,0.11$ & $22.8\,\pm\,5.7$ & $0.96\,\pm\,0.11$ & $1.19\,\pm\,0.09$ & $18.0\,\pm\,4.8$\\
\addlinespace[3pt]
20 & $0.93\,\pm\,0.07$ & $1.16\,\pm\,0.10$ & $19.5\,\pm\,5.0$ & $0.97\,\pm\,0.12$ & $1.25\,\pm\,0.11$ & $15.5\,\pm\,5.4$\\
21 & $0.92\,\pm\,0.10$ & $1.16\,\pm\,0.09$ & $20.2\,\pm\,4.3$ & $0.96\,\pm\,0.12$ & $1.27\,\pm\,0.15$ & $14.3\,\pm\,7.7$\\
22 & $0.94\,\pm\,0.09$ & $1.20\,\pm\,0.10$ & $17.8\,\pm\,5.1$ & $0.97\,\pm\,0.08$ & $1.28\,\pm\,0.15$ & $13.6\,\pm\,7.0$\\
23 & $0.94\,\pm\,0.11$ & $1.20\,\pm\,0.14$ & $17.7\,\pm\,7.3$ & $0.95\,\pm\,0.11$ & $1.27\,\pm\,0.17$ & $14.7\,\pm\,8.3$\\
24 & $0.97\,\pm\,0.10$ & $1.23\,\pm\,0.14$ & $16.7\,\pm\,6.2$ & $0.97\,\pm\,0.09$ & $1.27\,\pm\,0.11$ & $13.9\,\pm\,5.8$\\
\addlinespace[3pt]
25 & $0.98\,\pm\,0.12$ & $1.24\,\pm\,0.13$ & $15.9\,\pm\,6.1$ & $0.98\,\pm\,0.08$ & $1.25\,\pm\,0.11$ & $14.5\,\pm\,5.3$\\
26 & $0.94\,\pm\,0.15$ & $1.28\,\pm\,0.10$ & $14.9\,\pm\,4.4$ & $0.95\,\pm\,0.10$ & $1.25\,\pm\,0.15$ & $15.2\,\pm\,7.9$\\
27 & $0.96\,\pm\,0.12$ & $1.31\,\pm\,0.07$ & $12.7\,\pm\,2.6$ & $0.95\,\pm\,0.11$ & $1.25\,\pm\,0.18$ & $15.3\,\pm\,9.4$\\
28 & $0.99\,\pm\,0.12$ & $1.28\,\pm\,0.09$ & $13.8\,\pm\,4.0$ & $0.95\,\pm\,0.11$ & $1.27\,\pm\,0.18$ & $14.3\,\pm\,9.9$\\
29 & $1.00\,\pm\,0.15$ & $1.30\,\pm\,0.07$ & $13.3\,\pm\,3.6$ & $0.94\,\pm\,0.11$ & $1.26\,\pm\,0.19$ & $14.8\,\pm\,10.4$\\
\addlinespace[3pt]
30 & $0.99\,\pm\,0.15$ & $1.34\,\pm\,0.11$ & $11.7\,\pm\,5.4$ & $0.92\,\pm\,0.11$ & $1.28\,\pm\,0.18$ & $15.1\,\pm\,7.5$\\
31 & $1.00\,\pm\,0.17$ & $1.34\,\pm\,0.17$ & $13.7\,\pm\,5.8$ & $0.89\,\pm\,0.10$ & $1.30\,\pm\,0.19$ & $14.9\,\pm\,7.5$\\
32 & $1.03\,\pm\,0.16$ & $1.36\,\pm\,0.19$ & $13.4\,\pm\,6.7$ & $0.88\,\pm\,0.09$ & $1.31\,\pm\,0.17$ & $14.1\,\pm\,7.2$\\
33 & $1.04\,\pm\,0.16$ & $1.41\,\pm\,0.19$ & $12.1\,\pm\,7.2$ & $0.89\,\pm\,0.08$ & $1.35\,\pm\,0.14$ & $12.0\,\pm\,5.5$\\
34 & $1.03\,\pm\,0.14$ & $1.38\,\pm\,0.21$ & $11.9\,\pm\,8.4$ & $0.91\,\pm\,0.08$ & $1.34\,\pm\,0.18$ & $12.8\,\pm\,6.4$\\
\addlinespace[3pt]
35 & $1.00\,\pm\,0.13$ & $1.35\,\pm\,0.18$ & $12.1\,\pm\,7.5$ & $0.91\,\pm\,0.09$ & $1.34\,\pm\,0.14$ & $12.4\,\pm\,5.5$\\
36 & $1.02\,\pm\,0.14$ & $1.35\,\pm\,0.14$ & $12.2\,\pm\,5.0$ & $0.92\,\pm\,0.05$ & $1.34\,\pm\,0.18$ & $13.0\,\pm\,4.1$\\
37 & $0.99\,\pm\,0.15$ & $1.35\,\pm\,0.11$ & $12.1\,\pm\,3.3$ & $0.93\,\pm\,0.06$ & $1.36\,\pm\,0.19$ & $12.9\,\pm\,3.2$\\
38 & $0.97\,\pm\,0.16$ & $1.34\,\pm\,0.12$ & $13.0\,\pm\,4.5$ & $0.94\,\pm\,0.09$ & $1.36\,\pm\,0.17$ & $12.6\,\pm\,3.4$\\
39 & $0.98\,\pm\,0.16$ & $1.34\,\pm\,0.13$ & $12.9\,\pm\,4.8$ & $0.94\,\pm\,0.10$ & $1.39\,\pm\,0.18$ & $12.5\,\pm\,1.7$\\
\addlinespace[3pt]
40 & $0.98\,\pm\,0.14$ & $1.31\,\pm\,0.17$ & $13.4\,\pm\,7.5$ & $0.94\,\pm\,0.11$ & $1.41\,\pm\,0.19$ & $12.4\,\pm\,1.9$\\
41 & $0.96\,\pm\,0.12$ & $1.33\,\pm\,0.12$ & $11.4\,\pm\,6.7$ & $0.95\,\pm\,0.13$ & $1.37\,\pm\,0.22$ & $14.5\,\pm\,2.1$\\
42 & $0.96\,\pm\,0.15$ & $1.33\,\pm\,0.09$ & $12.1\,\pm\,5.3$ & $0.94\,\pm\,0.12$ & $1.40\,\pm\,0.24$ & $14.5\,\pm\,3.5$\\
43 & $0.97\,\pm\,0.13$ & $1.34\,\pm\,0.14$ & $11.9\,\pm\,7.0$ & $0.96\,\pm\,0.12$ & $1.41\,\pm\,0.24$ & $13.8\,\pm\,4.9$\\
44 & $0.97\,\pm\,0.16$ & $1.36\,\pm\,0.10$ & $11.0\,\pm\,6.1$ & $0.97\,\pm\,0.13$ & $1.46\,\pm\,0.26$ & $13.4\,\pm\,7.0$\\
\addlinespace[3pt]
45 & $0.95\,\pm\,0.16$ & $1.39\,\pm\,0.09$ & $10.1\,\pm\,5.7$ & $0.98\,\pm\,0.11$ & $1.42\,\pm\,0.29$ & $14.8\,\pm\,6.9$\\
\end{longtable}
\endgroup

\end{document}